\PassOptionsToPackage{dvipsnames}{xcolor}
\documentclass[a4paper, authoryear]{arxiv}
\usepackage{mystyle}
\usepackage{url}
\usepackage{booktabs}
\usepackage{nicefrac}
\usepackage[framemethod=tikz]{mdframed}
\ifvenue
  \mdfdefinestyle{paperbox}{linewidth=0.5pt, roundcorner=6pt,
    innertopmargin=2pt, innerbottommargin=2pt, innerleftmargin=2pt,
    innerrightmargin=2pt, skipabove=2pt, skipbelow=2pt}
\else
  \mdfdefinestyle{paperbox}{linewidth=0.5pt, roundcorner=6pt,
    innertopmargin=4pt, innerbottommargin=6pt}
\fi
\usepackage{mathtools}
\mathtoolsset{showonlyrefs}
\usepackage{dsfont}
\usepackage{tikz}
\usetikzlibrary{positioning, arrows.meta, calc, fit, backgrounds}
\usepackage{enumitem}
\usepackage{algorithm}
\usepackage{algorithmic}
\usepackage{subcaption}

\newcommand{\defi}[1]{{\bf #1}}

\newcommand{\BR}{{\normalfont\texttt{BR}}}

\title[Bayesian Learning of Parametric MFGs with Common Noise]{Bayesian Learning of Parametric Mean Field Games with Common Noise}
\author[1]{Gr\'egoire Lambrecht}
\author[2]{Mathieu Lauri\`ere}
\affil[1]{Center for Data Science, New York University and NYU Shanghai. \email{gl3048@nyu.edu}}
\affil[2]{Shanghai Center for Data Science; NYU-ECNU Institute of Mathematical Sciences at NYU Shanghai; NYU Shanghai, Shanghai, People's Republic of China. \email{mathieu.lauriere@nyu.edu}}
\funding{Gr\'egoire Lambrecht is partially supported by NSF Award 1922658.}

\begin{document}
\maketitle
\bibliographystyle{plainnat}

\begin{abstract}
We study the inverse problem of recovering unknown parameters governing the dynamics and rewards of Mean Field Games (MFGs) from partial observations of population trajectories, made stochastic by common noise. The observer seeks to recover a consistent parameter $\theta$ via Bayesian inference, facing the challenge that evaluating the likelihood requires solving a new Nash equilibrium (NE) for each candidate $\theta$. Due to common noise, the mean-field evolution is stochastic and equilibrium policies are population dependent, which makes solving for each parameter computationally expensive. To overcome this, we introduce the \emph{Bayesian-Extended Mean Field Game} (BE-MFG), whose single NE simultaneously encodes equilibrium policies for all parameters, enabling efficient learning of the conditional density of mean-field trajectories. We formalize parameter consistency and identifiability in MFGs with common noise and investigate conditions under which consistent parameter recovery is guaranteed. On three synthetic environments of various complexity, we recover the ground-truth parameter accurately from a small number of partial, stochastic observations, outperforming an adaptation of inverse reinforcement learning for MFGs. A fourth experiment, on public electricity market data, suggests that the approach can also be applied to real-world data.
\end{abstract}

\section{Introduction}

Learning in multi-agent environments presents significant hurdles due to the complex, often non-stationary, strategic interactions among agents. While Deep Reinforcement Learning has demonstrated remarkable results in small-scale games like Go \cite{silver2016mastering, silver2017mastering}, Chess \cite{Silver18AlphaZero}, Poker \cite{heinrich2016deep}, and complex strategy games such as StarCraft \cite{samvelyan2019starcraft} or Stratego \cite{perolat2022mastering}, extending these techniques to massive populations remains challenging. As the population grows, the joint state-action space expands exponentially, rendering traditional Multi-Agent Reinforcement Learning methods computationally infeasible \cite{busoniu2008comprehensive, yang2020overviewmarl, zhangyangbasar2021multiagentrloverview, gronauer2022multi, wong2023deep}.

Mean Field Games (MFGs) \cite{lasry2007mean, huang2006large} address this scaling problem by modeling the interaction between a representative individual and the aggregate distribution of the population. MFGs simplify high-dimensional games into a system of coupled equations describing the optimal behavior of a single agent and the macroscopic evolution of the distribution. 

This framework is highly effective for applications involving large populations in continuous spaces, such as economics \cite{lachapelle2016efficiency, achdou2017income}, finance \cite{cardaliaguet2018mean, yang2017mean}, engineering \cite{djehiche2017mean}, crowd dynamics \cite{lachapelle2011mean, achdou2018mean}, biological swarming \cite{nourian2010synthesisflocking, fornasier2014meanfield}, and large-scale network resource management \cite{hanif2016mean, mao2022mean}. 
In many realistic settings, the population is subjected to aggregate, macroscopic shocks (such as market fluctuations or weather events) which are modeled through common noise \cite{carmona2016mean}. 
In recent years, the numerical resolution of high-dimensional MFGs has been advanced by deep learning and reinforcement learning (RL) techniques; see \cite{lauriere2022learning, hu2024recent} for reviews. In particular, model-free RL algorithms \cite{guo2019learningneurips, subramanian2019reinforcement, elie2020convergencemodelfreemfg, cui2021approximately, angiuli2022unified, yardim2023policy} learn equilibrium policies through direct interaction with an environment implementing the MFG, but do not address the question of building this environment from real data.

\noindent\textbf{Challenges.} Inferring unknown model parameters from observed trajectories is central to real-world applications where the game's specifications are not known a priori. This inverse problem is closely related to Inverse Reinforcement Learning (IRL) and Imitation Learning (IL) for MFGs, which recover reward functions or dynamics from observed equilibrium trajectories \cite{chen2022individualinverseRL,chen2023adversarial,anahtarci2024maximum}, but without common noise or with access to many individual state-action trajectories. The latter is often unrealistic, e.g., due to data protection constraints in epidemiology, economics, or finance. Moreover, when the environment incorporates common noise and the latent parameter $\theta$ is unknown, standard IRL techniques struggle for two reasons: (1) the observer only witnesses stochastic realizations of the mean-field trajectory, necessitating a probabilistic approach to identifiability and inference; (2) the equilibrium policies are in general population-dependent, hence more complex. In short, the main challenges we tackle are:
\begin{mdframed}[style=paperbox, backgroundcolor=red!5, linecolor=red!80]
\begin{enumerate}[leftmargin=*]
    \item Learning MFG parameters in both the \textbf{dynamics} and the \textbf{reward};
    \item Using only \textbf{(state) population-level information}, possibly limited to certain time steps;
    \item Extracting information from \textbf{noisy mean-field trajectories} due to common noise;
    \item Handling \textbf{population-dependent policies} due to mean field stochasticity.
\end{enumerate}
\end{mdframed}

\noindent\textbf{Main Contributions.} In this paper, we address the challenge of learning the latent parameters of an MFG with common noise from observed population trajectories, either state-action or state-only, and potentially truncated in time. Our main contributions, summarized through Fig.~\ref{fig:schema}, are:
\begin{mdframed}[style=paperbox, backgroundcolor=blue!5, linecolor=blue!80]
\begin{enumerate}[leftmargin=*]
    \item We formalize the problem of \textbf{parameter identifiability} in MFGs with common noise, establishing \textbf{theoretical conditions} under which an observer can recover a parameter consistent with the observations, and the true parameter $\theta_0$ under identifiability, using Bayesian estimation, despite the latent nature of the common noise.
    \item We introduce the \textbf{Bayesian-Extended Mean Field Game (BE-MFG)}, a novel augmented game formulation that embeds the entire parametric family of MFGs. By computing a single equilibrium for the BE-MFG, we simultaneously solve for the equilibrium policies across all parameters, bypassing the need to repeatedly solve MFGs during the inference phase.
    \item We propose a \textbf{scalable method} combining a deep Fictitious Play (FP) algorithm with a Feature-wise Linear Modulation (FiLM) architecture and Masked Autoregressive Flows (MAF) to efficiently approximate the posterior distribution over the parameters.
    \item We report \textbf{numerical experiments} on three synthetic environments of increasing complexity, varying parameter dimensionality, state space size, and observation type, in which the ground-truth parameter is recovered and the method is compared against an adaptation of Mean-Field IRL (MF-IRL)~\cite{chen2022individualinverseRL}. A fourth experiment, on public electricity market data, illustrates the \textbf{potential applicability} of the pipeline to real-world data.
\end{enumerate}
\end{mdframed}

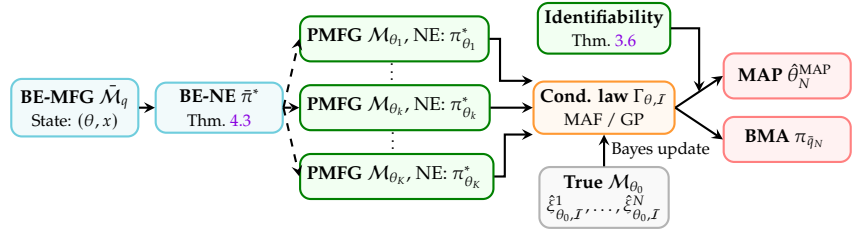
\begin{figure}[h]
\centering
\begin{tikzpicture}[
    scale=0.70, transform shape, 
    node distance = 0.5cm,
    box/.style = {rectangle, rounded corners, draw, thick, align=center, minimum width=2.4cm, minimum height=0.85cm, font=\small},
    arrow/.style = {->, thick, >=stealth},
    dasharrow/.style = {->, thick, >=stealth, dashed},
]
\node[box, draw=Turquoise!60, fill=Turquoise!5] (bemfg) {%
    \textbf{BE-MFG} $\bar{\mathcal{M}}_q$\\
    {\footnotesize State: $(\theta,x)$}
};
\node[box, draw=Turquoise!60, fill=Turquoise!5, right=0.3cm of bemfg] (bene) {%
    \textbf{BE-NE} $\bar\pi^*$\\
    {\footnotesize Thm.~\ref{thm:NE_BE-MFG_induces_NE_PMFG}}
};
\draw[arrow] (bemfg) -- (bene);
\node[box, draw=green!50!black, fill=green!5, right=0.3cm of bene, yshift=1.3cm] (pmfg1) {%
    \textbf{PMFG} $\mathcal{M}_{\theta_1}$, NE: $\pi^*_{\theta_1}$
};
\node[box, draw=green!50!black, fill=green!5, right=0.3cm of bene] (pmfgk) {%
    \textbf{PMFG} $\mathcal{M}_{\theta_k}$, NE: $\pi^*_{\theta_k}$
};
\node[box, draw=green!50!black, fill=green!5, right=0.3cm of bene, yshift=-1.3cm] (pmfgK) {%
    \textbf{PMFG} $\mathcal{M}_{\theta_K}$, NE: $\pi^*_{\theta_K}$
};
\draw[dasharrow] (bene.east) -- (pmfg1.west);
\draw[dasharrow] (bene.east) -- (pmfgk.west);
\draw[dasharrow] (bene.east) -- (pmfgK.west);
\node[font=\footnotesize] at ($(pmfg1)!0.4!(pmfgk) - (0.05,0)$) {$\vdots$};
\node[font=\footnotesize] at ($(pmfgk)!0.4!(pmfgK) - (0.05,0)$) {$\vdots$};
\node[box, draw=orange!80, fill=orange!5, right=0.8cm of pmfgk] (condlaw) {%
    \textbf{Cond. law} $\Gamma_{\theta,\mathcal{I}}$\\
    {\footnotesize MAF / GP}
};
\draw[arrow] (pmfg1.east) -- ++(0.2,0) |- (condlaw.north west);
\draw[arrow] (pmfgk.east) -- (condlaw.west);
\draw[arrow] (pmfgK.east) -- ++(0.2,0) |- (condlaw.south west);
\node[box, draw=green!50!black, fill=green!5, above=0.5cm of condlaw] (identif) {%
    \textbf{Identifiability}\\
    {\footnotesize Thm.~\ref{thm:identifiability}}
};
\node[box, draw=red!50, fill=red!5, right=0.9cm of condlaw, yshift=0.6cm] (map) {%
    \textbf{MAP} $\hat\theta^{\mathrm{MAP}}_N$
};
\node[box, draw=red!50, fill=red!5, right=0.9cm of condlaw, yshift=-0.6cm] (bma) {%
    \textbf{BMA} $\pi_{\bar q_N}$
};
\coordinate (midcondmap) at ($(condlaw.east)!0.5!(map.west)$);
\draw[arrow] (condlaw.east) -- (map.west);
\draw[arrow] (condlaw.east) -- (bma.west);
\draw[arrow, black!50!black] (identif.east) -| (midcondmap);
\node[box, draw=gray!60, fill=gray!5, below=0.6cm of condlaw] (obs) {%
    \textbf{True} $\mathcal{M}_{\theta_0}$\\
    {\footnotesize $\hat\xi^1_{\theta_0,\mathcal{I}},\dots,\hat\xi^N_{\theta_0,\mathcal{I}}$}
};
\draw[arrow] (obs) -- (condlaw) node[midway, right, font=\footnotesize] {Bayes update};
\end{tikzpicture}
\caption{\it Solving the BE-MFG once yields a Nash equilibrium that simultaneously encodes equilibrium policies for all parameters; one can simulate trajectories to learn the conditional density, which, combined with the true observations via Bayes' rule, gives the posterior, and hence the MAP estimator and the BMA policy. 
}
\label{fig:schema}
\end{figure}

\noindent\textbf{Related Work.} Our work relates primarily to IL and IRL in MFGs (\cite{shmaya2025bayesian} considers players performing Bayesian learning during the game, a different question). IL for MFGs, which infers the equilibrium policies used by experts, has been studied in \cite{ramponi2024imitation, lambrecht2026} from a theoretical perspective. On the IRL side, \cite{chen2022individualinverseRL,chen2023adversarial} recover reward functions from observed behaviors. Recent extensions have incorporated Maximum Causal Entropy \cite{anahtarci2024maximum, anahtarci2025maximum}, average reward criteria \cite{alkir2025inverse}, and kernel-based methods \cite{anahtarci2025kernel} to improve robustness and representational capacity. However, these methods typically assume deterministic mean-field at equilibrium (no common noise) and focus solely on the reward function, leaving the transition dynamics fixed and known. Furthermore, they often require computing the forward MFG equilibrium in the inner loop of the inference, which is computationally prohibitive; meta-IRL \cite{chen2024meta} still relies on nested optimization, and model-based RL for MFGs \cite{huang2024model} learns dynamics by exploration rather than from expert trajectories. In contrast, our Bayesian framework handles general parametric uncertainty (encompassing both rewards and dynamics) under common noise. Crucially, by formulating the inference problem through the Bayesian-Extended MFG, our method uniquely bypasses the inner-loop equilibrium computation, solving for all parameterized policies in a single forward phase. 

\noindent\textbf{Organization of the Paper.} In Sec.~\ref{sec:pb-setting}, we formalize the setting of Parametric Mean Field Games (PMFGs). In Sec.~\ref{sec:consistency-identifiability}, we define the concepts of consistency and identifiability. Sec.~\ref{sec:bays} presents the Bayesian learning setting, introduces the BE-MFG, establishes its equivalence to the family of PMFGs and details our algorithmic approach. Finally, Sec.~\ref{sec:expes} demonstrates the effectiveness of our method through numerical experiments.
\onlyin{venue}{\vspace{-0.5em}}
\section{Problem setting} 
\label{sec:pb-setting}
\onlyin{venue}{\vspace{-0.5em}}
We consider a family of MFGs indexed by a parameter $\theta$, and aim to recover the true parameter $\theta_0$, its Nash equilibrium policy, or best responses, from observations of mean-field trajectories.
\onlyin{venue}{\vspace{-0.5em}}
\subsection{Mean field dynamics and Nash equilibria} 

\onlyin{venue}{\vspace{-0.5em}}
\noindent\textbf{Notation.} For any integer $N$, $[N] = \{0, \dots, N\}$ and $[N]^* = [N] \setminus \{0\}$. $\Delta_{\mathcal{S}}$ (resp. $\mathcal{P}(E)$) denotes the set of probability measures over a finite set $\mathcal{S}$ (resp. a measurable space $E$). We otherwise use standard notation, all of which is defined, together with the abbreviations, in App.~\ref{appendix:notation}.

\noindent\textbf{Main components.} 
Let $\cX$ be a finite \defi{state space}, $\cA$ be a finite action space, $\cE^0$ be a Polish space for the \defi{common noise}, $\nu^0\in \cP(\cE^0)$ be a probability measure for the \defi{common noise law}, $H \in \mathbb{N}^*$ be the \defi{time horizon}, and $\rho_0 \in \Delta_{\cX}$ be the \defi{initial probability measure}. Let $d \in \mathbb{N}^*$, $\bar\Theta \subset \mathbb{R}^d$ be a compact \defi{set of parameters}, and $\lambda$ a reference measure on $\mathcal{B}(\bar\Theta)$. Let $P:\bar \Theta \times \cX \times \cA \times \Delta_{\cX} \times \cE^0 \to \Delta_{\cX}$ and $r:\bar \Theta \times \cX \times \cA \times \Delta_{\cX} \to \RR$ be two measurable mappings representing the \defi{transition kernel} and the \defi{reward function}, $r$ being moreover bounded. For any $\theta \in \bar \Theta$, for any $(x, a, \rho, e^0) \in \cX \times \cA \times \Delta_{\cX} \times \cE^0 $ define $P_\theta(x, a, \rho, e^0) = P(\theta, x, a, \rho, e^0)$ and $r_\theta(x, a, \rho) = r(\theta, x, a, \rho)$.  A \defi{Parametric Mean Field Game (PMFG)} associated to the parameter $\theta \in \bar \Theta$ is denoted by a tuple $\cM_\theta = (\cX, \cA, \cE^0, \nu^0, P_\theta, r_\theta, H, \rho_0)$. We call a \defi{family of PMFG} the set $\{\cM_\theta\}_{\theta \in \bar \Theta}$. We define a \defi{strategy} as a measurable mapping from $\cX \times \Delta_{\cX}$ to $\Delta_{\cA}$, and a \defi{policy} as a sequence of $H$ strategies. A \defi{policy} is denoted by $\pi = (\pi_0, \dots, \pi_{H-1})$, the set of all strategies is denoted by $\Pi_0$ and the set of all such policies is denoted by $\Pi$. Elements of $\Delta_{\cX}$ and $(\Delta_{\cX})^H$ are referred to as \defi{mean field} and \defi{mean field sequences} respectively.

\noindent\textbf{Distribution evolution.} Let $\theta \in \bar \Theta$ and consider the MFG $\cM_\theta$. 
For any pair of population distributions $\rho, \rho' \in \Delta_\cX$, any strategy $\pi_0$, we define the \defi{mean-field transition operator} $\Phi_{\theta}: \Delta_\cX \times \Pi_0 \times \Delta_\cX \times \cE^0 \to \Delta_\cX $ as
\begin{equation*}
    \Phi_{\theta}(\rho', \pi_0; \rho, e^0) = \sum_{x \in \cX, a \in \cA} \rho'(x) \pi_0(a \vert x, \rho) P_\theta(\cdot \vert x, a, \rho, e^0).
\end{equation*}
This operator $\Phi_\theta$ represents the evolution of the probability distribution of an agent between time $i$ and $i+1$, given that at time $i$ the agent's state distribution was $\rho'$, the agent used strategy $\pi_0$, the total population was distributed according to $\rho$, and the common noise value is $e^0$.
We consider an independent and identically distributed (i.i.d.) sequence of random variables $\varepsilon^0_1, \dots, \varepsilon^0_{H-1} \in \cE^0$ for the common noise. We denote the common noise trajectory as $\varepsilon^0 = (\varepsilon^0_1, \dots, \varepsilon^0_{H-1})$. For any $\pi \in \Pi$, we define the mean field sequence $\rho^{\pi,\varepsilon^0}_\theta = (\rho^{\pi,\varepsilon^0}_{\theta,0}, \rho^{\pi, \varepsilon^0}_{\theta,1}, \dots, \rho^{\pi,\varepsilon^0}_{\theta,H-1})$ by $\rho^{\pi,\varepsilon^0}_{\theta,0} = \rho_0$ and
\begin{equation}
\label{eq:mf-evol-pi}
        \rho^{\pi,\varepsilon^0}_{\theta, i+1} = \Phi_{\theta}(\rho^{\pi,\varepsilon^0}_{\theta, i}, \pi_i \,; \, \rho^{\pi,\varepsilon^0}_{\theta, i}, \varepsilon^0_{i+1}), \qquad i \in [H-2].
\end{equation}
This sequence represents the evolution of the population distribution when all agents follow policy $\pi$ and are affected by the common noise trajectory $\varepsilon^0_{1:H-1}$. For any pair of policies $(\pi, \pi') \in \Pi^2$, we define the mean field sequence $\rho^{\pi, \pi', \varepsilon^0}_\theta= (\rho^{\pi, \pi', \varepsilon^0}_{\theta,0}, \rho^{\pi, \pi', \varepsilon^0}_{\theta,1}, \dots, \rho^{\pi, \pi', \varepsilon^0}_{\theta,H-1})$ by: $\rho^{\pi, \pi', \varepsilon^0}_{\theta,0} = \rho_0$,
\begin{equation}
\label{eq:mf-evol-pi-pi'}
        \rho^{\pi, \pi', \varepsilon^0}_{\theta, i+1} = \Phi_{\theta}(\rho^{\pi, \pi', \varepsilon^0}_{\theta, i}, \pi'_i; \rho^{\pi, \varepsilon^0}_{\theta, i}, \varepsilon^0_{i+1}), \qquad i \in [H-2].
\end{equation}
This sequence represents the evolution of the state distribution of a \emph{single agent} following policy $\pi'$ when all other agents follow policy $\pi$ and are affected by the common noise trajectory $\varepsilon^0$. %

\noindent\textbf{Best response and Nash equilibrium.} 
When the population of agents follows the policy $\pi$, a single agent wishes to find a \defi{best response} to $\pi$, i.e., a policy $\pi'$ maximizing the \defi{value function}:
\onlyin{venue}{\vspace{-5pt}}
\begin{equation}
    \label{eq:V_theta}
    V_\theta(\pi'; \pi) = \mathbb{E} \Big[ \sum_{i = 0}^{H-1} r_\theta(X_i, \alpha_i, \rho^{\pi, \varepsilon^0}_{\theta,i}) \Big],
\end{equation}
where the process is defined by $X_0 \sim \rho_0$, $X_{i+1} \sim P_\theta(\cdot \vert X_i, \alpha_i, \rho^\pi_{\theta,i}, \varepsilon^0_{i+1})$, and $\alpha_i \sim \pi'(\cdot \vert X_i, \rho^\pi_{\theta,i})$. %
Note that under these conditions, the law of the state $X_i$ (conditioned on the common noise) is $\rho^{\pi, \pi', \varepsilon^0}_{\theta,i}$.  
We denote the law of state-action pair $(X_i, \alpha_i)$ by $\mu^{\pi, \pi',\varepsilon^0}_{ \theta, i}\in \Delta_{\cX\times \cA}$.
If $\pi = \pi'$, we will simply denote $\mu^{\pi, \varepsilon^0}_{\theta, i} = \mu^{\pi, \pi, \varepsilon^0}_{\theta, i}$.  In the sequel, we will often omit to write explicitly the common noise $\varepsilon^0$, but the sequences $\rho$ and $\mu$ are random variables due to $\varepsilon^0$.

\begin{definition}
    Let $\theta \in \bar \Theta$. For $\pi \in \Pi$, the set of \defi{best responses} to the PMFG $\cM_\theta$ is $\BR_\theta(\pi) = \arg \max_{\pi' \in \Pi} V_\theta(\pi'; \pi)$. $\pi \in \Pi$ is a \defi{(mean field) Nash equilibrium} (NE) for $\cM_\theta$ if $\pi \in \BR_\theta(\pi)$.
\end{definition}
\onlyin{venue}{\vspace{-5pt}}
This is a fixed-point condition: it implies that when the entire population follows policy $\pi$, the resulting mean-field flow $\rho^{\pi}_\theta$ creates a reward and transition structure such that the same policy $\pi$ remains optimal for any individual agent.

\begin{assumption}
    \label{assumption:uniqueness_nash}
    For every $\theta \in \bar \Theta$, if $\pi$ and $\pi'$ are two Nash equilibria to $\cM_\theta$, then $\cL(\mu^{\pi}_\theta) = \cL(\mu^{\pi'}_\theta)$. 
\end{assumption}
\onlyin{venue}{\vspace{-5pt}}
The existence and uniqueness of the NE for finite-state, discrete-time MFGs has been proved under various assumptions; see e.g.~\cite{cui2021approximately} for the finite-horizon setting and~\cite{saldi2018markov,anahtarci2020qregmfg} for the infinite horizon discounted setting. Moreover, our approach only requires uniqueness of the law of the state-action mean field $\cL(\mu_{\theta}^*)$. In the sequel we work under Assumption~\ref{assumption:uniqueness_nash} and we will focus solely on equality of laws. For this reason, we are not interested in a specific policy being a NE, and we will denote by $\rho_\theta^* = \rho_\theta^{\pi_\theta^*,\varepsilon^0}$ (resp. by $\mu_\theta^* = \mu_\theta^{\pi_\theta^*,\varepsilon^0}$) a mean field trajectory (resp. the state-action mean field trajectory) associated to some NE policy $\pi_\theta^*$. Both are random variables due to the common noise.

\noindent\textbf{Observations.} 
We will consider two types of observations: the observation of state-action mean-fields trajectories $\mu$, and the observation of only state mean-fields trajectories $\rho$. To have a common framework and simplify the notation, we will denote the observations by $\xi$ and specify case by case whether it represents $\mu$ or $\rho$. We denote the space of $\xi$ by $\Xi$, referring either to $\Delta_\cX$ or $\Delta_{\cX\times \cA}$. For any $\theta \in \bar \Theta$, we will write $\xi_\theta^*$ a realization of the state or state-action mean-fields trajectory, associated to the unique law of the NE (Assumption~\ref{assumption:uniqueness_nash}). 

In the following, we will consider situations where the observer has only access to some time steps in the mean-fields trajectories. For this purpose, let us introduce the following notation: for any set of accessible time steps from the observer $\cI \subset [H-1]$ and any type of observation $\xi \in \Xi$, let $\xi_{\cI} = \{\xi_{ i}\}_{i \in \cI}$. Recalling that $\xi^*_{\theta, \cI}$ is random due to the common noise, let $\PP_{\theta, \cI}$ denote the probability measure over $\Xi^{\vert \cI\vert}$ of $\xi^*_{\theta, \cI}$. 

\subsection{Objectives} 

Our objective is as follows. Let $\theta_0$ be the "true parameter" unknown to the observer, $\cI \subset [H-1]$ be the subset of accessible snapshots. Assume that an observer sees realizations $\hat \xi_{\theta_0, \cI}^1, \dots , \hat \xi_{\theta_0, \cI}^N \in \Xi^{\vert \cI\vert }$, where for any $n \in [N]^*$, $\hat \xi_{\theta_0, \cI}^n$ has been generated from $\pi_{\theta_0}^*$. The observer tries to recover a set of parameters that is consistent with the observations. We study four questions:
\begin{mdframed}[style=paperbox, backgroundcolor=green!5, linecolor=green!80]
\begin{enumerate}[
  label=\textbf{(Q\arabic*)},
  ref=\textbf{(Q\arabic*)},
  noitemsep,
  topsep=0pt,
  leftmargin=*
]
    \item\label{theory-q1} \textit{How should consistency and identifiability be formalized in the presence of common noise?}
    \item\label{theory-q2} \textit{Can the observer recover a parameter $\theta$ consistent with the observed trajectory?}
    \item\label{theory-q3} \textit{Can the observer recover the NE policy $\pi_{\theta_0}^*$?}
    \item\label{theory-q4} \textit{Can the observer reconstruct a best response against the NE policy $\pi_{\theta_0}^*$?}
\end{enumerate}
\end{mdframed}

We address \ref{theory-q1} in the next section. To answer \ref{theory-q2}-\ref{theory-q4} we propose \textbf{two Bayesian approaches}: i) The \defi{maximum a posteriori (MAP)} approach, ii) The \defi{Bayesian model averaging (BMA)} approach in Sec.~\ref{sec:bays}.
\onlyin{venue}{\vspace{-5pt}}
\section{Consistency and identifiability}
\label{sec:consistency-identifiability}
\onlyin{venue}{\vspace{-5pt}}
Injectivity from the parameter space to the observation space is often not verified. Consequently, we first introduce the notions of consistency and identifiability as necessary conditions for the inverse problem of MFGs to be well-posed.

\noindent\textbf{Parameter consistency.} Suppose an observer witnesses realizations of $\xi^*_{\theta_0, \cI}$ for some unknown $\theta_0 \in \bar \Theta$, without knowledge of the specific common noise realization $\varepsilon^0_{1:H-1}$. The observer may seek a parameter to explain these observations; however, the latent nature of the common noise path forces them to consider equality in law. 
\begin{definition}
\label{def:consistency}
Let $\theta, \theta' \in \bar \Theta$. $\theta$ and $\theta'$ are \defi{consistent with respect to $\cI$ for the observation $\xi$} if and only if $\PP_{\theta, \cI} = \PP_{\theta', \cI}$. With a slight abuse of terminology, we say that $\theta$ is \defi{consistent with respect to $\xi^*_{\theta', \cI}$} if and only if $\theta$ and $\theta'$ are consistent with respect to $\cI$ for the observation $\xi$.
\end{definition}
This notion of consistency, which refers to the equality of the laws of the observations, should not be confused with the statistical consistency of an estimator.

\begin{remark}
    Even if the dependency has not been made explicit in the last definition, the consistency is dependent on the type of observations considered ($\xi = \rho$ or $\xi = \mu$). 
\end{remark}

Def.~\ref{def:consistency} provides a natural framework for identifiability. We seek to address the following fundamental question: \textit{Does recovering a parameter that is consistent with respect to $\xi$ allow for the recovery of the original NE $\pi_{\theta_0}^*$?} We give a natural situation where recovering a consistent parameter is equivalent to finding a NE, and we provide the proof in App.~\ref{appendix:identifiability}.

\begin{theorem}
    \label{thm:consistency}
    Consider the observation $\xi = \mu$ with $\cI = [H-1]$. Suppose that Assumption~\ref{assumption:uniqueness_nash} holds, that $P_\theta = P$ for every $\theta \in \bar\Theta$, and that there is a measurable mapping $\Lambda$, not depending on $\theta$, such that $\varepsilon^0 = \Lambda(\mu^*_{\theta, \cI})$ $\PP$-a.s. for every $\theta \in \bar\Theta$. Then, for any $\theta, \theta' \in \bar \Theta$, if $\PP_{\theta, \cI} = \PP_{\theta', \cI}$, then any NE $\pi^*_\theta$ for $\cM_\theta$ is a NE for $\cM_{\theta'}$.
\end{theorem}

\begin{remark}
    \label{remark:shared-state}
    The common noise is recoverable from the observation in particular when it is carried as a \defi{shared state component}. This is the case, e.g., of the Beach Bar process of \cite{perrin2020fictitious}, in which the bar is alternatively open and closed , the status of the bar being shared by all the agents, and of the electricity model of App.~\ref{appendix:energy}. Without common noise, the condition is trivially satisfied.
\end{remark}

\noindent\textbf{Parameter identifiability.}  Another natural question is: \textit{When does consistency imply equality?} Consider a family of PMFGs where consistency is equivalent to equality. Then, the problem of recovering the equilibrium is equivalent to the problem of recovering a parameter consistent with some observations. We highlight that these problems have not been considered in previous work in the context of MFGs. We introduce the definition of \defi{Identifiable Family of PMFG}, prove a result to characterize such models, and provide an example in App.~\ref{appendix:identifiability}.

\begin{definition}
    \label{def:identifiability}
    The family of PMFGs $\{\cM_\theta\}_{\theta \in \bar \Theta}$ is \defi{identifiable with respect to $\cI$ for the observation $\xi$} if for any parameters $\theta, \theta' \in \bar \Theta$, their consistency w.r.t. $\cI$ for the observation $\xi$ implies their equality. 
\end{definition}

To discuss when Definition~\ref{def:identifiability} is satisfied, we extend the mapping $\BR$ to laws of partial mean-field trajectories.
We will use the following notation: $\fP_{\theta, \cI} =  \{ \PP_{\cI}'\in \cP(\Delta^{\vert \cI\vert }_\cX) \mid \, \exists \pi \in \Pi : \cL(\rho_{\theta, \cI}^{\pi}) = \PP_\cI'\}$ which is the set of all state distributions that can be generated by some policy.
For any $\theta \in \bar \Theta$, for any $\PP_{\cI} \in \fP_{\theta, \cI}$, we define 
$\BB \RR_\theta (\PP_\cI) = \{\cL(\rho^{\pi, \pi'}_{\theta, \cI}) \mid \, \pi'\in \BR_\theta(\pi), \, \pi \in \Pi : \cL(\rho_{\theta, \cI}^{\pi}) = \PP_\cI \}.$
Notice that the output of this mapping is not a set of policies, but a set of distributions of mean field trajectories. The definition takes the union over all policies $\pi$ inducing the law $\PP_\cI$, and all the policies $\pi' \in \BR_\theta(\pi)$ so that $\BB\RR_\theta(\PP_\cI)$ is well-defined. Theorem~\ref{thm:identifiability} gives what we believe is the first result on identifiability for inverse problem in discrete-time MFGs with common noise. The proof is provided in App.~\ref{appendix:identifiability}.

\begin{theorem}
    \label{thm:identifiability}
    If for any $\theta \neq \theta'\in \bar \Theta$, for any $\PP_{\cI} \in \fP_{\theta, \cI} \cap \fP_{\theta', \cI} $,  we have 
    $
    \BB\RR_\theta(\PP_\cI)\cap \BB\RR_{\theta'}(\PP_\cI) = \emptyset, 
    $
    then $\{\cM_\theta\}_{\theta \in \bar \Theta}$ is identifiable. 
\end{theorem}

\begin{corollary}
    \label{corollary:fP-cap-fP'=empty}
    If for any $\theta \neq \theta'\in \bar\Theta$,  $\fP_{\theta, \cI} \cap \fP_{\theta', \cI} = \emptyset$, then $\{\cM_\theta\}_{\theta \in \bar \Theta}$ is identifiable.
\end{corollary}
\onlyin{venue}{\vspace{-15pt}}

\section{Bayesian setting}
\label{sec:bays}
\onlyin{venue}{\vspace{-7pt}}

\subsection{Maximum a posteriori and Bayesian model averaging estimators}
To use the Bayesian approach, the observer wishes to consider a prior $q$ over $\bar \Theta$, such that $q(\theta)>0$ for each $\theta \in \bar \Theta$. Let $\cI \subset [H-1]$ be the time steps accessible to the observer. Recall that $\PP_{\theta, \cI}$ denotes the probability measure over $\Xi^{\vert \cI\vert}$ of $\xi^*_{\theta, \cI}$. Let $\Theta$ be a $\bar\Theta$-valued random variable with law $q$, and let $\PP_{\cI} = \int_{\bar \Theta } \PP_{\theta, \cI}q(\theta)\lambda(d\theta)$, which is the probability measure of $\xi^*_{\Theta, \cI}$. To define a likelihood, we need $\PP_{\theta, \cI} \ll \PP_{\cI}$, which is ensured by the following assumption.
\begin{assumption}
    \label{assumption:likelihood}
    One of the following conditions holds:
    \textbf{(C1)} $\bar\Theta$ is at most countable and $\lambda$ is the counting measure;
    \textbf{(C2)} for any $A \in \cB(\Xi^{\vert \cI\vert})$, the mapping $\theta \mapsto \PP_{\theta, \cI}(A)$ is continuous on $\bar\Theta$, and $\lambda(U) > 0$ for any non-empty open subset $U$ of $\bar\Theta$.
\end{assumption}
\begin{lemma}
    \label{lemma:likelihood}
    Under Assumption~\ref{assumption:likelihood}, for every $\theta \in \bar\Theta$, $\PP_{\theta, \cI} \ll \PP_{\cI}$, and the Radon-Nikodym derivative $\Gamma_{\theta,\cI} = d\PP_{\theta, \cI}/d\PP_{\cI}$ exists.
\end{lemma}
The proof is provided in App.~\ref{appendix:likelihood}.
Suppose that the observer has access to $N\in \NN^*$ independent observations $\hat \xi_{\theta_0, \cI}^1, \dots , \hat \xi_{\theta_0, \cI}^N $ of the possibly truncated mean field trajectory at equilibria of $\cM_{\theta_0}$, where $\theta_0\in \bar \Theta$ is the true parameter. If the observer has full knowledge of $\Gamma_{\theta, \cI}$ for any $\theta \in \bar \Theta$, then they can compute the posterior associated to $\Theta$ using Bayes' rule:
\onlyin{venue}{\vspace{-5pt}}
\begin{equation}
    \label{eq:barqN}
    \bar q_N(\theta ) = q(\theta)  \frac{ \prod_{n = 1}^N \Gamma_{\theta, \cI}(\hat \xi_{\theta_0, \cI}^n)}{\Gamma_{\cI}^N(\hat \xi_{\theta_0, \cI}^1, \dots, \hat \xi_{\theta_0, \cI}^N)}.
\end{equation}
where $\Gamma_{\cI}^N(\xi_{\cI}^{*,1},\dots, \xi_{\cI}^{*,N} ) = \int_{\bar \Theta} q(\theta)\prod_{n = 1}^N \Gamma_{\theta, \cI}(\xi_{\cI}^{*,n}) \lambda(d\theta) $. The observer can compute the \defi{MAP estimator}
$
    \hat \theta^{\mathrm{MAP}}_N = \arg \max_{\theta \in \bar \Theta}   \bar q_N(\theta).
$
Then, they can use this recovered parameter to compute the associated NE of $\cM_{\hat \theta^{\mathrm{MAP}}_N}$ and sample new mean-field trajectories. But they could also use the posterior distribution 
to construct an averaged policy $\pi_{\bar q_N}$ that would correspond to a \defi{BMA method},
$$
\pi_{\bar q_N}(a \vert x, \rho ) = \EE_{\Theta \sim  \bar q_N }\left[\pi_\theta^* (a \vert x, \rho)\right], \quad \forall(x, a, \rho) \in \cX\times \cA\times \Delta_\cX.
$$
The cornerstone of this method is to be able to compute $\Gamma_{\theta, \cI}$ for any $\theta \in \bar \Theta$. A naive approach would be to compute a NE on each node of a grid over $\bar \Theta$. But computing a single NE is computationally expensive \cite{perrin2020fictitious,perolat2022scaling,perrin2022generalization}, and this approach would not work: we show quantitatively in App.~\ref{appendix:grid-search} that already for a one-dimensional parameter, a single BE-MFG solve is cheaper than a grid of dedicated solves and recovers the parameter more accurately, and that a grid search becomes out of reach in higher dimension. To address this issue, we propose to compute a NE for an augmented MFG.

\subsection{Bayesian-extended mean field games}
\label{sec:BEMFG}
We now introduce the concept of BE-MFG, which can be interpreted as a game in which a parameter $\Theta \sim q$ is sampled once at time $0$ and is common to all agents. This parameter remains fixed throughout the game and is observed by all agents. Conditionally on $\Theta = \theta$, the game reduces exactly to the PMFG $\cM_\theta$. 
Solving for a NE in the BE-MFG will allow us to recover, simultaneously for all $\theta$, the corresponding equilibrium policies.
\onlyin{venue}{\vspace{-6pt}}

\subsubsection{Definitions}

Let us start by introducing the augmented dynamic kernel $\bar P$ and reward $\bar r$. For any $(\theta, x, a, \bar \rho, e^0,\theta ', x') \in \bar \Theta\times \cX \times \cA \times \cP(\bar \Theta \times \cX) \times \cE^0 \times \bar \Theta\times \cX$, we define
\begin{equation}
    \label{eq:barP-barr}
    \bar P(\theta',x'\mid\theta,x,a,\bar\rho,e^0) = \delta_{\theta=\theta'}P_\theta\!\left(x'\mid x,a,\textstyle\int_{\bar\Theta}\bar\rho(d\theta'',\cdot),e^0\right), \; \bar r(\theta,x,a,\bar\rho) = r_\theta\!\left(x,a,\textstyle\int_{\bar\Theta}\bar\rho(d\theta',\cdot)\right).
\end{equation}
Let $q$ be a probability distribution on $(\bar\Theta, \mathcal{B}(\bar\Theta))$. The \defi{Bayesian-Extended Mean Field Game (BE-MFG)} associated to the family of PMFGs $\{\cM_\theta\}_{\theta \in \bar \Theta}$ and the prior $q$ is the tuple $\bar{\cM}_q = (\bar\Theta \times \mathcal{X}, \mathcal{A}, \mathcal{E}^0, \nu^0, \bar{P}, \bar{r}, H, Q^{\rho_0})$, where $Q^{\rho_0} = q \circ (\delta_{\cdot} \otimes \rho_0)^{-1}$ is the pushforward of $q$ under the map $\theta \mapsto \delta_\theta \otimes \rho_0$, which defines a probability measure on $\mathcal{P}(\bar\Theta \times \cX)$. We define a \defi{Bayesian strategy} as a measurable mapping from $\bar \Theta \times \cX \times \cP(\bar \Theta \times \cX) $ to $\Delta_{\cA}$. The set of Bayesian strategies is denoted $\bar \Pi_0$. We define a  \defi{Bayesian policy} as a sequence of $H$ Bayesian strategies. A \defi{Bayesian policy} is denoted by $\bar \pi = (\bar \pi_0, \dots, \bar \pi_{H-1})$, and the set of all such policies is denoted by $\bar \Pi$. 

The distribution evolution, value function, best responses and NE of the BE-MFG are defined exactly as for the PMFG (Sec.~\ref{sec:pb-setting}), replacing the state space $\mathcal{X}$ by the augmented state space $\bar\Theta\times\mathcal{X}$, the mappings $P_\theta$, $r_\theta$ by $\bar P$,$\bar r$, and the policy class $\Pi$ by $\bar\Pi$. A NE of $\bar{\cM}_q$ is called a \defi{Bayesian-Extended Nash equilibrium (BE-NE)}. Full details are provided in Appendix~\ref{appendix:be-mfg-definitions}. We now proceed to the link between the original PMFG and the BE-MFG. 
\onlyin{venue}{\vspace{-5pt}}

\subsubsection{Link between the models}

 Since $\Theta$ is common to all agents, the evolution of the system conditional on $\Theta = \theta$ coincides with that of $\mathcal{M}_\theta$. Any Bayesian policy $\bar\pi$ thus induces, for each $\theta$, a policy $\pi_\theta$ in $\mathcal{M}_\theta$, and conversely any family $\{\pi_\theta\}_{\theta\in\bar\Theta}$ is lifted to a Bayesian policy. More precisely, the mean-field flow satisfies $\bar\rho^{\bar\pi,\varepsilon^0}(\Theta) = \delta_\Theta \otimes \rho^{\pi_\Theta,\varepsilon^0}_\Theta$ $\mathbb{P}$-a.s., and the value functions satisfy $\bar V_q(\bar\pi';\bar\pi) = \mathbb{E}_{\Theta\sim q}[V_\Theta(\pi'_\Theta;\pi_\Theta)]$, where $\bar\rho^{\bar\pi,\varepsilon^0}$ and $\bar V_q$ are the mean-field flow and value function of $\bar\cM_q$ (App.~\ref{appendix:be-mfg-definitions}). These two identities, established in Lemmas~\ref{lemma:bar_rho^pi(Theta)=d_Theta_rho^pi}--\ref{lemma:barV=EV_family} in App.~\ref{appendix:be-mfg-link}, are the key tools for establishing the main result of this section:

\begin{theorem}
\label{thm:NE_BE-MFG_induces_NE_PMFG}
    Let $\bar\pi^* \in \bar\Pi$ be a NE of $\bar{\mathcal{M}}_q$. Define, for any $\theta \in \bar\Theta$ and $(i, x, a, \rho) \in [H-1] \times \mathcal{X} \times \mathcal{A} \times \Delta_\mathcal{X}$,
    $
    \pi^{**}_{\theta, i}(a \mid x, \rho) = \bar\pi^*_i(a \mid \theta, x, \delta_\theta \otimes \rho).
    $
    Then $\pi^{**}_\Theta$ is a NE of $\mathcal{M}_\Theta$, $q$-a.s.
\end{theorem}
\onlyin{venue}{\vspace{-5pt}}

Under condition (C1) of Assumption~\ref{assumption:likelihood}
Theorem~\ref{thm:NE_BE-MFG_induces_NE_PMFG} holds for every $\theta \in \bar\Theta$.

Combined with Assumption~\ref{assumption:uniqueness_nash}, this result is very useful: Given access to $\bar\pi^*$, one obtains for any $\theta \in \bar\Theta$ a sampler of $\xi^*_{\theta,\mathcal{I}}$: it suffices to simulate the policy $\pi^{**}_\theta$ defined above under a drawn common noise $\varepsilon^0$. Indeed, Assumption~\ref{assumption:uniqueness_nash} guarantees that $\mathcal{L}(\xi^{\pi^{**}_\theta}_{\theta,\mathcal{I}}) = \mathcal{L}(\xi^*_{\theta,\mathcal{I}}) = \PP_{\theta,\mathcal{I}}$, so that the simulated trajectories are distributed according to the true conditional law $\Gamma_{\theta,\mathcal{I}}$. Yet $\bar{\mathcal{M}}_q$ is a standard MFG with common noise, and can be solved with the same algorithms as any individual PMFG $\mathcal{M}_\theta$, at a comparable computational cost (see App.~\ref{appendix:be-mfg-link} and App.~\ref{appendix:grid-search}).
\onlyin{venue}{\vspace{-5pt}}

\subsection{Algorithms}
\label{sec:algos}

\noindent\textbf{Learning via BE-MFG Pipeline.}
The methodology consists of three steps (Algorithm~\ref{algo:bays-inf-mfgs}, App.~\ref{appendix:nash_equilibrium}):
\begin{enumerate}
    \item \textbf{Computing the BE-NE.} Select a prior $q$ over $\bar\Theta$ and compute the BE-NE $\bar\pi^*$ of $\bar{\mathcal{M}}_q$, using a FiLM architecture~\cite{perez2018film} (App.~\ref{appendix:film}) and the FP and deep learning pipeline of~\cite{lambrecht2026} (App.~\ref{appendix:nash_equ_general_algo}). By Theorem~\ref{thm:NE_BE-MFG_induces_NE_PMFG}, this single solve grants access to $\{\pi^{**}_\theta\}_{\theta\in\bar\Theta}$.
    \item \textbf{Learning the conditional density.} Train a parametric approximation $\Gamma^{\gamma}_{\theta,\mathcal{I}}$, with trainable parameters $\gamma$, of the conditional density $\Gamma_{\theta,\mathcal{I}}$ by simulating trajectories from $\bar\pi^*$, using a MAF~\cite{papamakarios2017masked} (App.~\ref{appendix:maf}). As a lightweight alternative, a Gaussian Process (GP) inverse regression, trained on the same simulated trajectories, directly approximates the posterior (App.~\ref{appendix:gp}).
    \item \textbf{Bayesian inference and policy deployment.} Given observations $(\hat\xi^n_{\theta_0,\mathcal{I}})_{1\le n\le N}$, compute the posterior $\bar q_N$ via Bayes' rule, derive $\hat\theta^{\mathrm{MAP}}_N$ and $\pi_{\bar q_N}$, and deploy both directly using $\bar\pi^*$, i.e., using $\pi^{**}_\theta$ in place of $\pi^*_\theta$.
\end{enumerate}

\noindent\textbf{Baseline: Mean Field Inverse Reinforcement Learning.} 
As a baseline, we use the method introduced by \cite{chen2022individualinverseRL}. This method, as all existing methods in the literature to the best of our knowledge, requires access to full state-action observations $\xi = \mu$ and is not designed for the common noise setting. We adapted it to our framework following App.~\ref{appendix:mf-irl}. As discussed in Sec.~\ref{sec:bays} and quantified in App.~\ref{appendix:grid-search}, we do not include a grid search over dedicated PMFG solves as a baseline, because its computational cost makes it uncompetitive.
\onlyin{venue}{\vspace{-0.5em}}
\section{Numerical experiments}
\label{sec:expes}
\onlyin{venue}{\vspace{-5pt}}
\subsection{Metrics and observations}

We will evaluate the performance of our methods for each question using the following four metrics. Below, $\hat \pi$ and $\hat \PP_{\cI}$ refer respectively to the policy obtained and the induced probability measure of $\xi$, using any of the method (MAF with MAP, MAF with BMA, etc.). Each experiment was conducted using five seeds. 
\begin{itemize}[noitemsep,topsep=0pt]
    \item \textbf{Law Identifiability \ref{theory-q1}:} Wasserstein $1$ distance, $\mathcal{W}_1(\PP_{\theta_0, \cI}, \hat \PP_{\cI})$.
    \item \textbf{Parameters Inference \ref{theory-q2}:} For each example, we will introduce a specific metric to evaluate the inference. 
    \item \textbf{Equilibrium Quality \ref{theory-q3}:} Exploitability, $\cE_{\theta_0}(\hat \pi) = \sup_{\pi\in \Pi} V_{\theta_0}( \pi; \hat \pi) - V_{\theta_0}(\hat\pi; \hat\pi)$.
    \item \textbf{Performance vs Expert \ref{theory-q4}:} $V_{\theta_0}(\hat \pi; \pi_{\theta_0}^*) - V_{\theta_0}(\pi_{\theta_0}^*; \pi_{\theta_0}^*).$
\end{itemize}

\noindent\textbf{Observations:} We consider for the Bayesian setting $\xi = \rho$, while the MF-IRL needs intrinsically more informative observation $\xi = \mu$. We let $N\in \{1, 10, 100\}$, and consider three possible sets of available time steps: the full trajectory $\cI = [H-1]$, denoted as `\textbf{Full}'; subsampling every fourth step from the terminal time $\cI = \{H-1, H-5, \dots\}$, denoted as `\textbf{Sparse}'; and $\cI = \{H-1\}$ denoted as `\textbf{Terminal}'. 
\onlyin{venue}{\vspace{-0.5em}}

\onlyin{venue}{\vspace{-5pt}}
\subsection{Experiment 1: proof of concept}
\onlyin{venue}{\vspace{-5pt}}

We consider a variant of the \textit{Beach Bar} environment~\cite{perrin2022generalization}. In this setting, the bar is permanently open and the agents are subject to a common noise shifting their positions: at each time step and at each location, a random shift of the agents is activated with probability $\eta \ge 0$, the \defi{common noise intensity}.  In this example, $\theta$ is one-dimensional and appears in the reward; however the observation space is high-dimensional due to the number of states and time steps. As this environment is standard in MFGs, we defer its detailed description to App.~\ref{appendix:beachbar}.

\noindent\textbf{Bayesian setup:} We take the uniform prior $q$ on $\bar\Theta = [0, 2.5]$ and $\theta_0 \in \{0.5, 0.875, 1.25, 1.625, 2.0\}$. We present the results as a function of $\eta$, aggregated over $\theta_0$, for $\xi = \rho$; the full sensitivity analysis in $\eta$ and $N$ is in App.~\ref{appendix:beachbar}.

\noindent\textbf{Parameter Inference Error:} We measure the accuracy of parameter inference via the normalized error $\vert \hat\theta - \theta_0|/2.5$, expressed as a percentage, where $\theta_0$ is the true parameter and $\hat\theta$ is the parameter inferred by the method under consideration (MAF, GP, or MF-IRL).

\begin{figure}[h]
    \centering
    \begin{minipage}{0.76\linewidth}
        \centering
        \includegraphics[width=\linewidth]{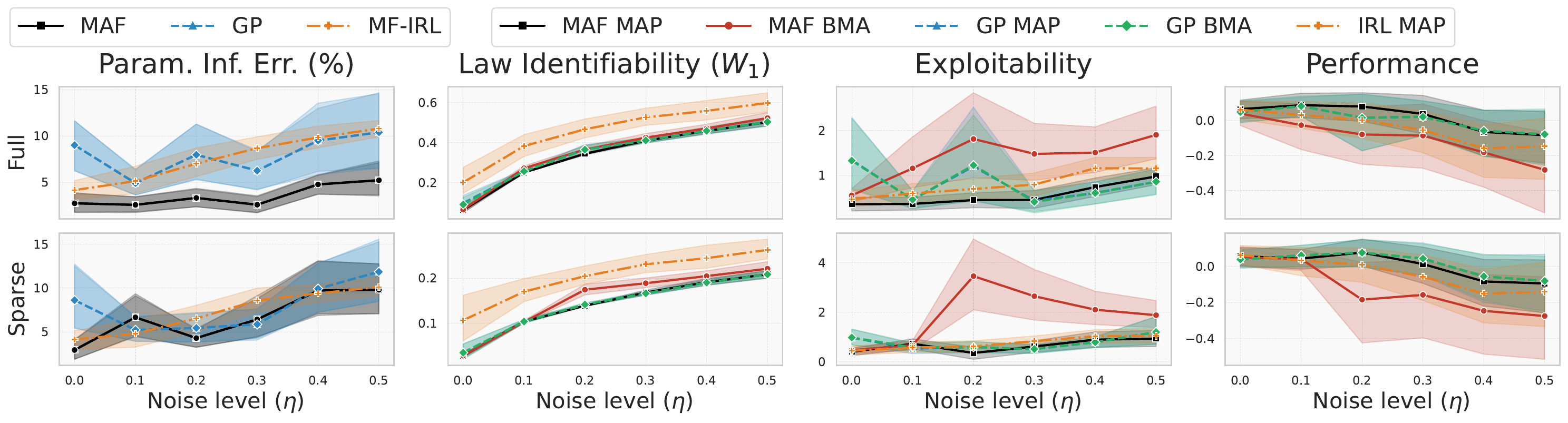}
    \end{minipage}
    \begin{minipage}{0.22\linewidth}
        \centering
        \includegraphics[width=\linewidth]{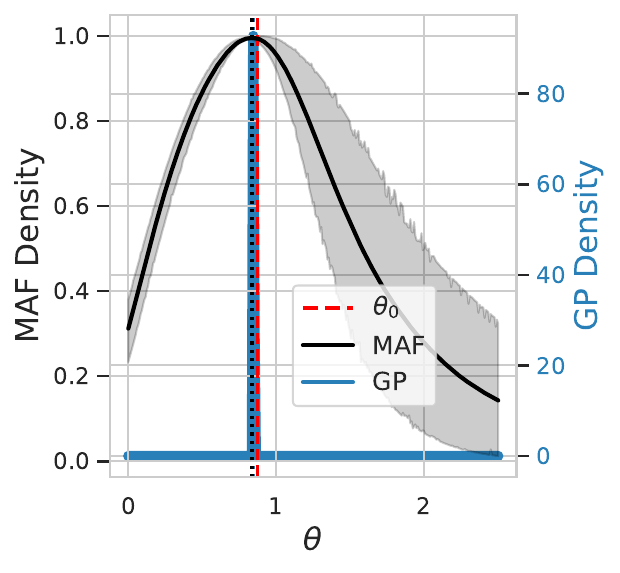}
    \end{minipage}
    \caption{\it Beach Bar, $\xi = \rho$, $N = 100$. \textbf{Left:} metrics as a function of noise level $\eta$ for two available timestep schemes. \textbf{Right:} posterior $\bar q_N$ for $\eta = 0.2$, $\theta_0 = 0.875$, $\mathcal{I} = [H-1]$.}
    \label{fig:beachbar_averaged_xi=rho_N=100}
\end{figure}

\noindent\textbf{Results.} Figure~\ref{fig:beachbar_averaged_xi=rho_N=100} shows that with $N=100$ state-only observations ($\xi=\rho$), the MAF estimator recovers $\theta_0$ with high accuracy even under the \textbf{Sparse} regime, corresponding to only five time steps $i\in\{3,7,11,15,19\}$ out of $H=20$, and outperforms MF-IRL on all metrics, despite MF-IRL having access to the more informative $\xi=\mu$. Across all settings, MAP outperforms BMA for both MAF and GP, suggesting that with $N=100$ the learned conditional densities remain insufficiently sharp for reliable posterior sampling. 
Figure~\ref{fig:beachbar_averaged_xi=rho_N=100} confirms this picture at the posterior level: even for $N = 100$, the likelihood curves remain spread. However, the MAP estimate aligns closely with the true $\theta_0 = 0.875$ even for $N=1$. Overall, these results demonstrate the \emph{reliability} of the proposed method. From sparse, state-only, and potentially partial observations, without any knowledge of agent actions, the observer successfully recovers the parameter, reproduces the law of the observed densities, derives a near-optimal policy, and recovers a NE under $\theta_0$. These findings go significantly beyond existing inverse MFG results, which require either no common noise or direct access to state-action trajectories.
\onlyin{venue}{\vspace{-5pt}}
\subsection{Experiment 2: higher parametric dimension space}
\onlyin{venue}{\vspace{-5pt}}
We now demonstrate our method on a higher-dimensional inference problem, where both the spatial preferences and congestion aversion of agents are unknown.

\noindent\textbf{Model:} Agents move on a $3\times 3$ torus; at each location, a common noise shock of probability $\eta$ replaces the chosen move by a uniformly random one. The reward combines an attraction $\theta^e \in \RR_+^9$ to hotspots and a congestion penalty of intensity $\theta^c \in [0,5]$; full equations are given in App.~\ref{appendix:district}.

\noindent\textbf{Bayesian setup:} We consider $\bar \Theta = \RR^9_+\times [0, 5]$. The observer considers the prior $q$ being the density of $\cN\big(((\frac{1}{9}, \dots,\frac{1}{9}), \frac{5}{2}), 0.01\big)$ while using GP and $\text{Dirichlet}(1, \dots, 1) \otimes \text{Unif}[0, 5]$ while using MAF. We vary the ground truth $\theta_0$ across the three scenarios (Bullseye, Dual Core, Corridor) provided in Figure~\ref{fig:district_traj_eta=0.3} in App.~\ref{appendix:district}, showing three realizations of trajectories with $\eta = 0.3$. We refer to this model as the \textbf{Entertainment District} environment.

\noindent\textbf{Parameter Inference Error:} Since $\theta^e$ and $\theta^c$ have very different effects on the equilibrium, we distinguish the inference error between these two parameters. For $\theta^e$, we use the District Error (L1) $\sum_x |\hat\theta^e(x) - \theta^e_0(x)|$, and for $\theta^c$, we use the normalized error $|\hat\theta^c - \theta^c_0|/5$, expressed as a percentage.

\noindent\textbf{Challenges:} While the general shape of $\theta^e$ is partially recoverable from agent aggregation patterns, finding precise values at both hot and non-hot spots is difficult due to congestion-induced non-trivial behavior. Moreover, the absence of prior knowledge on the congestion level $\theta^c$ makes this a genuinely hard 10-dimensional inference problem.

\begin{figure}[htbp]
    \centering
    \begin{minipage}{0.75\linewidth}
        \centering
        \includegraphics[width=\linewidth]{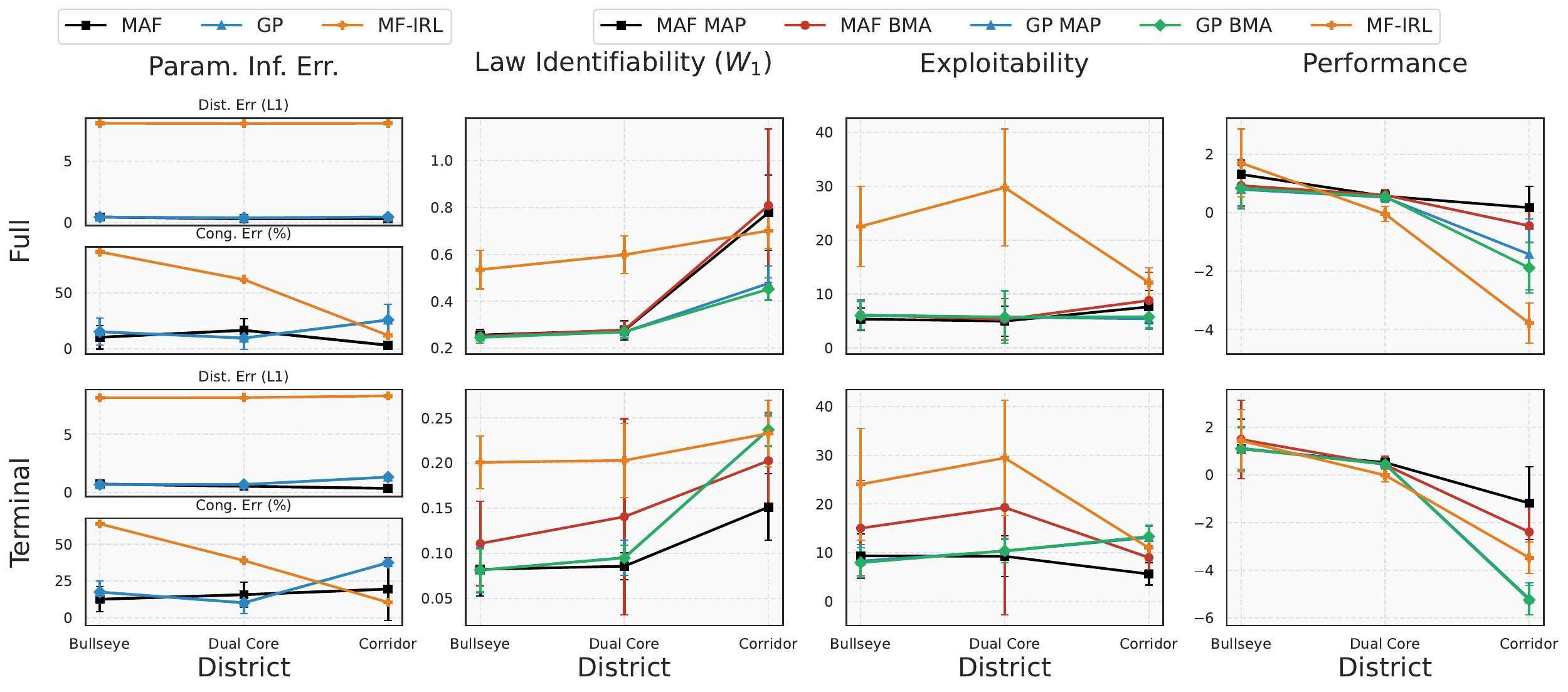}
    \end{minipage}%
    \begin{minipage}{0.23\linewidth}
        \centering
        \includegraphics[width=\linewidth]{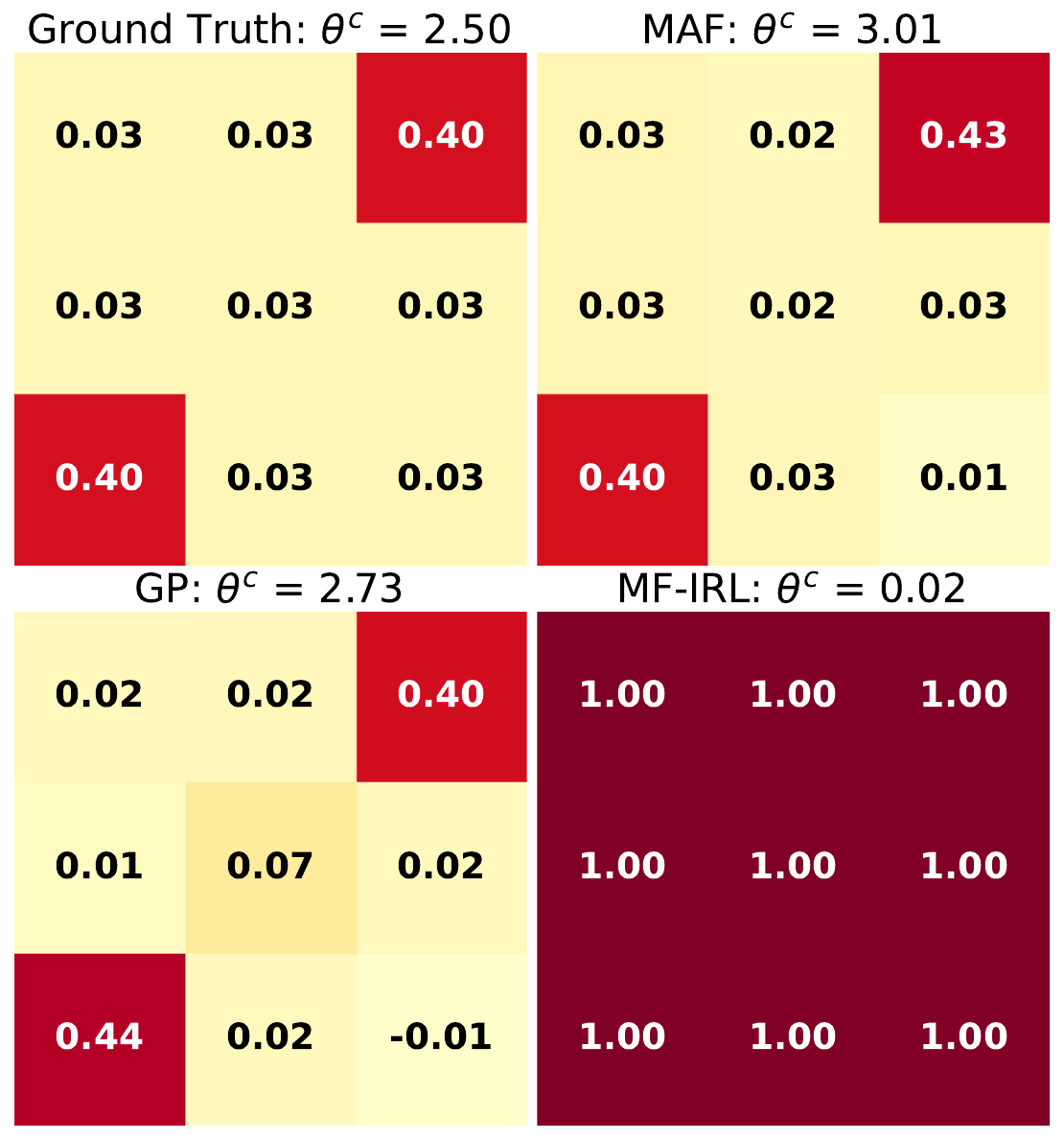}
    \end{minipage}
    \caption{\it District, $\xi = \rho$, $N = 1$. \textbf{Left:} metrics as a function of the district scenario for the three sampling schemes. \textbf{Right:} inferred $\theta_0 = (\theta^e_0, \theta^c_0)$ for the Dual Core scenario, $\eta = 0.5$, $\mathcal{I} = [H-1]$.}
    \label{fig:district_metrics_xi=rho_N=1}
\end{figure}

\noindent\textbf{Results:} Figure~\ref{fig:district_metrics_xi=rho_N=1} reports metrics averaged over $\eta \in \{0.3, 0.5\}$, for each ground truth. The Bayesian approach consistently succeeds across all three scenarios. As shown in Figure~\ref{fig:district_metrics_xi=rho_N=1}, it accurately reconstructs the shape of $\theta^e$ despite the 10-dimensional nature of the problem, and does so from a single observation ($N=1$) with remarkably low variance across seeds. The MF-IRL baseline partially outputs a near-uniform $\theta^e$, inconsistent with the observed trajectories, and fails on all remaining metrics. This is particularly striking given that MF-IRL operates on the more informative $\xi = \mu$ observations.

\onlyin{venue}{\vspace{-5pt}}
\subsection{Experiments 3 \& 4: dynamic inference and application to real data}
\onlyin{venue}{\vspace{-5pt}}
To show that our framework also applies to dynamics inference and to the state-action observation setting $\xi = \mu$, we consider the Susceptible--Infected--Susceptible (SIS) epidemic model~\cite{cui2021approximately} (App.~\ref{sec:exp-sis}). To illustrate that the method can be applied to real data, we consider an MFG of electricity producers, whose common noise is the observed electricity demand, and apply the BE-MFG pipeline to public data of the French electricity market (App.~\ref{appendix:energy}). 

\onlyin{venue}{\vspace{-1em}}
\section{Conclusion}
\label{section:conclusion}
\onlyin{venue}{\vspace{-1em}}
We introduced a Bayesian framework for learning the latent parameters of MFGs with common noise from partial, population trajectories. The central contribution, the BE-MFG, embeds the entire parametric family into a single augmented game whose Nash equilibrium simultaneously encodes equilibrium policies for all parameters, bypassing repeated forward solves. We formalized identifiability under common noise and established consistency guarantees. On three synthetic environments of increasing complexity, we recover the ground-truth parameter accurately and derive near-optimal policies, consistently outperforming MF-IRL despite the latter having access to more informative observations. A fourth experiment, on public electricity market data, suggests that the pipeline can also be run on real-world observations, although the model we use there is deliberately stylized and we make no claim of predictive value.
\textbf{Limitations:} We provided the basic ideas around identifiability but one should look for general sufficient conditions. Moreover, we restricted our study to finite state models; scaling the BE-MFG approach to continuous state-action spaces via function approximation is a natural next step.

\begin{onlyinblock}{arxiv,notes}
\section*{Acknowledgements}

Gr\'egoire Lambrecht is partially supported by NSF Award 1922658.

The authors gratefully acknowledge the NYU Center for Data Science and NYU Shanghai for providing the computational resources, the access to the scientific literature, and the research tools without which this work would not have been possible.
\end{onlyinblock}

\subsection*{AI use statement}

In this work, we used generative AI tools to assist with implementing the code of the experiments, improving the English of the manuscript, analyzing the numerical results, creating the figures, and detecting typographical and formatting issues. The research questions, the model and problem formulation, the theoretical statements and their proofs, the design of the experiments, and the interpretation of the results are our own. We have reviewed all AI-assisted work. We take responsibility for the final content of this work, including text, claims or artifacts produced with the aid of generative AI.

\bibliography{biblio}

\clearpage

\appendix

\section{Notation and abbreviations}
\label{appendix:notation}

\begin{table}[htbp]
\centering
\small
\caption{Abbreviations.}
\label{tab:abbreviations}
\begin{tabular}{@{}ll@{}}
\toprule
\textbf{Abbreviation} & \textbf{Meaning} \\
\midrule
BE-MFG & Bayesian-Extended Mean Field Game (Sec.~\ref{sec:BEMFG}) \\
BE-NE & Bayesian-Extended Nash equilibrium, i.e., a NE of the BE-MFG \\
BMA & Bayesian model averaging \\
FiLM & Feature-wise Linear Modulation \\
FP & Fictitious Play \\
GP & Gaussian Process \\
IL / IRL & Imitation Learning / Inverse Reinforcement Learning \\
MAF & Masked Autoregressive Flow \\
MAP & maximum a posteriori \\
MF-IL & Mean Field Imitation Learning \\
MF-IRL & Mean-Field Inverse Reinforcement Learning \\
MFG & Mean Field Game \\
NE & Nash equilibrium \\
PCA & principal component analysis \\
PMFG & Parametric Mean Field Game \\
RL & Reinforcement Learning \\
SIS & Susceptible--Infected--Susceptible \\
\bottomrule
\end{tabular}
\end{table}

\begin{table}[htbp]
\centering
\small
\caption{Main notation.}
\label{tab:notation}
\begin{tabular}{@{}lp{0.68\linewidth}@{}}
\toprule
\textbf{Symbol} & \textbf{Meaning} \\
\midrule
\multicolumn{2}{@{}l}{\textit{General}} \\
$[N]$, $[N]^*$ & $\{0, \dots, N\}$ and $\{1, \dots, N\}$ \\
$\Delta_{\mathcal{S}}$, $\cP(E)$ & probability measures on a finite set $\mathcal{S}$, on a measurable space $E$ \\
$\cB(U)$ & Borel $\sigma$-algebra on $U \subset \RR^n$, $n \in \NN^*$, induced by the Euclidean topology \\
$\cL(X)$ & law of a random variable $X$ \\
$P \otimes Q$, $P \ll Q$, $\frac{dP}{dQ}$ & product measure, absolute continuity, Radon-Nikodym derivative of two probability measures $P$, $Q$ \\
$\mathcal{W}_1$, $\mathrm{dist}$ & Wasserstein-1 distance on a metric space $(\mathcal{S}, \mathrm{dist})$: $\mathcal{W}_1(P, Q) = \inf_{\nu \in \Pi(P,Q)} \mathbb{E}_{(X,Y)\sim\nu}[\mathrm{dist}(X,Y)]$, where $\Pi(P,Q)$ is the set of couplings with marginals $P$ and $Q$ \\
\midrule
\multicolumn{2}{@{}l}{\textit{Parametric MFGs (Sec.~\ref{sec:pb-setting})}} \\
$\cX$, $\cA$ & finite state and action spaces \\
$\cE^0$, $\nu^0$, $\varepsilon^0$ & common noise space, common noise law, common noise trajectory $(\varepsilon^0_1, \dots, \varepsilon^0_{H-1})$ \\
$\eta$ & common noise intensity (experiments) \\
$H$, $\rho_0$ & time horizon, initial distribution \\
$\bar\Theta \subset \RR^d$, $\lambda$ & parameter set, reference measure on $\bar\Theta$ \\
$P_\theta$, $r_\theta$, $\cM_\theta$ & transition kernel, reward and PMFG associated with $\theta$ \\
$\pi = (\pi_0, \dots, \pi_{H-1})$, $\Pi_0$, $\Pi$ & policy, set of strategies, set of policies \\
$\Phi_\theta$ & mean-field transition operator \\
$\rho^{\pi, \varepsilon^0}_\theta$, $\rho^{\pi, \pi', \varepsilon^0}_\theta$ & population flow under $\pi$; flow of a single agent playing $\pi'$ against a population playing $\pi$ \\
$\mu^{\pi, \pi', \varepsilon^0}_\theta$ & state-action counterpart of $\rho^{\pi, \pi', \varepsilon^0}_\theta$ \\
$V_\theta(\pi'; \pi)$, $\BR_\theta(\pi)$ & value of $\pi'$ against $\pi$, set of best responses to $\pi$ \\
$\pi^*_\theta$, $\rho^*_\theta$, $\mu^*_\theta$ & a NE of $\cM_\theta$ and its (state, state-action) mean-field trajectories \\
$\cE_{\theta}(\pi)$ & exploitability of $\pi$ in $\cM_\theta$ \\
\midrule
\multicolumn{2}{@{}l}{\textit{Observations and identifiability (Sec.~\ref{sec:pb-setting}--\ref{sec:consistency-identifiability})}} \\
$\xi \in \{\rho, \mu\}$, $\Xi$ & type of observation and its space \\
$\cI \subset [H-1]$, $\xi_\cI$ & observed time steps, observation restricted to $\cI$ \\
$\PP_{\theta, \cI}$ & probability measure of $\xi^*_{\theta, \cI}$ \\
$\theta_0$ & true (unknown) parameter \\
$N$, $\hat\xi^n_{\theta_0, \cI}$ & number of observations, $n$-th observation \\
$\fP_{\theta, \cI}$, $\BB\RR_\theta$ & laws generated by some policy, best-response map on laws \\
\midrule
\multicolumn{2}{@{}l}{\textit{Bayesian setting (Sec.~\ref{sec:bays})}} \\
$q$, $\Theta$ & prior on $\bar\Theta$, random parameter with law $q$ \\
$\PP_\cI$, $\Gamma_{\theta, \cI}$ & prior predictive measure, likelihood $d\PP_{\theta, \cI}/d\PP_\cI$ \\
$\bar q_N$ & posterior after $N$ observations \\
$\hat\theta^{\mathrm{MAP}}_N$, $\pi_{\bar q_N}$ & MAP estimator, BMA policy \\
$\bar\cM_q$, $\bar P$, $\bar r$ & BE-MFG, its transition kernel and reward \\
$\bar\Pi_0$, $\bar\Pi$, $\bar V_q$ & Bayesian strategies, Bayesian policies, value function of $\bar\cM_q$ \\
$\bar\pi^*$, $\pi^{**}_\theta$ & a BE-NE, and the policy it induces in $\cM_\theta$ (Thm.~\ref{thm:NE_BE-MFG_induces_NE_PMFG}) \\
$\Gamma^\gamma_{\theta, \cI}$, $\gamma$ & learned approximation of $\Gamma_{\theta, \cI}$, trainable parameters \\
\bottomrule
\end{tabular}
\end{table}

\clearpage

\section{Additional definitions, lemmas and proofs}
\label{appendix:proofs}
\subsection{Identifiability and consistency }
\label{appendix:identifiability}

We first present the proof of Theorem~\ref{thm:consistency}.
\begin{proof}
    By assumption $\PP_{\theta, \cI} = \cL(\mu^{*}_\theta) = \cL(\mu^{*}_{\theta'}) = \PP_{\theta', \cI}$. Since $\rho^*_{\theta, i}$ is the first marginal of $\mu^*_{\theta, i}$, we deduce that $\cL(\rho^{*}_\theta) = \cL(\rho^{*}_{\theta'})$.

    Since the kernel does not depend on the parameter, the mean-field transition operator $\Phi_\theta = \Phi$ of \eqref{eq:mf-evol-pi} does not depend on it either. Consequently, the flows induced by a given policy are the same in $\cM_\theta$ and in $\cM_{\theta'}$; in particular
    \begin{equation}
    \label{eq:same-flow-H1}
        \rho^{\pi^*_\theta, \varepsilon^0}_{\theta'} = \rho^{\pi^*_\theta, \varepsilon^0}_{\theta} = \rho^*_\theta, \qquad \mu^{\pi^*_\theta, \varepsilon^0}_{\theta'} = \mu^{\pi^*_\theta, \varepsilon^0}_{\theta} = \mu^*_\theta, \qquad \PP\text{-a.s.}
    \end{equation}

    \textbf{Step 1: the deviation values coincide.} Let $\pi \in \Pi$. We prove by induction on $i \in [H-1]$ that there is a measurable mapping $G_i$, depending only on $\pi$, such that
    $$
    \rho^{\pi^*_\theta, \pi, \varepsilon^0}_{\theta', i} = G_i\big(\rho^*_{\theta, 0:i-1}, \varepsilon^0_{1:i}\big), \quad \PP\text{-a.s.}
    $$
    For $i = 0$, $\rho^{\pi^*_\theta, \pi, \varepsilon^0}_{\theta', 0} = \rho_0$ is deterministic. Assume the identity holds at rank $i \in [H-2]$. By \eqref{eq:mf-evol-pi-pi'} and \eqref{eq:same-flow-H1},
    $$
    \rho^{\pi^*_\theta, \pi, \varepsilon^0}_{\theta', i+1} = \Phi\big(\rho^{\pi^*_\theta, \pi, \varepsilon^0}_{\theta', i}, \pi_i \,;\, \rho^{*}_{\theta, i}, \varepsilon^0_{i+1}\big) = \Phi\big(G_i(\rho^*_{\theta, 0:i-1}, \varepsilon^0_{1:i}), \pi_i \,;\, \rho^{*}_{\theta, i}, \varepsilon^0_{i+1}\big),
    $$
    which is of the announced form and completes the induction. The same mappings $G_i$ give $\rho^{\pi^*_{\theta'}, \pi, \varepsilon^0}_{\theta', i} = G_i(\rho^*_{\theta', 0:i-1}, \varepsilon^0_{1:i})$. Since the state-action flow of the deviating agent is obtained from its state flow and from $\pi$, there is therefore a measurable mapping $\Psi$, depending only on $\theta'$ and $\pi$, such that
    $$
    V_{\theta'}(\pi; \pi^*_\theta) = \EE\big[\Psi(\rho^*_\theta, \varepsilon^0)\big], \qquad V_{\theta'}(\pi; \pi^*_{\theta'}) = \EE\big[\Psi(\rho^*_{\theta'}, \varepsilon^0)\big].
    $$
    By recoverability of the common noise, $\varepsilon^0 = \Lambda(\mu^*_{\theta})$ $\PP$-a.s., so that $\Psi(\rho^*_\theta, \varepsilon^0) = \tilde\Psi(\mu^*_\theta)$ $\PP$-a.s., where $\tilde\Psi$ maps $\nu$ to $\Psi$ evaluated at the first marginal of $\nu$ and at $\Lambda(\nu)$; the same identity holds for $\theta'$ with the same $\tilde\Psi$. As $\cL(\mu^*_\theta) = \cL(\mu^*_{\theta'})$, we conclude that
    \begin{equation}
    \label{eq:dev-values-equal}
        V_{\theta'}(\pi; \pi^*_\theta) = V_{\theta'}(\pi; \pi^*_{\theta'}), \qquad \forall \pi \in \Pi.
    \end{equation}

    \textbf{Step 2: the equilibrium values coincide.} By \eqref{eq:same-flow-H1},
    $
    V_{\theta'}(\pi^*_\theta; \pi^*_\theta) = \EE\big[\sum_{i, x, a} \mu^{*}_{\theta, i}(x, a) \, r_{\theta'}(x, a, \rho^{*}_{\theta, i})\big],
    $
    which is the image of $\mu^*_\theta$ by a measurable mapping depending only on $\theta'$. The same mapping applied to $\mu^*_{\theta'}$ gives $V_{\theta'}(\pi^*_{\theta'}; \pi^*_{\theta'})$, so that $\cL(\mu^*_\theta) = \cL(\mu^*_{\theta'})$ yields
    \begin{equation}
    \label{eq:eq-values-equal}
        V_{\theta'}(\pi^*_\theta; \pi^*_\theta) = V_{\theta'}(\pi^*_{\theta'}; \pi^*_{\theta'}).
    \end{equation}

    \textbf{Conclusion.} Taking the supremum over $\pi \in \Pi$ in \eqref{eq:dev-values-equal}, using that $\pi^*_{\theta'}$ is a NE of $\cM_{\theta'}$, and then \eqref{eq:eq-values-equal},
    $$
    \sup_{\pi \in \Pi} V_{\theta'}(\pi; \pi^*_\theta) = \sup_{\pi \in \Pi} V_{\theta'}(\pi; \pi^*_{\theta'}) = V_{\theta'}(\pi^*_{\theta'}; \pi^*_{\theta'}) = V_{\theta'}(\pi^*_\theta; \pi^*_\theta),
    $$
    i.e. $\pi^*_\theta \in \BR_{\theta'}(\pi^*_\theta)$, which means that $\pi_\theta^*$ is a NE for $\cM_{\theta'}$.
\end{proof}

We now pass to the proof of Theorem~\ref{thm:identifiability}.
\begin{proof}
    Suppose that $\{\cM_\theta\}_{\theta \in \bar \Theta}$ is not identifiable. Therefore, there exists two distinct parameters $\theta, \theta'\in \bar \Theta$, such that $\PP_{\theta, \cI} = \PP_{\theta', \cI}$. In the two situations, $\xi = \rho$ and $\xi = \mu$, this leads to $\cL(\rho^*_{\theta, \cI}) = \cL(\rho^*_{\theta', \cI})$. By definition $\cL(\rho^*_{\theta, \cI}) \in \BB\RR_{\theta}(\cL(\rho^*_{\theta, \cI}))$ and $\cL(\rho^*_{\theta, \cI}) \in \BB\RR_{\theta'}(\cL(\rho^*_{\theta, \cI}))$. Therefore $\BB\RR_\theta(\PP_\cI)\cap \BB\RR_{\theta'}(\PP_\cI) \ne \emptyset$. That concludes the proof. 
\end{proof}

We now give an example of identifiable family of PMFGs, satisfying the condition of Corollary~\ref{corollary:fP-cap-fP'=empty}.

\begin{example}
    \label{example:fP-cap-fP'=empty}
    Let $\cX  = \{ 0, \dots, 6\}$, $\cA = \{-1, 0, 1\}$, $\bar \Theta = \{0, 1\}$, $\cE^0 = \{-2, -1, 0, 1, 2\}$, $\rho_0 = \delta_3$ and $H > 1$. Let for any $(x, x', a, \rho, e^0) \in \cX\times \cX\times \cA \times \Delta_\cX\times \cE^0$, $P_0(x'\vert x, a, \rho, e^0) = \delta_{\{x' = x + a + e^0 \text{ mod } 7\}}$ and $P_1(x'\vert x, a, \rho, e^0) = \delta_{\{x' = x + a \text{ mod } 7\}}$. Suppose that $\nu^0(0) + \nu^0(1) + \nu^0(-1) \ne 1$. The condition in Corollary~\ref{corollary:fP-cap-fP'=empty} is satisfied for any $\cI \subset [H-1]$ such that $1\in \cI$.
\end{example}

Finally, we prove the statement in Example~\ref{example:fP-cap-fP'=empty}.
\begin{proof}
    For any $\pi \in \Pi$, and $\PP(\rho^\pi_{\theta = 0, 1}(2) + \rho^\pi_{\theta = 0, 1}(3) + \rho^\pi_{\theta = 0, 1}(4) \ne 1)>0$ because of the presence of the common noise. Yet, $\PP(\rho^\pi_{\theta = 1, 1}(2) + \rho^\pi_{\theta = 1, 1}(3) + \rho^\pi_{\theta = 1, 1}(4) \ne 1) = 0$. That is enough to conclude that Cor.\ref{corollary:fP-cap-fP'=empty} is verified.
\end{proof}

\subsection{Existence of the likelihood}
\label{appendix:likelihood}

We prove Lemma~\ref{lemma:likelihood}.
\begin{proof}
    Let $A \in \cB(\Xi^{\vert \cI\vert})$ such that $\PP_{\cI}(A) = 0$, i.e., $\int_{\bar \Theta} \PP_{\theta, \cI}(A) q(\theta) \lambda(d\theta) = 0$. Since $q(\theta) > 0$ for every $\theta \in \bar\Theta$, we get $\PP_{\theta, \cI}(A) = 0$ for $\lambda$-almost every $\theta \in \bar\Theta$.

    Under (C1), $\lambda$ is the counting measure on the countable set $\bar\Theta$, so the integral is the sum $\sum_{\theta \in \bar\Theta} \PP_{\theta, \cI}(A) q(\theta)$, whose terms are non-negative. Hence $\PP_{\theta, \cI}(A) = 0$ for every $\theta \in \bar\Theta$.

    Under (C2), the set $U_A = \{\theta \in \bar\Theta : \PP_{\theta, \cI}(A) > 0\}$ is open in $\bar\Theta$, since $\theta \mapsto \PP_{\theta, \cI}(A)$ is continuous, and $\lambda(U_A) = 0$. Since $\lambda$ charges every non-empty open subset of $\bar\Theta$, $U_A = \emptyset$, i.e., $\PP_{\theta, \cI}(A) = 0$ for every $\theta \in \bar\Theta$.

    In both cases, $\PP_{\cI}(A) = 0$ implies $\PP_{\theta, \cI}(A) = 0$ for every $\theta \in \bar\Theta$, i.e., $\PP_{\theta, \cI} \ll \PP_{\cI}$. Since these are finite measures, the Radon-Nikodym theorem gives the existence of $\Gamma_{\theta, \cI} = d\PP_{\theta, \cI}/d\PP_{\cI}$.
\end{proof}

\subsection{BE-MFG: full definitions}
\label{appendix:be-mfg-definitions}

We provide the full definitions for the BE-MFG $\bar{\mathcal{M}}_q$, following the same structure as Sec.~\ref{sec:pb-setting} with augmented state space $\bar\Theta\times\mathcal{X}$.

\noindent\textbf{Distribution evolution.} 
For any pair of population distributions $\bar \rho, \bar \rho' \in \cP(\bar \Theta \times \cX)$, any strategy $\bar \pi_0\in \bar \Pi_0$, we define the mean-field transition operator $\bar \Phi(\bar \rho', \bar \pi_0; \bar \rho, e^0)$ as
$$
    \bar \Phi(\bar \rho', \bar \pi_0; \bar \rho, e^0)(\theta ', x') = \int_{\theta \in \bar \Theta}\sum_{x \in \cX, a \in \cA} \bar \rho'(d\theta , x)\bar  \pi_0(a \vert\theta,  x, \bar \rho) \bar P(\theta ', x' \vert \theta, x, a, \bar \rho, e^0).
$$

Let $\Theta$ be a $\bar \Theta$-valued extra source of randomness for all the paper, with density $q$. The evolution of the population distribution when all agents follow policy $\bar \pi$ and are affected by the common noise trajectory $\varepsilon^0_{1:H-1}$ is defined by $\bar \rho^{ \bar \pi,\varepsilon^0}(\Theta) = (\bar \rho^{\bar \pi,\varepsilon^0}_{0}(\Theta) , \bar \rho^{\bar \pi,\varepsilon^0}_{1}(\Theta), \dots, \bar \rho^{\bar \pi,\varepsilon^0}_{H-1}(\Theta))$ where $\bar \rho^{\bar \pi,\varepsilon^0}_{0}(\Theta) = \delta_\Theta \otimes \rho_0$ and for $i \in [H-2]$,
\begin{equation}
\label{eq:b-mf-evol-pi}
        \bar \rho^{\bar \pi,\varepsilon^0}_{i+1}(\Theta) = \bar \Phi(\bar \rho^{\bar \pi,\varepsilon^0}_{i}(\Theta), \bar \pi_i \,; \, \bar \rho^{\bar \pi,\varepsilon^0}_{ i}(\Theta), \varepsilon^0_{i+1}).
\end{equation}
The evolution of the state distribution of a \emph{single agent} following policy $\bar \pi'$ when all other agents follow policy $\pi$ is defined by $\bar \rho^{\bar \pi, \bar \pi',\varepsilon^0}(\Theta) = (\bar \rho^{\bar \pi, \bar \pi',\varepsilon^0}_{0}(\Theta), \bar \rho^{\bar \pi, \bar \pi',\varepsilon^0}_{1}(\Theta), \dots, \bar \rho^{\bar \pi, \bar \pi',\varepsilon^0}_{H-1}(\Theta))$ where $\bar \rho^{\bar \pi, \bar \pi',\varepsilon^0}_{0}(\Theta) = \delta_\Theta \otimes \rho_0$ and for $i \in [H-2]$,
\begin{equation}
        \label{eq:b-mf-evol-pi-pi'}
        \bar \rho^{\bar \pi, \bar \pi',\varepsilon^0}_{i+1}(\Theta) = \bar \Phi(\bar \rho^{\bar \pi, \bar \pi',\varepsilon^0}_{i}(\Theta), \bar  \pi'_i; \bar \rho^{\bar \pi,\varepsilon^0}_{ i}(\Theta), \varepsilon^0_{i+1}).
\end{equation}

\noindent\textbf{Value Function of BE-MFG.} 
When the population of agents follows the policy $\bar \pi$, and a single agent follows $\bar \pi'$, its associated \textbf{Bayesian-Extended value function} is:
\begin{equation}
    \label{eq:barVq}
    \bar V_q(\bar \pi'; \bar \pi) = \mathbb{E} \left[ \sum_{i = 0}^{H-1} \bar r(\Theta_i, X_i, \alpha_i, \bar \rho^{\bar \pi, \varepsilon^0}_{i}(\Theta)) \right] ,
\end{equation}
where the process is defined by $\Theta_0 = \Theta \sim q$, $X_0 \sim\rho_0$, $(\Theta_{i+1}, X_{i+1}) \sim \bar P(\cdot \vert \Theta_i, X_i, \alpha_i, \bar \rho^{\bar \pi, \varepsilon^0}_{i}(\Theta), \varepsilon^0_{i+1})$, and $\alpha_i \sim \bar \pi'(\cdot \vert \Theta_i, X_i, \bar{\rho}^{\bar\pi, \varepsilon^0}_{i}(\Theta))$. Note that under these conditions, the law of the state $(\Theta_i, X_i)$ (conditioned on the common noise) is $\bar \rho^{\bar \pi,\bar  \pi', \varepsilon^0}_{i}(\Theta)$.  We define the Bayesian best responses and the NE in the exact same way that for the original PMFGs.

\begin{definition}
    For $\pi \in \bar\Pi$, the set of \defi{Bayesian best responses} to the PMFG $\bar \cM_q$ is defined by $\bar \BR_q(\bar \pi) = \arg \max_{\bar \pi' \in \bar \Pi} \bar V_q(\bar \pi'; \bar \pi)$. A policy $\bar\pi \in \bar\Pi$ is a \defi{(mean field) Bayesian-Extended Nash equilibrium} (BE-NE) for $\bar \cM_q$ if $\bar \pi \in \bar \BR_q(\bar \pi)$.
\end{definition}

\subsection{BE-MFG and PMFG: link between the models }
\label{appendix:be-mfg-link}

This appendix collects the technical lemmas used in the proofs of the main results concerning the BE-MFG framework. The first two lemmas establish that the population flow in the BE-MFG reduces, type by type, to the population flow in the corresponding PMFG. Specifically, Lemma~\ref{lemma:bar_rho^pi(Theta)=d_Theta_rho^pi} shows this reduction when the BE-MFG policy is constructed from a single MFG policy via a point-mass embedding of the parameter, while Lemma~\ref{lemma:bar_rho^pi(Theta)=d_Theta_rho^pi_family} establishes the same conclusion when the BE-MFG policy is constructed from a family of PMFG policies. The next two lemmas, Lemmas~\ref{lemma:bar_rho^pi_pi'(Theta)=d_Theta_rho^pi_pi'} and~\ref{lemma:bar_rho^pi_pi'(Theta)=d_Theta_rho^pi_pi'_family}, extend these reduction results to the setting where a deviating agent follows a policy different from the population. These four lemmas together ensure that the BE-MFG and the family of PMFGs are consistent in terms of population dynamics, regardless of whether one agent deviates. Building on these results, Lemmas~\ref{lemma:barV=EV} and~\ref{lemma:barV=EV_family} show that the value function in the BE-MFG decomposes as the expectation over parameters of the corresponding value functions in the PMFGs, again under the two respective policy constructions. This decomposition is the key identity linking Nash equilibria of the BE-MFG to Nash equilibria of the parametric MFGs, and it is what the subsequent proof of Theorem~\ref{thm:NE_BE-MFG_induces_NE_PMFG} exploits directly.

\begin{lemma}
    \label{lemma:bar_rho^pi(Theta)=d_Theta_rho^pi}
    Let $\pi \in \Pi$ and $\bar \pi\in \bar\Pi$ be such that for any $\theta \in \bar \Theta$ and $(i,x, a, \rho) \in [H-1]\times \mathcal{X} \times \mathcal{A} \times \Delta_\mathcal{X}$,
    $$
    \pi_{\theta, i}(a \mid x, \rho) = \bar \pi_i (a \mid \theta, x, \delta_\theta \otimes \rho).
    $$
    Then 
    $$
    \bar \rho^{\bar \pi, \varepsilon^0}(\Theta) = \delta_{\Theta} \otimes \rho^{\pi_\Theta,  \varepsilon^0}_\Theta , \quad \mathbb{P}\text{-a.s.},
    $$
    where $\bar \rho^{\bar \pi, \varepsilon^0}(\Theta)$ is defined as in \eqref{eq:b-mf-evol-pi} and $\rho^{\pi_\Theta, \varepsilon^0}_\Theta$ is defined as in \eqref{eq:mf-evol-pi}.
\end{lemma}

\begin{proof}
    Let $(\Theta_i, X_i, \bar \alpha_i)_{0\leq i \leq H-1}$ be the state-action trajectory of a representative agent following policy $\bar \pi \in \bar \Pi$ when the population follows $\bar \pi$, in the BE-MFG. This trajectory is defined as in Equation~\eqref{eq:barVq}, and the marginal law of $(\Theta_i, X_i)$ at each time step $i \in [H-1]$ is $\bar \rho^{\bar \pi, \varepsilon^0}_i$. Let $(X_i', \alpha_i')_{0\leq i \leq H-1}$ be the state-action trajectory of a representative agent following policy $\pi \in \Pi$ when the population follows $\pi$, in the PMFG $\mathcal{M}_{\Theta}$, where $\pi$ is defined as in the lemma statement. This trajectory is defined as in Equation~\eqref{eq:V_theta}, and the marginal law of $X_i'$ at each time step $i \in [H-1]$ is $\rho^{\pi_\Theta, \varepsilon^0}_{\Theta, i}$.

    We prove by induction on $i \in [H-1]$ that, conditionally on $\varepsilon^0$,
    $$
    \bar \rho^{\bar \pi, \varepsilon^0}_i(\Theta) = \delta_{\Theta} \otimes \rho^{\pi_\Theta, \varepsilon^0}_{\Theta, i}, \quad \mathbb{P}\text{-a.s.}
    $$

    \textbf{Base case.} By definition, $\bar \rho^{\bar \pi, \varepsilon^0}_0(\Theta) = \delta_{\Theta} \otimes \rho_0$ and $\rho^{\pi_\Theta, \varepsilon^0}_{\Theta, 0} = \rho_0$, so the relation holds at $i = 0$.

    \textbf{Inductive step.} Suppose that $\bar \rho^{\bar \pi, \varepsilon^0}_i(\Theta) = \delta_{\Theta} \otimes \rho^{\pi_\Theta, \varepsilon^0}_{\Theta, i}$ holds $\mathbb{P}$-a.s.\ for some $i \in [H-2]$. By definition of $\pi$,
    $$
    \alpha_i' \sim \pi_{\Theta, i}(\cdot \mid X_i', \rho^{\pi_\Theta, \varepsilon^0}_{\Theta, i}) = \bar \pi_i(\cdot \mid \Theta, X_i', \delta_\Theta \otimes \rho^{\pi_\Theta, \varepsilon^0}_{\Theta, i}).
    $$
    By the induction hypothesis, $\bar \rho^{\bar \pi, \varepsilon^0}_i(\Theta) = \delta_{\Theta} \otimes \rho^{\pi_\Theta, \varepsilon^0}_{\Theta, i}$, so
    $$
    \bar \alpha_i \sim \bar \pi_i(\cdot \mid \Theta_i, X_i, \bar \rho^{\bar \pi, \varepsilon^0}_i(\Theta)) = \bar \pi_i(\cdot \mid \Theta, X_i, \delta_\Theta \otimes \rho^{\pi_\Theta, \varepsilon^0}_{\Theta, i}).
    $$
    Hence $\mathcal{L}(X_i, \bar\alpha_i \mid \varepsilon^0) = \mathcal{L}(X_i', \alpha_i' \mid \varepsilon^0)$. Now, by Equation~\eqref{eq:barP-barr},
    $$
    (\Theta_{i+1}, X_{i+1}) \sim \bar P(\cdot \mid \Theta_i, X_i, \bar\alpha_i, \bar \rho^{\bar \pi, \varepsilon^0}_i(\Theta)) = \delta_{\Theta_i} \otimes P_\Theta(\cdot \mid X_i, \bar\alpha_i, \rho^{\pi_\Theta, \varepsilon^0}_{\Theta, i}),
    $$
    and $X_{i+1}' \sim P_\Theta(\cdot \mid X_i', \alpha_i', \rho^{\pi_\Theta, \varepsilon^0}_{\Theta, i})$. Since $\Theta_i = \Theta$ $\mathbb{P}$-a.s.\ by the induction hypothesis, we conclude that
    $$
    \bar \rho^{\bar \pi, \varepsilon^0}_{i+1}(\Theta) = \delta_{\Theta} \otimes \rho^{\pi_\Theta, \varepsilon^0}_{\Theta, i+1}, \quad \mathbb{P}\text{-a.s.},
    $$
    which completes the induction and the proof.
\end{proof}

\begin{lemma}
\label{lemma:bar_rho^pi(Theta)=d_Theta_rho^pi_family}
    Let $\{\pi_\theta\}_{\theta \in \bar\Theta} \subset \Pi$ be a family of policies, and define $\bar \pi \in \bar\Pi$ by
    $$
    \bar \pi_i(a \mid \theta, x, \bar \rho) = \pi_{\theta, i}\!\left(a \;\middle|\; x,\, \int_{\bar\Theta} \bar \rho(d\theta', \cdot)\, \right), \quad (i, \theta, x, a, \bar\rho) \in [H-1] \times \bar\Theta \times \mathcal{X} \times \mathcal{A} \times \Delta_{\bar\Theta \times \mathcal{X}}.
    $$
    Then
    $$
    \bar \rho^{\bar \pi, \varepsilon^0}(\Theta) = \delta_{\Theta} \otimes \rho^{\pi_\Theta, \varepsilon^0}_\Theta, \quad \mathbb{P}\text{-a.s.},
    $$
    where $\bar \rho^{\bar \pi, \varepsilon^0}$ is defined as in \eqref{eq:b-mf-evol-pi} and $\rho^{\pi_\Theta, \varepsilon^0}_\Theta$ is defined as in \eqref{eq:mf-evol-pi}.
\end{lemma}

\begin{proof}
    We prove by induction on $i \in [H-1]$ that, conditionally on $\varepsilon^0$,
    $$
    \bar \rho^{\bar \pi, \varepsilon^0}_i(\Theta) = \delta_{\Theta} \otimes \rho^{\pi_\Theta, \varepsilon^0}_{\Theta, i}, \quad \mathbb{P}\text{-a.s.}
    $$

    \textbf{Base case.} By definition, $\bar \rho^{\bar \pi, \varepsilon^0}_0(\Theta) = \delta_\Theta \otimes \rho_0$ and $\rho^{\pi_\Theta, \varepsilon^0}_{\Theta, 0} = \rho_0$, so the relation holds at $i = 0$.

    \textbf{Inductive step.} Suppose that $\bar \rho^{\bar \pi, \varepsilon^0}_i(\Theta) = \delta_{\Theta} \otimes \rho^{\pi_\Theta, \varepsilon^0}_{\Theta, i}$ holds $\mathbb{P}$-a.s.\ for some $i \in [H-2]$. By the induction hypothesis,
    $$
    \int_{\bar\Theta} \bar \rho^{\bar \pi, \varepsilon^0}_i(\Theta)(d\theta', \cdot) = \int_{\bar\Theta} \delta_\Theta(d\theta')\, \rho^{\pi_\Theta, \varepsilon^0}_{\Theta, i}(\cdot)= \rho^{\pi_\Theta, \varepsilon^0}_{\Theta, i}.
    $$
    Therefore, for any agent with type $\Theta$ and state $x$, the action distribution under $\bar \pi$ satisfies
    $$
    \bar \pi_i(\cdot \mid \Theta, x, \bar \rho^{\bar \pi, \varepsilon^0}_i(\Theta)) = \pi_{\Theta, i}\!\left(\cdot \;\middle|\; x,\, \int_{\bar\Theta} \bar \rho^{\bar \pi, \varepsilon^0}_i(\Theta)(d\theta', \cdot)\right) = \pi_{\Theta, i}(\cdot \mid x, \rho^{\pi_\Theta, \varepsilon^0}_{\Theta, i}).
    $$
    Now, by Equation~\eqref{eq:barP-barr} and the induction hypothesis,
    $$
    (\Theta_{i+1}, X_{i+1}) \sim \bar P(\cdot \mid \Theta_i, X_i, \bar\alpha_i, \bar \rho^{\bar \pi, \varepsilon^0}_i(\Theta)) = \delta_{\Theta_i} \otimes P_\Theta(\cdot \mid X_i, \bar\alpha_i, \rho^{\pi_\Theta, \varepsilon^0}_{\Theta, i}),
    $$
    where $\Theta_i = \Theta$ $\mathbb{P}$-a.s.\ by the induction hypothesis. Since the action and transition distributions coincide with those of the PMFG $\mathcal{M}_\Theta$ under policy $\pi_\Theta$, we conclude that
    $$
    \bar \rho^{\bar \pi, \varepsilon^0}_{i+1}(\Theta) = \delta_{\Theta} \otimes \rho^{\pi_\Theta, \varepsilon^0}_{\Theta, i+1}, \quad \mathbb{P}\text{-a.s.},
    $$
    which completes the induction and the proof.
\end{proof}

\begin{lemma}
\label{lemma:bar_rho^pi_pi'(Theta)=d_Theta_rho^pi_pi'}
    Let $\pi, \pi' \in \Pi$ and $\bar \pi, \bar \pi' \in \bar\Pi$ be such that for any $\theta \in \bar \Theta$ and $(i, x, a, \rho) \in [H-1] \times \mathcal{X} \times \mathcal{A} \times \Delta_\mathcal{X}$,
    $$
    \pi_{\theta, i}(a \mid x, \rho) = \bar \pi_i (a \mid \theta, x, \delta_\theta \otimes \rho)
    \quad \text{and} \quad
    \pi_{\theta, i}'(a \mid x, \rho) = \bar \pi_i' (a \mid \theta, x, \delta_\theta \otimes \rho).
    $$
    Then
    $$
    \bar \rho^{\bar \pi, \bar \pi', \varepsilon^0}(\Theta) = \delta_{\Theta} \otimes \rho^{\pi_\Theta, \pi'_\Theta, \varepsilon^0}_\Theta , \quad \mathbb{P}\text{-a.s.},
    $$
    where $\bar \rho^{\bar \pi, \bar \pi', \varepsilon^0}(\Theta)$ is defined as in \eqref{eq:b-mf-evol-pi-pi'} and $\rho^{\pi_\Theta, \pi'_\Theta, \varepsilon^0}_\Theta$ is defined as in \eqref{eq:mf-evol-pi-pi'}.
\end{lemma}

\begin{proof}
    The proof follows the same induction argument as in the proof of Lemma~\ref{lemma:bar_rho^pi(Theta)=d_Theta_rho^pi}. Lemma~\ref{lemma:bar_rho^pi(Theta)=d_Theta_rho^pi} is required to identify the population flow $\bar \rho^{\bar \pi, \varepsilon^0}(\Theta) = \delta_\Theta \otimes \rho^{\pi_\Theta, \varepsilon^0}_\Theta$ at each step of the induction.
\end{proof}

\begin{lemma}
\label{lemma:bar_rho^pi_pi'(Theta)=d_Theta_rho^pi_pi'_family}
    Let $\{\pi_\theta\}_{\theta \in \bar\Theta}, \{\pi'_\theta\}_{\theta \in \bar\Theta} \subset \Pi$ be two families of policies, and define $\bar \pi, \bar\pi' \in \bar\Pi$ by
    $$
    \bar \pi_i(a \mid \theta, x, \bar \rho) = \pi_{\theta, i}\!\left(a \;\middle|\; x,\, \int_{\bar\Theta} \bar \rho(d\theta', \cdot)\, \right), \quad \bar \pi_i'(a \mid \theta, x, \bar \rho) = \pi_{\theta, i}'\!\left(a \;\middle|\; x,\, \int_{\bar\Theta} \bar \rho(d\theta', \cdot)\, \right).
    $$
    Then
    $$
    \bar \rho^{\bar \pi, \bar\pi', \varepsilon^0}(\Theta) = \delta_{\Theta} \otimes \rho^{\pi_\Theta, \pi'_\Theta, \varepsilon^0}_\Theta, \quad \mathbb{P}\text{-a.s.},
    $$
    where $\bar \rho^{\bar \pi, \bar\pi', \varepsilon^0}$ is defined as in \eqref{eq:b-mf-evol-pi-pi'} and $\rho^{\pi_\Theta, \pi'_\Theta, \varepsilon^0}_\Theta$ is defined as in \eqref{eq:mf-evol-pi-pi'}.
\end{lemma}

\begin{proof}
    The proof follows the same induction argument as in Lemma~\ref{lemma:bar_rho^pi(Theta)=d_Theta_rho^pi_family}, which is also required to identify the population flow $\bar\rho^{\bar\pi, \varepsilon^0}(\Theta) = \delta_\Theta \otimes \rho^{\pi_\Theta, \varepsilon^0}_\Theta$ at each step of the induction.
\end{proof}

\begin{lemma}
\label{lemma:barV=EV}
    Let $\bar \pi \in \bar \Pi$ and define, for any $\theta \in \bar\Theta$ and $(i, x, a, \rho) \in [H-1] \times \mathcal{X} \times \mathcal{A} \times \Delta_\mathcal{X}$,
    $$
    \pi_{\theta, i}(a \mid x, \rho) = \bar \pi_i(a \mid \theta, x, \delta_\theta \otimes \rho).
    $$
    Let $\bar\pi' \in \bar\Pi$ and define $\pi'_\theta$ similarly. Then
    $$
    \bar V_q(\bar \pi'; \bar \pi) = \mathbb{E}_{\Theta \sim q}\!\left[V_\Theta(\pi'_\Theta; \pi_\Theta)\right].
    $$
\end{lemma}

\begin{proof}
    By definition of $\bar V_q(\bar\pi'; \bar\pi)$ in \eqref{eq:barVq}, we have
    $$
    \bar V_q(\bar\pi'; \bar\pi) = \mathbb{E}\!\left[\sum_{i=0}^{H-1} \bar r\!\left(\Theta_i, X_i, \bar\alpha_i, \bar\rho^{\bar\pi, \varepsilon^0}_i(\Theta)\right)\right],
    $$
    where $(\Theta_i, X_i, \alpha_i)$ is the trajectory of a representative agent acting under $\bar\pi'$ while the population follows $\bar\pi$. By Lemma~\ref{lemma:bar_rho^pi(Theta)=d_Theta_rho^pi}, $\bar\rho^{\bar\pi, \varepsilon^0}_i(\Theta) = \delta_\Theta \otimes \rho^{\pi_\Theta, \varepsilon^0}_{\Theta, i}$ $\PP$-a.s. In the proof of Lemma~\ref{lemma:bar_rho^pi_pi'(Theta)=d_Theta_rho^pi_pi'} it is shown that the joint law of $(X_i, \alpha_i)$ conditioned on $\Theta$ is equal to $\mu_{\Theta, i}^{\pi, \pi', \varepsilon^0}$. Using the definition of $\bar r$ in \eqref{eq:barP-barr} and the fact that $\bar\rho^{\bar\pi, \varepsilon^0}_i(\Theta) = \delta_\Theta \otimes \rho^{\pi_\Theta, \varepsilon^0}_{\Theta, i}$, we obtain
    $$
    \bar r\!\left(\Theta_i, X_i, \bar\alpha_i, \bar\rho^{\bar\pi, \varepsilon^0}_i(\Theta)\right) = r_\Theta\!\left( X_i, \bar\alpha_i, \rho^{\pi_\Theta, \varepsilon^0}_{\Theta, i}\right), \quad \mathbb{P}\text{-a.s.}
    $$
    Taking expectations and conditioning on $\Theta \sim q$ yields
    $$
    \bar V(\bar\pi'; \bar\pi) = \mathbb{E}_{\Theta \sim q}\!\left[\mathbb{E}\!\left[\sum_{i=0}^{H-1} r_{\Theta}\!\left( X_i', \alpha_i', \rho^{\pi_\Theta, \varepsilon^0}_{\Theta, i}\right) \;\middle|\; \Theta\right]\right] = \mathbb{E}_{\Theta \sim q}\!\left[V_\Theta(\pi'_\Theta; \pi_\Theta)\right],
    $$
    which completes the proof.
\end{proof}

\begin{lemma}
\label{lemma:barV=EV_family}
    Let $\{\pi_\theta\}_{\theta \in \bar\Theta} \subset \Pi$ be a family of policies and define $\bar\pi \in \bar\Pi$ by
    $$
    \bar\pi_i(a \mid \theta, x, \bar\rho) = \pi_{\theta,i}\!\left(a \;\middle|\; x,\, \int_{\bar\Theta} \bar\rho(d\theta', \cdot)\,\right).
    $$
    Let $\{\pi'_\theta\}_{\theta \in \bar\Theta} \subset \Pi$ and define $\bar\pi'$ similarly. Then
    $$
    \bar V(\bar\pi'; \bar\pi) = \mathbb{E}_{\Theta \sim q}\!\left[V_\Theta(\pi'_\Theta; \pi_\Theta)\right].
    $$
\end{lemma}

\begin{proof}
    By construction of $\bar\pi$ and $\bar\pi'$, we have for any $\theta \in \bar\Theta$ and $(i, x, a, \rho) \in [H-1] \times \mathcal{X} \times \mathcal{A} \times \Delta_\mathcal{X}$,
    $$
    \bar\pi_i(a \mid \theta, x, \delta_\theta \otimes \rho) = \pi_{\theta,i}(a \mid x, \rho), \quad \bar\pi_i'(a \mid \theta, x, \delta_\theta \otimes \rho) = \pi_{\theta,i}'(a \mid x, \rho).
    $$
    Hence $\bar\pi$ and $\bar\pi'$ satisfy the assumptions of Lemma~\ref{lemma:barV=EV}, which directly gives
    $$
    \bar V(\bar\pi'; \bar\pi) = \mathbb{E}_{\Theta \sim q}\!\left[V_\Theta(\pi'_\Theta; \pi_\Theta)\right]. \qedhere
    $$
\end{proof}

The proof of Theorem~\ref{thm:NE_BE-MFG_induces_NE_PMFG} requires selecting, for each parameter, a best response in $\cM_\theta$, in such a way that the resulting family is again a Bayesian policy. The following lemma shows that such a measurable selection always exists, and that no measurable selection theorem is needed: since $\cA$ is finite, the selection can be written explicitly.

\begin{lemma}
\label{lemma:measurable-br-selection}
    For every $\bar\pi \in \bar\Pi$, there exists $\bar\pi^{\mathrm{br}} \in \bar\Pi$ such that, for every $\theta \in \bar\Theta$, $\pi^{\mathrm{br}}_\theta \in \BR_\theta(\pi_\theta)$, where $\pi_\theta$ and $\pi^{\mathrm{br}}_\theta$ denote the policies induced by $\bar\pi$ and $\bar\pi^{\mathrm{br}}$ as in Theorem~\ref{thm:NE_BE-MFG_induces_NE_PMFG}.
\end{lemma}

\begin{proof}
    Let $\bar\pi \in \bar\Pi$ and $\theta \in \bar\Theta$. An agent deviating from the population policy $\pi_\theta$ faces a finite-horizon MDP whose state is the pair $(X_i, \rho_i)$, where $\rho_i$ is the population mean field. Given $\theta$ and $\pi_\theta$, $\rho_{i+1}$ is a function of $(\rho_i, \varepsilon^0_{i+1})$ only, so $(X_i, \rho_i)$ is a Markov chain and the strategies of $\Pi_0$, which are mappings from $\cX \times \Delta_\cX$ to $\Delta_\cA$, are sufficient: the dynamic programming principle applies over the policy class $\Pi$ used throughout. Define on $\bar\Theta \times \cX \times \Delta_\cX \times \cA$
    $$
    Q_{H-1}(\theta, x, \rho, a) = r_\theta(x, a, \rho), \qquad V_i(\theta, x, \rho) = \max_{a \in \cA} Q_i(\theta, x, \rho, a),
    $$
    $$
    Q_i(\theta, x, \rho, a) = r_\theta(x, a, \rho) + \EE_{\varepsilon^0 \sim \nu^0}\Big[\sum_{x' \in \cX} P_\theta(x' \mid x, a, \rho, \varepsilon^0) \, V_{i+1}\big(\theta, x', \Phi_\theta(\rho, \pi_{\theta, i}; \rho, \varepsilon^0)\big)\Big],
    $$
    for $i \in [H-2]$.

    \textbf{Each $Q_i$ is Borel measurable.} We proceed by downward induction on $i$. For $i = H-1$ this is the measurability of $r$. Assume $V_{i+1}$ is Borel. The mapping $(\theta, \rho, \varepsilon^0) \mapsto P_\theta(x' \mid x, a, \rho, \varepsilon^0)$ is Borel by measurability of $P$, and $(\theta, \rho, \varepsilon^0) \mapsto \Phi_\theta(\rho, \pi_{\theta, i}; \rho, \varepsilon^0)$ is Borel as a finite sum of products of measurable mappings, using that $\bar\pi \in \bar\Pi$ is measurable. Composing with $V_{i+1}$ and summing over the finite set $\cX$ preserves measurability. Since $r$ is bounded, an immediate induction gives $\vert V_j \vert \leq (H-j)\Vert r\Vert_\infty$, so the integrand is bounded and integrating against $\nu^0$ preserves Borel measurability in $(\theta, x, \rho)$ by the Fubini--Tonelli theorem. Hence $Q_i$ is Borel, and so is $V_i$, being a maximum of finitely many Borel mappings.

    \textbf{Selection.} Write $\cZ = \bar\Theta \times \cX \times \Delta_\cX$ and $z = (\theta, x, \rho) \in \cZ$. Fix a total order $<$ on the finite set $\cA$, chosen once and independently of $z$, and let
    $$
    a^*_i(z) = \min\Big\{a \in \cA \mid Q_i(z, a) = \max_{a' \in \cA} Q_i(z, a')\Big\}, \qquad z \in \cZ.
    $$
    For each $a \in \cA$, we have
    \begin{align*}
    \big\{z \in \cZ \mid a^*_i(z) = a\big\} = \, & \Big(\bigcap_{a' \in \cA} \big\{z \in \cZ \mid Q_i(z, a) \geq Q_i(z, a')\big\}\Big) \\
    & \cap \Big(\bigcap_{a'' \in \cA, \, a'' < a} \big\{z \in \cZ \mid Q_i(z, a) > Q_i(z, a'')\big\}\Big),
    \end{align*}
    the first intersection expressing that $a$ is a maximizer of $Q_i(z, \cdot)$, and the second that no action smaller than $a$ is one. Both are finite intersections of Borel subsets of $\cZ$, so $\{z \in \cZ \mid a^*_i(z) = a\}$ is Borel and $a^*_i$ is Borel. Therefore $\bar\pi^{\mathrm{br}}_i(\cdot \mid \theta, x, \rho) = \delta_{a^*_i(\theta, x, \rho)}$ defines a Bayesian strategy and $\bar\pi^{\mathrm{br}} \in \bar\Pi$. By the dynamic programming principle, for every $\theta \in \bar\Theta$ the induced policy $\pi^{\mathrm{br}}_\theta$ maximizes $V_\theta(\cdot \, ; \pi_\theta)$ over $\Pi$, i.e. $\pi^{\mathrm{br}}_\theta \in \BR_\theta(\pi_\theta)$.
\end{proof}

\begin{remark}
    The finiteness of $\cA$ is what makes the selection explicit, so that no measurable selection theorem is needed. Continuity of $P$ and $r$ in $\theta$ would be sufficient but is not required: only their measurability, together with the boundedness of $r$, is used. Under condition (C1) of Assumption~\ref{assumption:likelihood}, $\bar\Theta$ is at most countable and the statement is immediate.
\end{remark}

\textbf{We now provide the proof of Theorem~\ref{thm:NE_BE-MFG_induces_NE_PMFG}.}
\begin{proof}
    Let $\bar\pi^{\mathrm{br}} \in \bar\Pi$ be given by Lemma~\ref{lemma:measurable-br-selection} applied to $\bar\pi^*$, and define
    $$
    \delta(\theta) = V_\theta\big(\pi^{\mathrm{br}}_\theta ; \pi^{**}_\theta\big) - V_\theta\big(\pi^{**}_\theta ; \pi^{**}_\theta\big), \qquad \theta \in \bar\Theta.
    $$
    The mapping $\delta$ is measurable, being a difference of values associated with two policies of $\bar\Pi$, and $\delta \geq 0$ since $\pi^{\mathrm{br}}_\theta \in \BR_\theta(\pi^{**}_\theta)$ for every $\theta \in \bar\Theta$. Applying Lemma~\ref{lemma:barV=EV} twice,
    $$
    \EE_{\Theta \sim q}\big[\delta(\Theta)\big] = \bar V_q\big(\bar\pi^{\mathrm{br}} ; \bar\pi^*\big) - \bar V_q\big(\bar\pi^* ; \bar\pi^*\big) \leq 0,
    $$
    where the inequality holds because $\bar\pi^*$ is a NE of $\bar{\cM}_q$. A non-negative measurable mapping with non-positive integral vanishes almost everywhere, so $\delta = 0$ $q$-a.s. Since $\delta(\theta) = 0$ means exactly that $\pi^{**}_\theta \in \BR_\theta(\pi^{**}_\theta)$, we conclude that $\pi^{**}_\Theta$ is a NE of $\cM_\Theta$, $q$-a.s.
\end{proof}

\noindent\textbf{Dimension of the BE-MFG.} At first sight, the BE-MFG reintroduces a curse of dimensionality: its state space is $\bar\Theta\times\cX$, and Bayesian policies depend on mean fields in $\cP(\bar\Theta\times\cX)$. However, the mean fields generated by the BE-MFG have a specific structure: for any Bayesian policy $\bar\pi$, $\bar\rho^{\bar\pi,\varepsilon^0}_i(\Theta) = \delta_\Theta \otimes \rho^{\pi_\Theta,\varepsilon^0}_{\Theta,i}$ for every $i \in [H-1]$ (Lemmas~\ref{lemma:bar_rho^pi(Theta)=d_Theta_rho^pi}--\ref{lemma:bar_rho^pi(Theta)=d_Theta_rho^pi_family}), so that the population never leaves the set $\{\delta_\theta \otimes \rho : \theta \in \bar\Theta, \rho \in \Delta_\cX\}$. Consequently, as in Theorem~\ref{thm:NE_BE-MFG_induces_NE_PMFG}, Bayesian strategies only need to be evaluated at mean fields of the form $\delta_\theta \otimes \rho$, i.e., it suffices to consider strategies defined on $\bar\Theta \times \cX \times \Delta_\cX$ instead of $\bar\Theta \times \cX \times \cP(\bar\Theta \times \cX)$: the augmentation enlarges the state space, but not the space of mean fields on which the policy depends. Moreover, when $P_\theta$ and $r_\theta$ depend regularly on $\theta$, a neural network policy can share information across parameters, which is impossible when each parameter is solved separately. This is the parametrization we use (App.~\ref{appendix:film}).

\clearpage

\section{Computing Nash equilibrium for PMFG and BE-MFG}
\label{appendix:nash_equilibrium}

The overall Bayesian inference pipeline of Sec.~\ref{sec:algos} is summarized in Algorithm~\ref{algo:bays-inf-mfgs}. Its first step, computing the BE-NE, is detailed in the rest of this section.

\begin{algorithm}[htbp]
\caption{Bayesian Inference for PMFGs}
\begin{algorithmic}[1]
\label{algo:bays-inf-mfgs}
\REQUIRE PMFGs $\{\cM_\theta\}_{\theta\in\bar\Theta}$; prior $q$ over $\bar \Theta$; available time steps $\cI$; observations $(\hat \xi^n_{\theta_0, \cI})_{1\leq n \leq N}$.   \\
\STATE Compute a BE-NE $\bar \pi^*$ of $\bar \cM_q$.
\STATE Compute conditional density $ \Gamma_{\cdot,\cI}$ to learn the conditional distributions $(\Gamma_{\theta, \cI})_{\theta \in \bar \Theta}$, by sampling mean-field trajectories using $\bar \pi^*$ and $\bar \cM_q$.
\STATE Compute posterior $\bar q_N$ using $ \Gamma_{\cdot, \cI}^\gamma$ and the observations $(\hat \xi^n_{\theta_0, \cI})_{1\leq n \leq N}$.
\STATE \textbf{Return} $\bar q_N$, $\hat \theta^{\mathrm{MAP}}_N = \arg\max_{\theta \in \bar \Theta}\bar q_N(\theta)$, $\bar \pi^*$.
\end{algorithmic}
\end{algorithm}

\subsection{Algorithms for general MFGs with common noise}
\label{appendix:nash_equ_general_algo}

Computing a Nash equilibrium for a MFG with common noise requires solving a fixed-point problem: the policy must be a best response to the mean-field flow it induces. We rely on the FP algorithm~\cite{perrin2020fictitious}, which iteratively trains a best response against the average of the policy history, and has been shown to converge to a Nash equilibrium under standard assumptions. Each best response is trained by maximizing the expected return of a representative agent against the current aggregate mean-field flow, using a neural network policy $\pi^\gamma$ (Algorithm~\ref{algo:nn_best_response}). The FP loop is summarized in Algorithm~\ref{algo:fictitious}. Once FP has converged, the history of policies $(\pi^0, \dots, \pi^K)$ encodes the Nash equilibrium as an implicit mixture. To obtain a single compact policy, we distill this history into a single neural network $\pi^\gamma$ via Mean Field Imitation Learning (MF-IL, Algorithm~\ref{algo:il_agg_meanfield}), which matches the state-action mean-field trajectory induced by $\pi^\gamma$ to the aggregate of the history, following~\cite{lambrecht2026}.

\begin{algorithm}[htbp]
\caption{Best Response Algorithm}
\begin{algorithmic}[1]
\label{algo:nn_best_response}
\REQUIRE MFG $\cM = (\cX, \cA, \cE^0, \nu^0, P, r, H, \rho_0)$; Sequence of policies $(\pi^0, \dots, \pi^m) \in \Pi^{m+1}$; Iterations $J$; Batch size $B$; Initial Parameters $\gamma_0$; Optimizer. \\ 
\FOR{$j = 1, \dots, J$} 
    \STATE Reset batch loss: $\cL(\gamma_{j-1}) = 0$.
    \STATE Sample a batch of i.i.d. common noises $(\varepsilon^{0, s}_t)_{0 \leq t \leq H-1}$ for $s=1, \dots, B$.
    \FOR{$s = 1, \dots, B$}
        \STATE Compute the aggregated sequence $\tilde{\rho}^{m,s}$ induced by $(\pi^0, \dots, \pi^m)$.
        \STATE Compute the mean-field sequence $\rho^{\gamma_{j-1}, s}$ induced by the learner $\pi^{\gamma_{j-1}}$ responding to $\tilde{\rho}^{m,s}$.
        \STATE Accumulate the loss:
        \[
        \cL(\gamma_{j-1}) \leftarrow \cL(\gamma_{j-1}) - \frac{1}{B} \sum_{t=0}^{H-1} \sum_{x, a} \rho^{\gamma_{j-1},s}_t(x) \pi^{\gamma_{j-1}}_t(a \mid x, \tilde{\rho}^{m,s}_t) r(x, a, \tilde{\rho}^{m,s}_t)
        \]
    \ENDFOR
    \STATE Compute the gradient $\nabla_\gamma \cL(\gamma_{j-1})$ and update parameters via the optimizer to obtain $\gamma_{j}$. 
\ENDFOR 
\STATE \textbf{Return} Optimized parameters $\gamma_{J}$ defining the best response policy $\pi^{\gamma_J}$.
\end{algorithmic}
\end{algorithm}
\vspace{1em}

\begin{algorithm}[htbp]
\caption{Fictitious Play for MFG with Common Noise}
\begin{algorithmic}[1]
\label{algo:fictitious}
\REQUIRE MFG $\cM = (\cX, \cA, \cE^0, \nu^0, P, r, H, \rho_0)$; Initial policy $\pi^0 \in \Pi$; Number of iterations $K$.
\FOR{$k = 1, \dots, K$}
    \STATE Train $\pi^k$ as the best response to $(\pi^0, \dots, \pi^{k-1})$, using Algorithm~\ref{algo:nn_best_response}.
\ENDFOR
\STATE \textbf{Return} policy tuple $(\pi^0, \dots, \pi^K)$.
\end{algorithmic}
\end{algorithm}
\vspace{1em}

\begin{algorithm}[htbp]
\caption{Mean Field Imitation Learning (MF-IL)}
\begin{algorithmic}[1]
\label{algo:il_agg_meanfield}
\REQUIRE MFG $\cM$; Policy sequence $(\pi^0, \dots, \pi^K)$; 
Initial parameters $\gamma_0$ for parametric policy $\pi^\gamma$; Optimizer; 
Iteration number $J$; Batch size $B$.
\FOR{$j = 1, \dots, J-1$}
    \STATE Sample common noises $(\varepsilon^{0, s}_t)_{0\leq t \leq H-1}^{1\leq s\leq B}$. %
    \FOR{$s = 1, \dots, B$}
        \STATE Compute $\tilde{\mu}^{K,s}$ from $(\pi^0, \dots, \pi^K)$.%
        \STATE Generate $\mu^{s, \pi^{\gamma_{j-1}}}$ induced by $\pi^{\gamma_{j-1}}$.
    \ENDFOR
    \STATE Compute 
        $\cL(\gamma_{j-1}) = \frac{1}{B} \sum_{s,t} \|\mu^{s, \pi^{\gamma_{j-1}}}_t - \tilde{\mu}^{K,s}_t \|_1.$
    \STATE Obtain $\gamma_j$ by one step of gradient of $\cL(\gamma_{j-1})$.
\ENDFOR
\STATE \textbf{Return} Trained parametric policy $\pi^{\gamma_J}$.
\end{algorithmic}
\end{algorithm}

\subsection{Nash equilibrium for PMFGs}
To compute the NE policies from which we sample observations of the ground truth model $\mathcal{M}_{\theta_0}$, we apply the algorithms of App.~\ref{appendix:nash_equ_general_algo} with a standard Multi-Layer Perceptron (MLP) as the parametric policy.

\subsection{Nash equilibrium for BE-MFG}
\label{appendix:film}

To compute the BE-NE $\bar\pi^*$, we apply the same FP pipeline of App.~\ref{appendix:nash_equ_general_algo}, replacing the standard MLP policy with a \textbf{FiLM-conditioned policy}. We briefly describe how FiLM works and refer to \cite{perez2018film} for more details. 

\noindent\textbf{Architecture.} The network takes as input the time step $t$, the state $x$, the current mean field $\rho$, and the parameter $\theta$, and outputs a softmax distribution over actions. The backbone consists of $L$ linear layers of hidden dimension $l$ with tanh activations. At each layer $i$, before applying the activation, the hidden representation $h_i$ is modulated by a small MLP applied to $\theta$:
\[
h^{\text{FiLM}}_i = \tanh(W^{(1)}_i \theta + b^{(1)}_i), \quad
(w_i, b_i^{(3)} )= W^{(2)}_i\,h^{\text{FiLM}}_i + b^{(2)}_i, \quad h_i = W_ih'_{i-1} + b_i
\]
$$
h_{i}' = \tanh\!\left(h_{i} \odot (1 + w_i) + b_i^{(3)}\right),
$$
The FiLM MLP maps $\theta \in \mathbb{R}^{d}$ to a hidden layer of size $l'$, then outputs a vector of size $2 \times l$ which is split into the scale $w_i$ and shift $b_i^{(3)}$ applied to the backbone activations.

The learnable parameters of the network are $\gamma = \{W^{(1)}_i, b^{(1)}_i, W^{(2)}_i, b^{(2)}_i\}_{i\in[L]}$ for the FiLM MLPs, $\{W_i, b_i\}_{i\in[L]}$ for the backbone layers, and the output layer weights.

\noindent\textbf{Motivation.} The BE-NE $\bar\pi^*$ must represent a \emph{family} of NE policies indexed by $\theta$. A naive shared MLP with $\theta$ concatenated to the input tends to average over $\theta$ rather than sharply distinguishing between parameter regimes. FiLM solves this by allowing $\theta$ to directly rescale and shift the internal representations at every layer, making the network's computation path highly sensitive to $\theta$. This is essential: if $\bar\pi^*$ does not distinguish between different values of $\theta$, the simulated trajectories $\xi^{\bar\pi^*, \varepsilon^0}_{\theta, \mathcal{I}}$ will be incorrectly distributed, corrupting the training of $\Gamma^\gamma_{\theta,\mathcal{I}}$ and all downstream inference.

\clearpage

\section{Mean field inverse reinforcement learning}
\label{appendix:mf-irl}

We refer to \cite{chen2022individualinverseRL} for the original description of the method.
In this appendix we describe precisely how we adapted it to our setting.

\paragraph{Overview of the original method.}
The MFIRL framework of~\cite{chen2022individualinverseRL} recovers a reward
function $r_\theta(x,a,\rho)$ from expert state-action trajectories issued of a PMFG $\tilde \cM_\theta = (\cX, \cA, \tilde P, r_\theta, H, \rho_0)$ in absence of common noise, meaning that the kernel $\tilde P$ does not depend of an extra source of randomness $\varepsilon^0$, and the policies are sequences of population-unaware policies $\tilde \pi(a\vert x)$. We denote the previous set by $\tilde \Pi$. We emphaze that in absence of common noise, the mean fields trajectories at equilibrium are deterministic. 
Given a state-action mean-field trajectory $\hat \mu_{\theta_0}$ for an unknown parameter $\theta_0$, the method solves the objective
$$
  \max_{\theta}\;\mathcal{L}(\theta)
  \;=\;
  \sum_{i=0}^{H-1}\sum_{x, a}\kappa^i \hat \mu_{\theta_0, i}(x, a)
      r_\theta(x,a,\hat \rho_{\theta_0, i})
  -\,J\!\left(\hat \rho_{\theta_0},\tilde \pi_\theta \right),
$$
where $\kappa$ is a discount factor, $\hat \rho_{\theta_0, i} = \sum_a \hat \mu_{\theta_0, i}(\cdot, a)$, $J(\hat \rho_{\theta_0},\, \tilde\pi_\theta) = \sum_{i=0}^{H-1} \kappa^i \sum_{x, a} \hat \rho_{\theta_0, i}(x)\, \tilde\pi_{\theta,i}(a|x)\, r_\theta(x, a, \hat \rho_{\theta_0, i})$
and 
$\tilde \pi_\theta$ is a Boltzmann best response to
$\hat \mu_{\theta_0}$ under $r_{\theta}$, computed by backward induction
via the soft Q-function
\[
  \tilde Q_\theta (i,x,a,\hat \rho_{\theta_0})
  \;=\;
  r_\theta(x,a,\rho_{ i})
  +\kappa\,\mathbb{E}_{X\sim \tilde P(\cdot \vert x, a, \hat \rho_{\theta_0, i}) }\!\left[\sum_{a'}\tilde\pi_{\theta, i+1}(a'|X)
      \tilde Q_\theta(i+1,X,a',\hat \rho_{\theta_0})\right],
\]
with $\tilde\pi_{\theta, i}(a|x)\propto\exp(\tau^{-1}\,\tilde Q_\theta(i,x,a,\rho))$, where $\tau > 0$ is a temperature.
In the original paper the $Q$-function is tabular, which is tractable because there is
no common noise and the state space is small.

\paragraph{Adaptations required by our setting.}
\;

\textbf{(i) Common Noise Approximation.}
Our setting considers PMFGs $\mathcal{M}_\theta = (\mathcal{X}, \mathcal{A}, \cE^0, \nu^0, P, r_\theta, H, \rho_0)$ where the transition kernel $P$ depends explicitly on a common noise realization $\varepsilon^0 \sim \nu^0$. When applying MF-IRL, we replace $P$ by its common noise averaged conterpart $\tilde P (\cdot \vert x, a, \rho) = \EE_{\varepsilon^0}[P(\cdot|x,a,\rho,\varepsilon^0)]$.

Under common noise, the mean-field trajectory $\hat \rho_{\theta_0}$ is stochastic, so the observer has access to $N$ independent realizations $(\hat{\mu}^n_{\theta_0,i})_{1 \le n \le N}$ rather than a single deterministic trajectory. We approximate $\hat \mu_{\theta_0}$ by the average $\hat{\mu}^E_i = \frac{1}{N}\sum_{n=1}^N \hat{\mu}^n_{\theta_0,i}$, and compute $\tilde{Q}_\theta$ using $\tilde P(\cdot|x,a,\hat{\rho}^E_i)$, where $\hat{\rho}^E_i = \sum_a \hat{\mu}^E_i(\cdot,a)$. 

We therefore follow the population-unaware approximation: the Boltzmann policy
$\tilde\pi_{\theta, t}(a|x)$ is computed from the soft Q-function evaluated at the
\emph{fixed} averaged mean field $\hat\mu^E_{\theta_0, i} = \frac{1}{N}\sum_{n = 1}^N\hat\mu_{\theta_0}^n$,
using the common noise averaged transition matrix $\tilde P(\cdot|x,a,\hat\rho^E_i)$.

This approximation is consistent with the population-unaware/aware policy
distinction introduced in~\cite{lambrecht2026} for the common-noise setting.

Concretely, our implementation computes the averaged transition matrix $\tilde P$ once, then
runs the backward scan
\begin{align*}
  \tilde Q_{\theta}(H-1) &= R_{\theta}(H-1),\\
  \tilde Q_{\theta }(i) &= R_\theta(i) + \kappa\,\tilde P\cdot V_\theta(i+1),\quad
         V_\theta(i+1, x)=\sum_a\tilde\pi_{\theta, i+1}(a|x)\,\tilde Q_{\theta}(i+1, x,a),
\end{align*}
where $R_{\theta}(i,x,a)=r_\theta(x,a,\hat\rho^E_i)$ uses the parametric reward at
the empirical density, and $\tilde\pi_{\theta, i}=\operatorname{softmax}(\tau^{-1} \tilde Q_{\theta}(i))$.

\textbf{(ii) Requirement of state--action observations ($\xi=\mu$).}
The backward induction requires the full state-action mean-field trajectory
$\hat\mu_{\theta_0}\in(\Delta_{\mathcal{X}\times\mathcal{A}})^H$.
The method therefore cannot be applied when the observer only sees state
distributions $\xi=\rho$, which is the harder and more realistic setting we
also evaluate. This is a structural limitation relative to our Bayesian approach,
which handles both $\xi=\mu$ and $\xi=\rho$.

\textbf{(iii) Handling partial observations ($\mathcal{I}\ne[H-1]$).}
The original paper observes the full trajectory.  When only a subset
$\mathcal{I}\subset[H-1]$ of time steps is available we construct the full
$H$-step empirical flow $\hat \mu^E_{\theta_0}$ by nearest-future imputation:
for any unobserved step $t$, we substitute the next observed step
$\min\{i\in\mathcal{I}:i\ge t\}$, or the last observed step if none exists.
This natural imputation allows the backward induction to run over the full
horizon, and we investigate its degradation as $|\mathcal{I}|$ decreases.

The full procedure is given in Algorithm~\ref{algo:mf-irl}.

\begin{algorithm}[htbp]
\caption{Adaptation of MF-IRL to our setting}
\label{algo:mf-irl}
\begin{algorithmic}[1]
\REQUIRE Family of PMFGs $\{M_\theta\}_{\theta\in\bar\Theta}$;
         parametric reward $r_\theta$;
         $N$ expert state-action trajectories
         $(\hat\mu^{n}_{\theta_0, \mathcal{I}})_{1\le n\le N}$
         at accessible steps $\mathcal{I}$;
         temperature $\tau$; discount $\kappa$; Initial guess $\theta^0 \in \bar \Theta$; Number of iterations $T_{\mathrm{IRL}}$; Optimizer.
\STATE Compute empirical mean-field:
       $\hat\mu^E_{\theta_0, i} = \frac{1}{N}\sum_{n=1}^N \hat\mu^n_{\theta_0, i}$
       for $i\in\mathcal{I}$.
\STATE Impute missing steps: for $i\notin\mathcal{I}$, set
       $\hat\mu^E_i = \hat\mu^E_{\min\{i'\in\mathcal{I}:i'\ge i\}}$
\STATE Compute the common noise averaged transition matrix $\tilde P(x'\vert x, a,\hat\rho^E_i) = \EE_{\varepsilon^0}[ P(x'\vert x, a, \hat\rho^E_i,\varepsilon^0)]$.
\FOR{for $m = 0$ to $T_{\mathrm{IRL}}-1$}
  \STATE Compute $\tilde Q_{\theta^m}$.
  \STATE Compute the loss $\cL(\theta^m) = \sum_{i=0}^{H-1}\sum_{x, a}\kappa^i \hat \mu_{\theta_0, i}(x, a)
      r_{\theta^m}(x,a,\hat \rho_{\theta_0, i})
  -\,J\!\left(\hat \rho_{\theta_0},\tilde \pi_{\theta^m} \right)$.
  \STATE Proceed with one gradient step to obtain $\theta^{m+1}$.
\ENDFOR
\STATE \textbf{Return} $\theta^{T_{\mathrm{IRL}}}$.
\end{algorithmic}
\end{algorithm}

\clearpage

\section{Learning $\Gamma_{\theta, \cI}$}
\label{appendix:learning-gamma}

\subsection{Principal component analysis}
\label{app:pca}
Since $\xi_\mathcal{I} \in (\Delta_\mathcal{X})^{|\mathcal{I}|}$ is high-dimensional, we first apply principal component analysis (PCA) to a training set $\{(\theta^n, \xi^{*,n}_{\theta^n, \cI})\}_{n=1}^{N_{\mathrm{pca}}}$ sampled using the BE-NE $\bar \pi^*$ and the prior $q$. Centering the trajectories and computing the SVD of the centered matrix, we retain the top $d_{\mathrm{PCA}}$ components explaining $99\%$ of the variance, yielding a projection $W \in \mathbb{R}^{D \vert \cI \vert \times d_{\mathrm{PCA}}}$ where $D$ refers to the dimension of the observation space $\Xi$. Each observation is then mapped to a normalized PCA embedding $z = ((\xi_\mathcal{I} - \bar\xi) W - \bar z) / \sigma_z \in \mathbb{R}^{d_{\mathrm{PCA}}}$. PCA is a well understood and classic method in the Machine Learning literature \cite{jolliffe2016principal}. 
\subsection{Gaussian processes estimation}
\label{appendix:gp}

We describe the GP-based method used as a lightweight alternative to the MAF. Contrary to the MAF, it does not learn the conditional density $\Gamma_{\theta,\mathcal{I}}$: it is an inverse regression, which directly models the map $z \mapsto \theta$ in PCA space, and outputs a Gaussian approximation of the posterior used for MAP and BMA inference.

\paragraph{GP training in PCA space.} We place a GP prior on the map $z \mapsto \theta$ with a linear kernel $k(z, z') = z^\top z'$, making the model equivalent to Bayesian linear regression in the PCA feature space. Given $N_{\mathrm{train}}$ training pairs $\{(z^n, \theta^n)\}_{n=1}^{N_{\mathrm{train}}}$ sampled from the BE-NE $\bar\pi^*$ and the prior $q$, the Gram matrix is $\mathbf{K} = (k(z^n, z^{n'}))_{n,n'} \in \mathbb{R}^{N_{\mathrm{train}}\times N_{\mathrm{train}}}$ and the GP is fitted by computing
$w = (\mathbf{K} + \sigma^2_{\mathrm{noise}} I)^{-1} Y,$
where $Y = (\theta^1, \dots, \theta^{N_{\mathrm{train}}})^\top \in \mathbb{R}^{N_{\mathrm{train}} \times d}$ stacks the targets and $\sigma_{\mathrm{noise}}^2$ is a hyperparameter controlling the confidence of the model. The training algorithm is summarized in Algorithm~\ref{algo:train-gp}.

\paragraph{Posterior inference for $N=1$.} For a single observation $\hat\xi_\mathcal{I}$, let $\hat z$ be its PCA embedding and $\hat k = (k(\hat z, z^n))_{n=1}^{N_{\mathrm{train}}} \in \mathbb{R}^{N_{\mathrm{train}}}$. The GP posterior is Gaussian $\mathcal{N}(\hat\theta^{\mathrm{MAP}}, \hat\sigma^2)$ with
\begin{align*}
\hat\theta^{\mathrm{MAP}} &= \hat k^\top w, \\
\hat\sigma^2 &= k(\hat z, \hat z) - \hat k^\top (\mathbf{K} + \sigma^2_{\mathrm{noise}} I)^{-1} \hat k,
\end{align*}
where $\hat\sigma^2$ is the posterior variance measuring the uncertainty of the estimate. The MAP estimator is $\hat\theta^{\mathrm{MAP}}$.

\paragraph{Multiple observations ($N > 1$).} When $N$ i.i.d.\ observations $(\hat\xi^n_{\theta_0,\mathcal{I}})_{1 \le n \le N}$ are available, each observation $\hat\xi^n_{\theta_0, \mathcal{I}}$ yields an independent GP posterior $\mathcal{N}(m_n, v_n)$, where $m_n = \hat k_n^\top w$ is the predicted parameter and $v_n = k(\hat z_n, \hat z_n) - \hat k_n^\top(\mathbf{K} + \sigma^2_{\mathrm{noise}}I)^{-1}\hat k_n$ is the associated posterior variance. Each $\mathcal{N}(\theta; m_n, v_n)$ approximates the single-observation posterior, which is proportional to $q(\theta)\Gamma_{\theta,\cI}(\hat\xi^n_{\theta_0,\cI})$. Since the observations are conditionally independent given $\theta$, Bayes' rule~\eqref{eq:barqN} involves the product of the likelihoods and a single prior factor, whereas the product of the $N$ single-observation posteriors contains the prior $N$ times. We therefore combine them as
\[
\bar q_N^{\mathrm{GP}}(\theta) \propto q(\theta)^{1-N} \prod_{n=1}^N \mathcal{N}(\theta; m_n, v_n).
\]
When the prior is Gaussian, $q = \mathcal{N}(\mu_0, v_0)$, $\bar q_N^{\mathrm{GP}}$ is Gaussian, $\bar q_N^{\mathrm{GP}} = \mathcal{N}(\hat\theta^{\mathrm{MAP}}_N, v_{\mathrm{comb}})$, with
\[
v_{\mathrm{comb}}^{-1} = \sum_{n=1}^N v_n^{-1} - (N-1) v_0^{-1}, \qquad
\hat\theta^{\mathrm{MAP}}_N = v_{\mathrm{comb}} \Big( \sum_{n=1}^N v_n^{-1} m_n - (N-1) v_0^{-1} \mu_0 \Big).
\]
When the prior is uniform on $\bar\Theta$, $q^{1-N}$ is constant on $\bar\Theta$, and $\bar q_N^{\mathrm{GP}}$ is the product of the $N$ Gaussians restricted to $\bar\Theta$. Both reduce to the single-observation GP posterior when $N = 1$. The BMA policy is then $\pi_{\bar q_N^{\mathrm{GP}}} = \mathbb{E}_{\theta \sim \bar q_N^{\mathrm{GP}}}[\pi^*_\theta]$.

\begin{algorithm}[htbp]
\caption{Training the Linear Gaussian Processes Estimator}
\begin{algorithmic}[1]
\label{algo:train-gp}
\REQUIRE Family of PMFGs $\{\mathcal{M}_\theta\}_{\theta\in\bar\Theta}$; BE-NE policy $\bar\pi^*$ of $\bar{\mathcal{M}}_{q}$; Type of observations $\xi$; available time steps $\mathcal{I}$; Number of PCA training samples $N_{\mathrm{pca}}$; PCA variance threshold $p_{\mathrm{var}}$; Number of GP training samples $N_{\mathrm{train}}$; Noise hyperparameter $\sigma^2_{\mathrm{noise}}$.
\STATE Sample $N_{\mathrm{pca}}$ i.i.d.\ pairs $(\Theta^b, \xi^{\bar\pi^*, \varepsilon^{0,b}}_{\Theta^b, \mathcal{I}})_{1\leq b \leq N_{\mathrm{pca}}}$ with $\Theta^b \sim q$ and $\varepsilon^{0,b} \sim (\nu^0)^{\otimes H}$, using $\bar \pi^*$.
\STATE Compute PCA projection $W \in \mathbb{R}^{D \vert \cI \vert \times d_{\mathrm{PCA}}}$ retaining $p_{\mathrm{var}}$ variance; store mean $\bar\xi$ and normalization $(\bar z, \sigma_z)$.
\STATE Sample $N_{\mathrm{train}}$ i.i.d.\ pairs $(\Theta^n, \xi^{\bar\pi^*, \varepsilon^{0,n}}_{\Theta^n, \mathcal{I}})_{1\leq n \leq N_{\mathrm{train}}}$ with $\Theta^n \sim q$ and $\varepsilon^{0,n} \sim (\nu^0)^{\otimes H}$.
\FOR{$n = 1$ \TO $N_{\mathrm{train}}$}
    \STATE Project: $z^n = \left((\xi^n_\mathcal{I} - \bar\xi)\,W - \bar z\right) / \sigma_z$.
\ENDFOR
\STATE Compute Gram matrix: $\mathbf{K}_{nn'} = (z^n)^\top z^{n'}$ for all $n, n' \in [N_{\mathrm{train}}]$.
\STATE Compute $w = (\mathbf{K} + \sigma^2_{\mathrm{noise}}\,I)^{-1} Y$, where $Y = (\Theta^1, \dots, \Theta^{N_{\mathrm{train}}})^\top$.
\STATE \textbf{Return} $w, \mathbf{K}$.
\end{algorithmic}
\end{algorithm}

\subsection{Masked auto regressive flows}
\label{appendix:maf}

Normalizing flows like MAF are powerful density estimators that transform a simple distribution into a complex one through a series of bijective transformations, making them ideal for capturing the intricate dependencies in Mean Field trajectories.  

MAF learns the  conditional density $\Gamma_{\theta,\mathcal{I}}(\xi)$ by transforming a simple base distribution $p_0 = \mathcal{N}(0, I_{d_{pca}})$ (where $d_{pca}$ is the dimension retained by the PCA as explained in \ref{app:pca}) into a complex one through a sequence of $L$ bijective coupling layers. 
 At each layer, the current representation $x^{(i-1)} \in \mathbb{R}^{d_{\mathrm{PCA}}}$ (with $x^{(0)} = z$ the PCA embedding of $\xi_\mathcal{I}$) is split into two halves $(x^{(i-1)}_1, x^{(i-1)}_2)$ and only the second half is transformed:
$$
[s^{(i)}, t^{(i)}] = \mathrm{MLP}^{(i)}([x^{(i-1)}_1;\, \theta]), \qquad y^{(i)}_2 = x^{(i-1)}_2 \odot \exp(s^{(i)}) + t^{(i)}, \qquad x^{(i)} = [ y^{(i)}_2, x^{(i-1)}_1\,],
$$
where $\mathrm{MLP}^{(i)}$ is a two-hidden-layer network with tanh activations. Let $\gamma$ be the trainable parameters of a MAF, and denote $\Gamma^\gamma_{\cdot, \cI}$ the transformation defined above. By the change-of-variables formula, the MAF approximates the log-likelihood of a PCA-projected observation $z = ((\xi^*_{\theta, \cI} - \bar\xi) W - \bar z) / \sigma_z$ under parameter $\theta$ as:
$$
\log \Gamma^{\gamma}_{\theta,\mathcal{I}}(z) = \log p_0(x^{(L)}) + \sum_{i=1}^L \mathbf{1}^\top s^{(i)},
$$
where $x^{(L)}$ is the final transformed variable and $z = x^{(0)}$ is the input. The MAF thus returns a tractable scalar approximation of $\log \Gamma_{\theta,\mathcal{I}}(\xi)$ for any $(\theta, \xi)$ pair, which is the quantity needed for Bayesian inference. The training objective is the expected negative log-likelihood (NLL) over pairs $(\Theta, \xi^*_{\Theta,\mathcal{I}})$ sampled from the BE-NE $\bar\pi^*$ and the prior $q$:
$$
\mathcal{L}(\gamma) = -\mathbb{E}_{\Theta \sim q}\!\left[\log \Gamma^{\gamma}_{\Theta,\mathcal{I}}(z^*_\Theta)\right], \quad z^*_\Theta = ((\xi^{*}_{\Theta,\mathcal{I}} - \bar\xi) W - \bar z) / \sigma_z,
$$
which is minimized by stochastic gradient descent as detailed in Algorithm~\ref{algo:train-maf-pca}.

\paragraph{Ensemble of MAF.}  It is important to note that training a MAF can be highly sensitive to the initial values of its parameters. To mitigate the dependency on a specific initialization and ensure more robust density estimation, we employ an ensemble approach. Specifically, we train $M$ independent MAF models starting from different random initializations. Let us denote the models $\Gamma^{\gamma_1}_\mathcal{I}, \dots, \Gamma^{\gamma_M}_\mathcal{I}$. We then aggregate their outputs by averaging their log-likelihoods to produce a more stable and reliable estimate of the global log-likelihood. For each $\theta \in \bar\Theta$ and $\xi \in \Xi$, define
$$
\Gamma^{\gamma_1,\dots,\gamma_M}_{\theta,\mathcal{I}}(z) \propto \exp\!\left(\frac{1}{M}\sum_{j=1}^M \log \Gamma^{\gamma_j}_{\theta,\mathcal{I}}(z)\right).
$$

\begin{algorithm}[htbp]
\caption{Training MAF $\Gamma^{\gamma}_{\mathcal{I}}$}
\begin{algorithmic}[1]
\label{algo:train-maf-pca}
\REQUIRE Family of PMFGs $(\mathcal{M}_\theta)_{\theta\in\bar\Theta}$; BE-NE policy $\bar\pi^*$ of $\bar{\mathcal{M}}_q$; Type of observations $\xi \in \{\rho, \mu\}$; available time steps $\mathcal{I} \subset [H-1]$; Initial parameters $\gamma_0$ for the MAF $\Gamma^{\gamma}_{\mathcal{I}}$; Number of iterations $T$; Batch size $B$; Optimizer; 
 Number of PCA training samples $N_{\mathrm{pca}}$; PCA variance threshold $p_{\mathrm{var}}$.
\STATE Sample $N_{\mathrm{pca}}$ i.i.d.\ pairs $(\Theta^b, \xi^{\bar\pi^*, \varepsilon^{0,b}}_{\Theta^b, \mathcal{I}})_{1\leq b \leq N_{\mathrm{pca}}}$ with $\Theta^b \sim q$ and $\varepsilon^{0,b} \sim (\nu^0)^{\otimes H}$, using $\bar \pi^*$.
\STATE Compute PCA projection $W \in \mathbb{R}^{D \vert \cI \vert \times d_{\mathrm{PCA}}}$ retaining $p_{\mathrm{var}}$ variance; store mean $\bar\xi$ and normalization $(\bar z, \sigma_z)$.
\FOR{$t = 1$ \TO $T$}
    \STATE Sample an i.i.d.\ batch of parameters $(\Theta^b)_{1\leq b\leq B}$ from prior $q$.
    \FOR{$b = 1$ \TO $B$}
        \STATE Sample one common noise trajectory $\varepsilon^{0,b} \sim (\nu^0)^{\otimes H}$.
        \STATE Compute associated mean field trajectory $\xi^{\bar\pi^*, \varepsilon^{0,b}}_{\Theta^b, \mathcal{I}}$.
        \STATE Project into PCA space: $z^b = ((\xi^{\bar\pi^*, \varepsilon^{0,b}}_{\Theta^b, \mathcal{I}} - \bar\xi) W - \bar z) / \sigma_z$.
    \ENDFOR
    \STATE Compute NLL loss: $\mathcal{L}(\gamma_t) = -\dfrac{1}{B} \displaystyle\sum_{b=1}^B \log \Gamma^{\gamma_t}_{\Theta^b, \mathcal{I}}(z^b)$.
    \STATE Update $\gamma_{t+1}$ using one step of the Optimizer on $\nabla_\gamma \mathcal{L}(\gamma_t)$.
\ENDFOR
\STATE \textbf{Return} Trained estimator $\Gamma^{\gamma_T}_{\mathcal{I}}$.
\end{algorithmic}
\end{algorithm}

\clearpage

\section{Experiment 1 (Beach Bar): details}
\label{appendix:beachbar}

\subsection{Model description} 

Let $\cX = \{0, \dots, 6 \}$, $\cA = \{-1, 0, 1\}$, $H = 20$, and $\cE^0 = \{-1, 0, 1\}^{|\cX|}$.
The environment is disrupted by external events that simultaneously displace all agents at specific locations. This is governed by a common noise whose frequency depends on a constant $\eta\ge0$.  
Given $\eta$, and common noise $e^0 \in \cE^0$, the dynamics are governed by a transition kernel $P$ that incorporates both \emph{common noise} $e^0(x)$ and \emph{idiosyncratic noise} (represented by a uniform shift $s \in \{-1, 0, 1\}$). 
Let for any $(x, a, \rho, e^0) \in \cX\times\cA\times\Delta_\cX\times\cE^0$, $P(x, a, \rho, e^0) = \frac{1}{3} \sum_{s \in \{-1, 0, 1\}} \delta_{x + a + e^0(x) + s \pmod{|\cX|}}$.
The common noise is $\varepsilon^0 \sim (\nu^0_\eta)^{\otimes 7}$, where $\nu^0_\eta$ is the distribution of $\varepsilon^0_1 \varepsilon^0_2$ with $\varepsilon^0_1 \sim \text{Bernoulli}(\eta)$ and $\varepsilon^0_2 \sim \text{Unif}\{-1, 0, 1\}$.
The reward function, reflecting a preference for a central bar at $x_{\text{bar}} = 3 $ and a penalty for overcrowding, is parameterized by $\theta \ge 0$:
$r_\theta(x, a, \rho) = -|x - x_{\text{bar}}| - \theta \log(\rho(x)+10^{-5}) - |a|/\vert \cX\vert $. Fig.~\ref{fig:beachbar_traj_2d} shows realizations of the mean field for two different values of $\theta$: while the agents try to be close to the bar, larger $\theta$ leads to lower density values because the agent spread more across states.

\begin{figure}[htbp]
    \centering
    \includegraphics[width=1\linewidth]{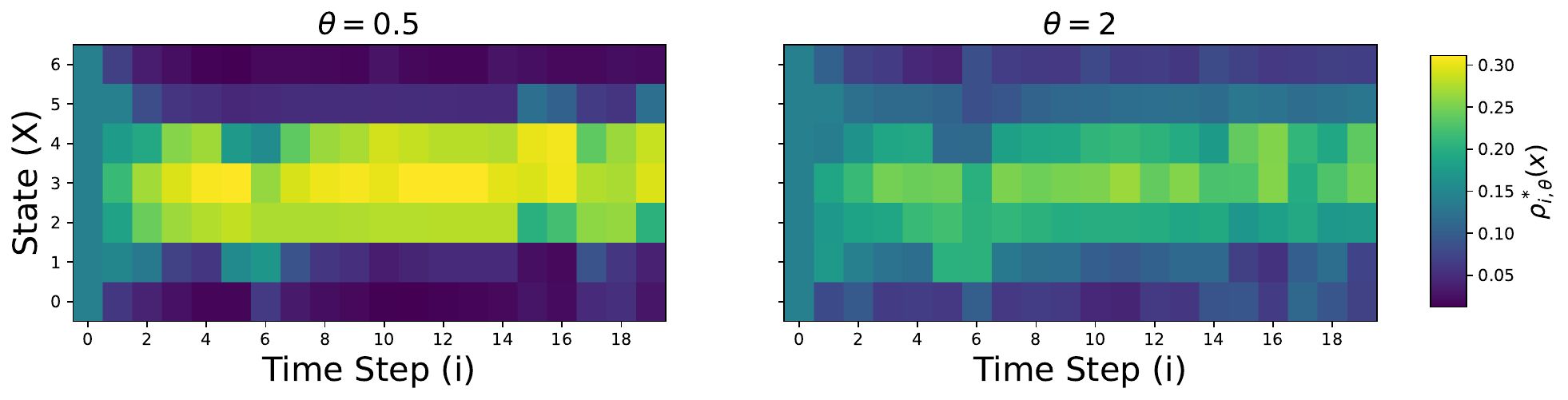}
    \caption{Mean Field Equilibrium Trajectories for two different values of $\theta$, with $\eta = 0.2$, on the same common noise realization.}
    \label{fig:beachbar_traj_2d}
\end{figure}

\subsection{Hyperparameters}

\begin{table}[htbp]
\centering
\caption{Hyperparameters for the Beach Bar Experiment}
\label{tab:beachbar_hyperparameters}
\begin{tabular}{lc}
\hline
\textbf{Hyperparameter} & \textbf{Value} \\ \hline
State/Action Space & $|\cX|=7$, $|\cA|=3$ \\
Horizon ($H$) & 20 \\
Common Noise Intensity ($\eta$) & $\{0.0, 0.1, 0.2, 0.3, 0.4, 0.5 \}$ \\
Prior $q(\theta)$ & GP: $\cN(\frac{5}{4}, 0.01)$ ; MAF: $\text{Unif}([0, \frac{5}{2}])$ \\
Test Parameters & $\{0.5, 0.875, 1.25, 1.625, 2\}$ \\
Observation Horizon ($|\cI|$) & $\{19, 5\}$ \\ \hline
Bayesian Fictitious Play Rounds ($K_{\text{bays}}$) & $15$ \\
Fictitious Play Inner Epochs & $8000$ \\
Distillation Epochs & $50000$ \\
Monte Carlo Sample Size ($M$) & $10000$ \\
Learning Rate (Fictitious Play / Distillation) & $5 \times 10^{-4}$ \\
Batch Size & $512$ \\ \hline
PCA Variance Threshold & $0.99$ \\
Normalizing Flow (MAF) Epochs & $25\,000 $\\
Normalizing Flow Learning Rate & $1 \times 10^{-4}$ \\
Normalizing Flow Ensemble Size & $5$ models \\
Normalizing Flow Architecture & $6$ layers, $128$ hidden units \\
Posterior Grid Resolution ($n_{\text{grid}}$) & $300$ \\ \hline
\end{tabular}
\end{table}

\subsection{Full sensitivity analysis}

\begin{figure}[htbp]
    \centering
    \includegraphics[width=0.8\linewidth]{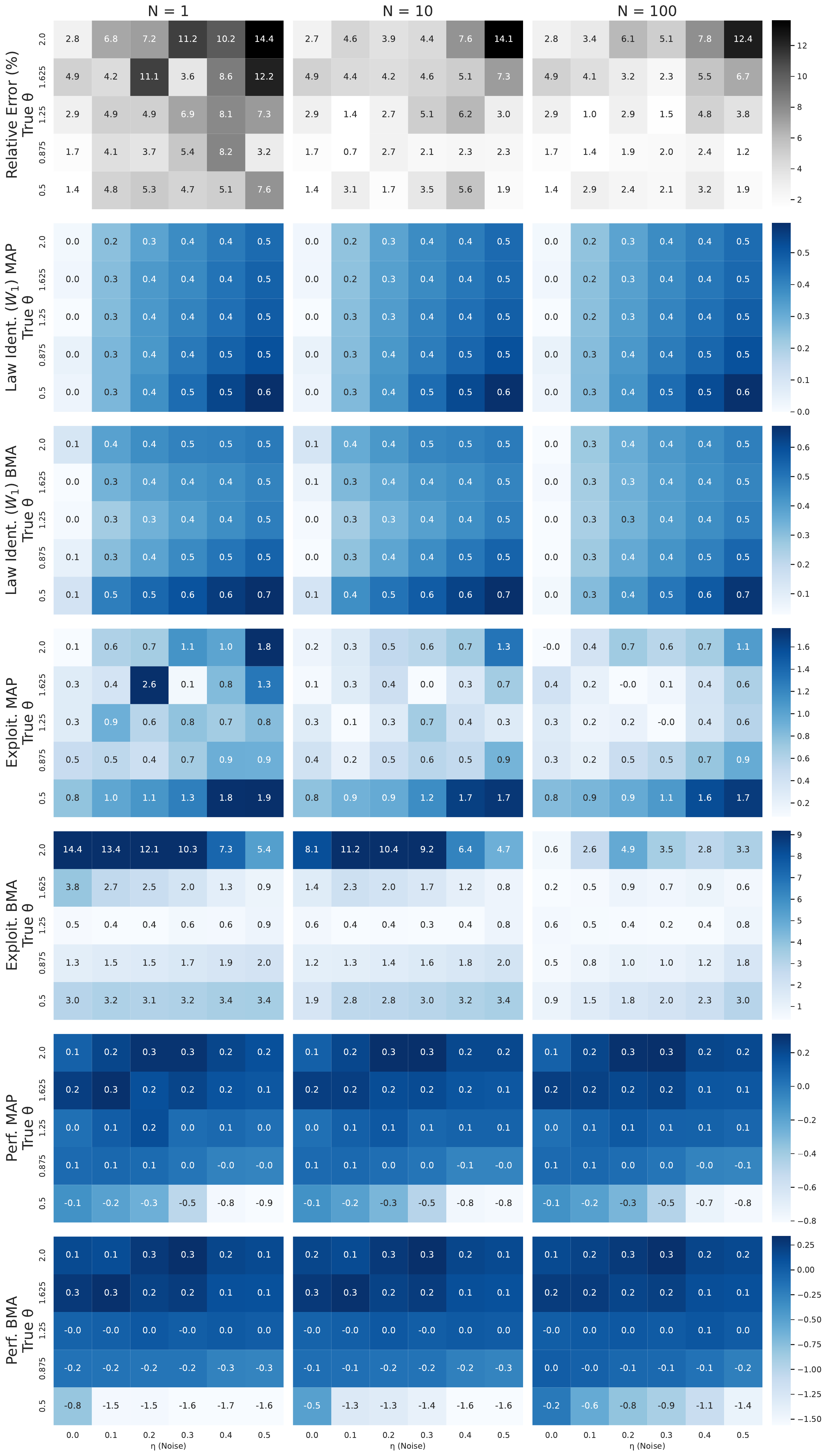}
    \caption{Sensitivity analysis, MAF, Averaged Metrics per $\theta_0$, $\xi = \rho$, $\cI = [H-1]$.}
\end{figure}

\begin{figure}
    \centering
    \includegraphics[width=0.78\linewidth]{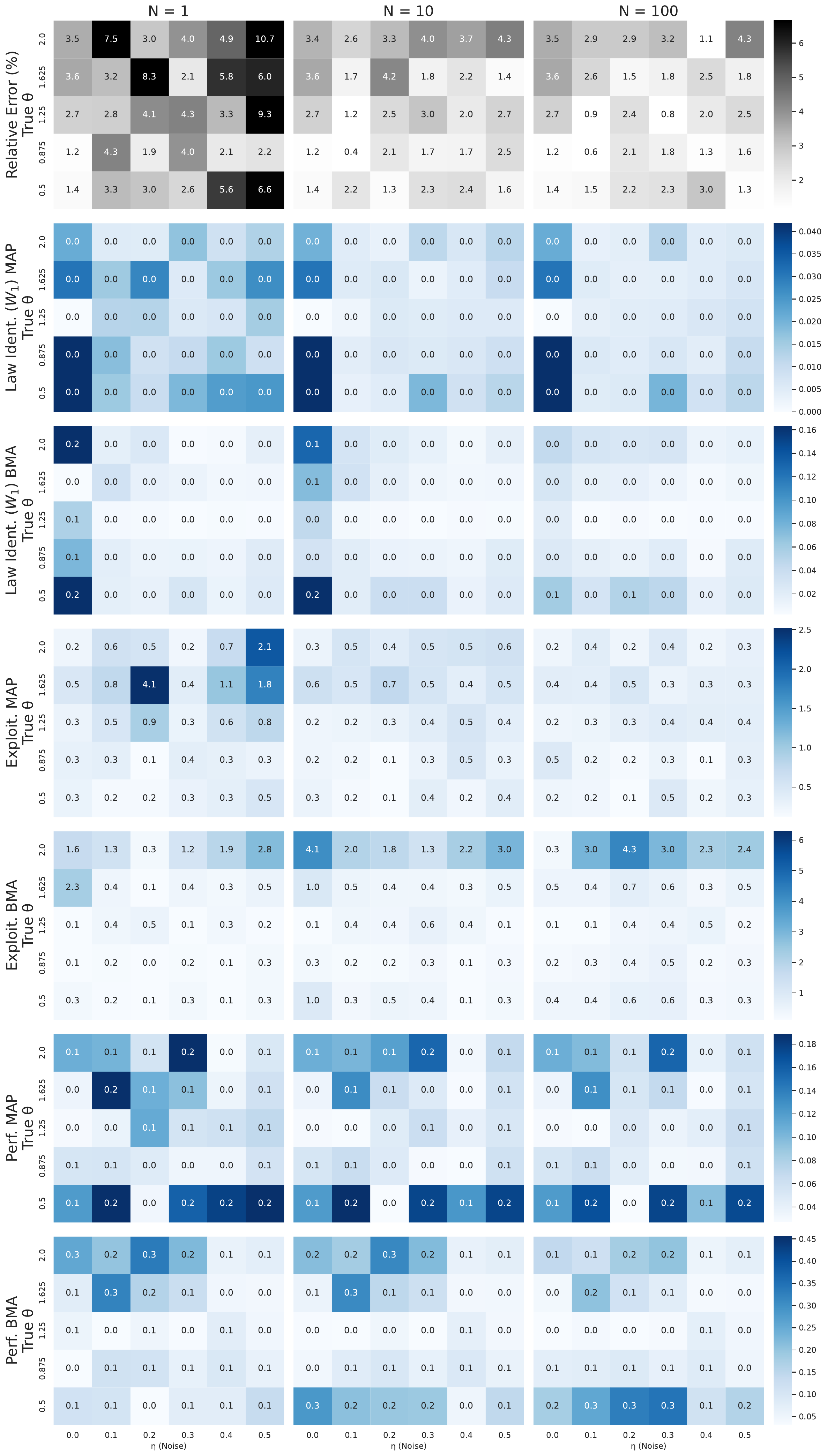}
    \caption{Sensitivity analysis, MAF, Standard Deviation Metrics per $\theta_0$, $\xi = \rho$, $\cI = [H-1]$.}
\end{figure}

\begin{figure}[htbp]
    \centering
    \includegraphics[width=0.8\linewidth]{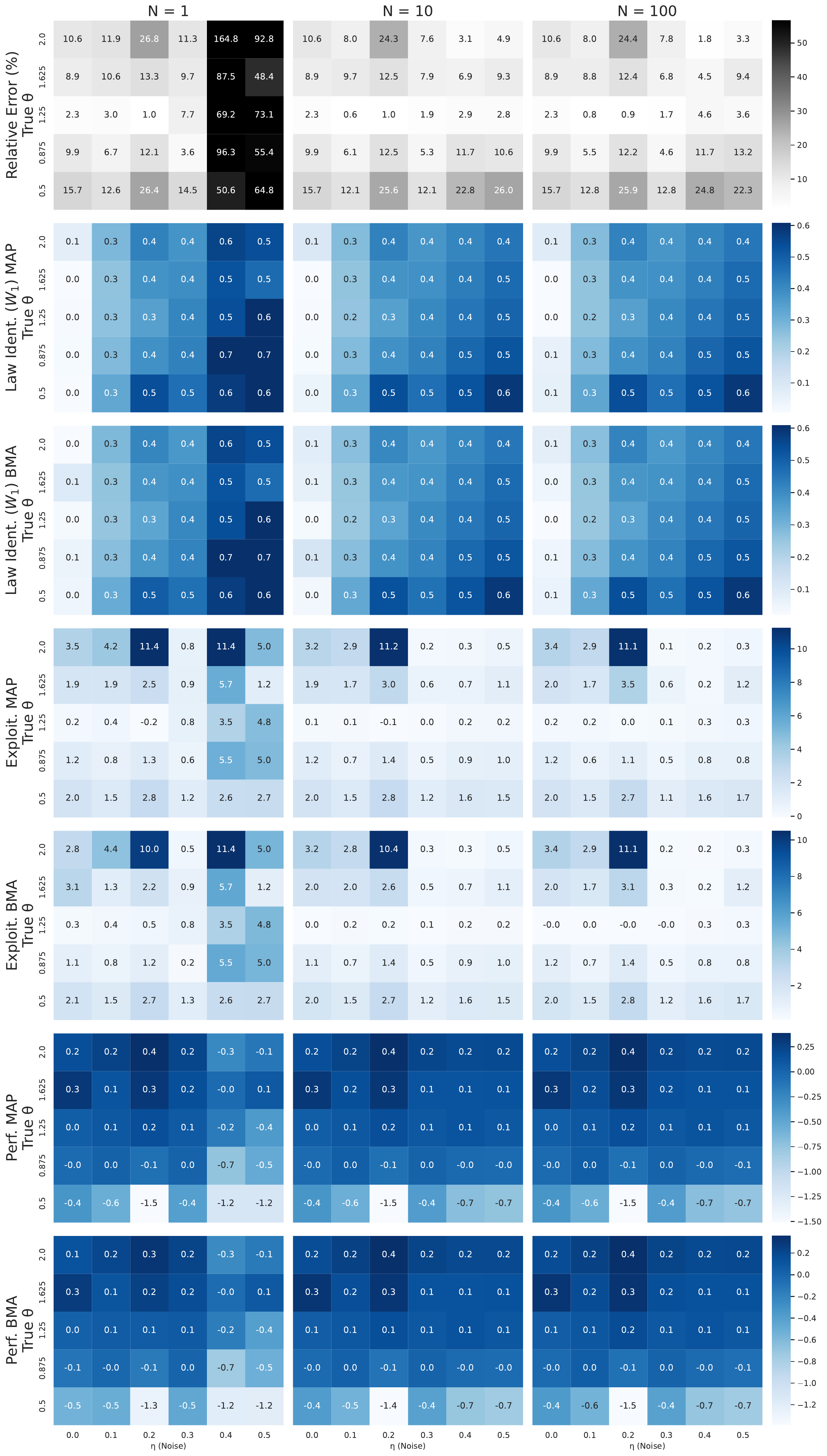}
    \caption{Sensitivity analysis, GP, Averaged Metrics per $\theta_0$, $\xi = \rho$, $\cI = [H-1]$.}
\end{figure}

\begin{figure}
    \centering
    \includegraphics[width=0.78\linewidth]{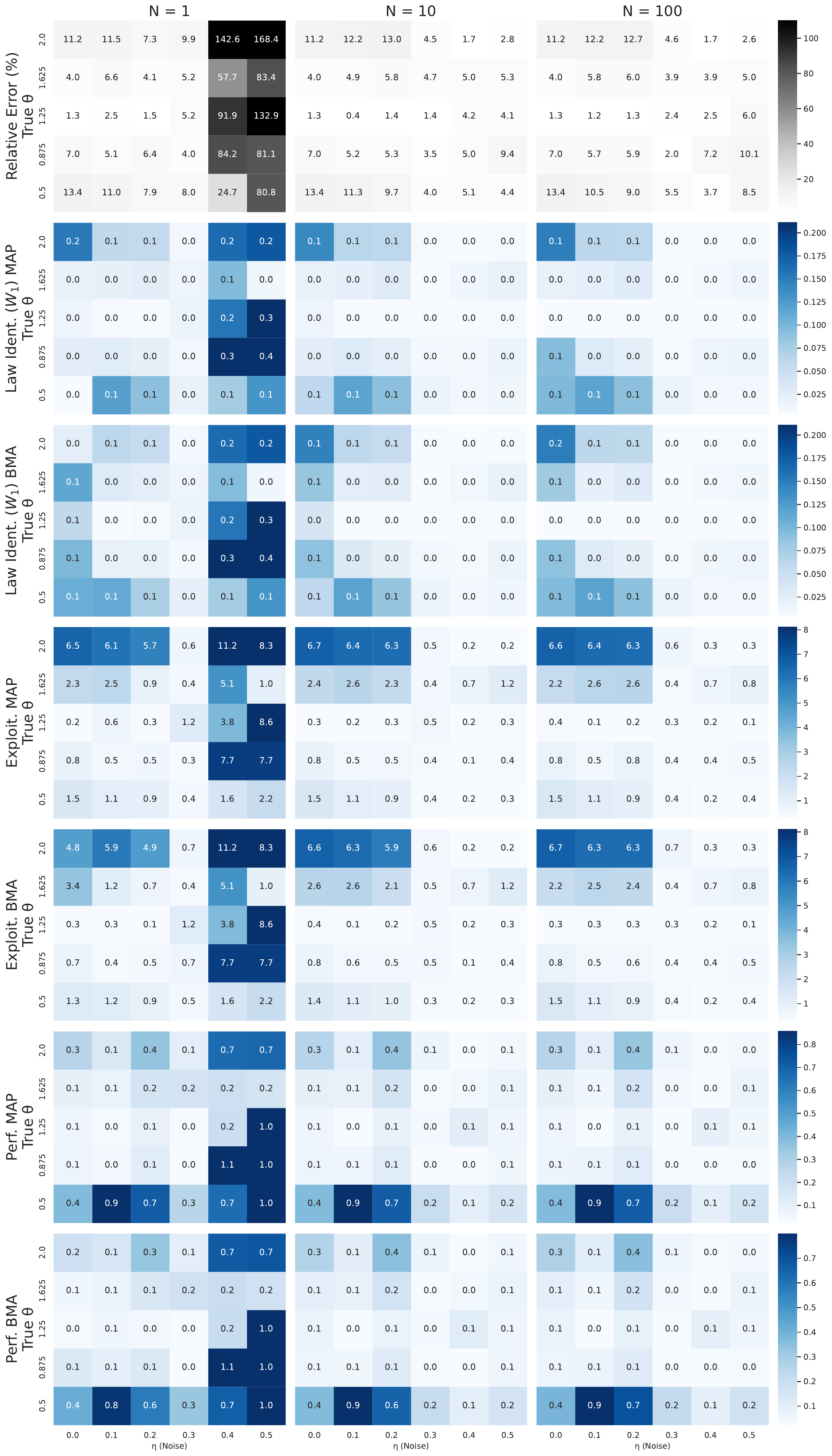}
    \caption{Sensitivity analysis, GP, Standard Deviation Metrics per $\theta_0$, $\xi = \rho$, $\cI = [H-1]$.}
\end{figure}

\begin{figure}[htbp]
    \centering
    \includegraphics[width=1\linewidth]{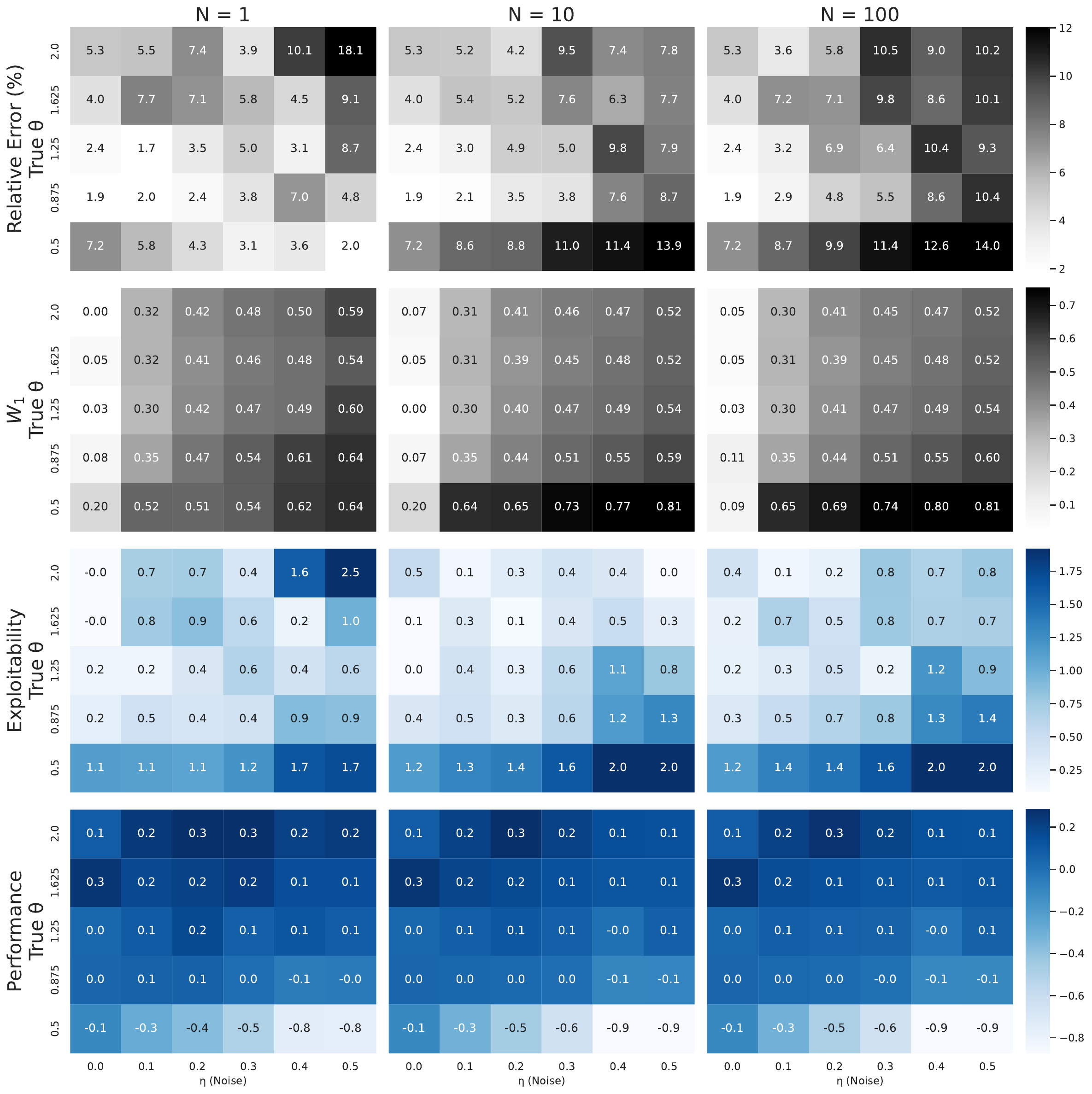}
    \caption{Sensitivity analysis, MF-IRL, Averaged Metrics per $\theta_0$, $\xi = \mu$, $\cI = [H-1]$.}
\end{figure}

\begin{figure}
    \centering
    \includegraphics[width=1\linewidth]{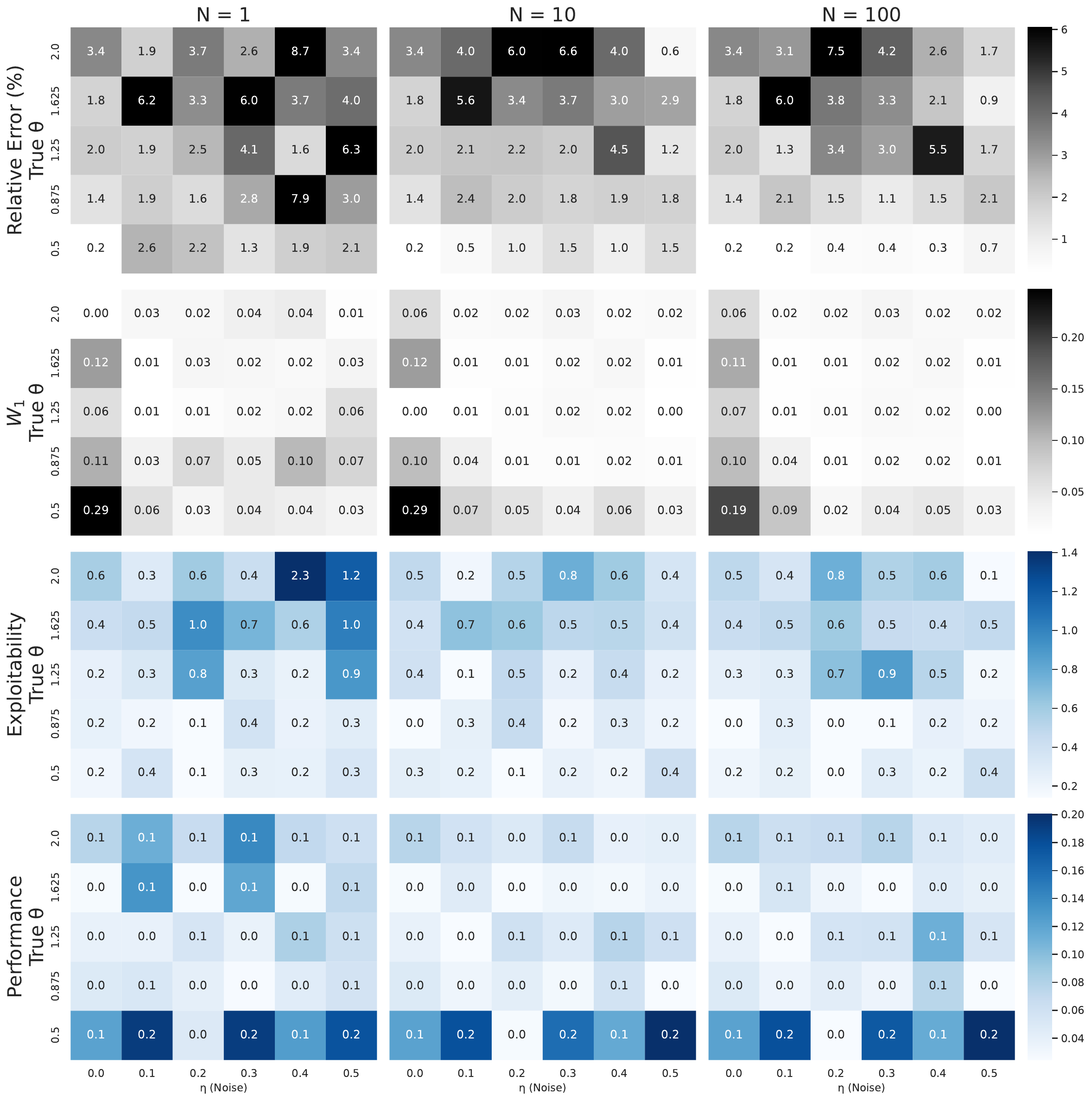}
    \caption{Sensitivity analysis,  MF-IRL, Standard Deviation Metrics per $\theta_0$, $\xi = \mu$, $\cI = [H-1]$.}
\end{figure}

\clearpage

\section{Experiment 2 (District): details}
\label{appendix:district}

\subsection{Model description}
\label{appendix:district-model}
Let $\cX = \{0, 1, 2\}^2$, $\cA = \{(0,0), (\pm 1, 0), (0, \pm 1)\}$, and $\cE^0 = \{0, 1\}^{|\cX|}$. Let for any $(x, a, \rho, e^0) \in \cX \times \cA \times \Delta_\cX \times \cE^0$, $ P(x, a, \rho, e^0) = \mathds{1}_{\{e^0(x) = 0\}}\delta_{x + a \bmod 3} + \mathds{1}_{\{e^0(x) = 1\}}\frac{1}{|\cA|} \sum_{s \in \cA} \delta_{x + s \bmod 3},$ where the sum $x + a$ and the reduction modulo $3$ are taken componentwise.
The common noise law is defined as $\varepsilon^0 \sim (\nu^0_\eta)^{\otimes |\cX|}$, where $\nu^0_\eta = \text{Bernoulli}(\eta)$ and $\eta \ge 0$ is the common noise intensity. The reward function, reflecting a preference for specific hotspots and a penalty for overcrowding, is: $r_{(\theta^e, \theta^c)}(x, a, \rho) = 100\theta^e(x) - 2\theta^c \log(\rho(x) + 10^{-5}) - |a|/5 $
for $\theta = (\theta^e, \theta^c) \in \bar \Theta = \RR_+^9 \times [0,5]$.

\subsection{Mean field illustration and ground truth ($\theta_0$)}

Fig.~\ref{fig:district_traj_eta=0.3} shows realizations of the mean field for three different values of $\theta$.

\begin{figure}[h]
    \centering
    \includegraphics[width=\linewidth]{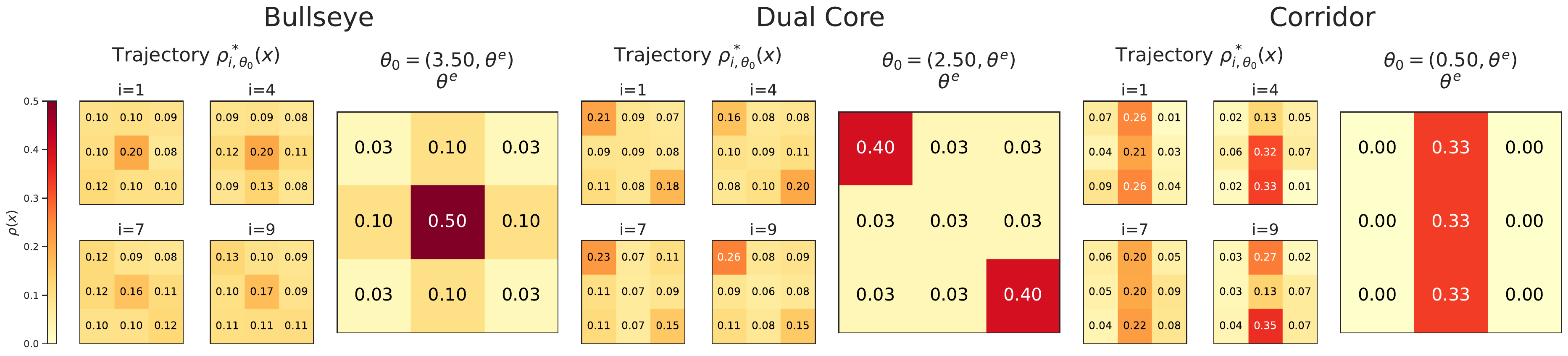}
    \caption{Mean Field Equilibrium Trajectories for three different scenarios of $\theta_0$, with $\eta = 0.3$.}
    \label{fig:district_traj_eta=0.3}
\end{figure}

\subsection{Hyperparameters}

\begin{table}[htbp]
\centering
\caption{Hyperparameters for the Entertainment District Experiment}
\label{tab:district_hyperparameters}
\begin{tabular}{lc}
\hline
\textbf{Hyperparameter} & \textbf{Value} \\ \hline
State/Action Space & $|\cX|=9$, $|\cA|=5$ \\
Horizon ($H$) & $10$ \\
Common Noise Intensity ($\eta$) & $\{0.3, 0.5 \}$ \\
Prior $q(\theta)$ (GP) & $\cN\big(((\frac{1}{9}, \dots,\frac{1}{9}), \frac{5}{2}), 0.01\big)$ \\
Prior $q(\theta)$ (MAF) & $\text{Dirichlet}(1, \dots, 1) \otimes \text{Unif}[0, 5]$ \\
Test Parameters & See Figure~\ref{fig:district_traj_eta=0.3} \\
Observation Horizon ($|\cI|$) & $\{9,1\}$ \\ \hline
Bayesian Fictitious Play Rounds ($K_{\text{bays}}$) & $20$ \\
Fictitious Play Inner Epochs & $8000$ \\
Distillation Epochs & $50000$ \\
Monte Carlo Sample Size ($M$) & $10000$ \\
Learning Rate (Fictitious Play / Distillation) & $5 \times 10^{-4}$ \\
Batch Size & $512$ \\ \hline
PCA Variance Threshold & $1$ \\
Normalizing Flow (MAF) Epochs & $25\,000 $\\
Normalizing Flow Learning Rate & $1 \times 10^{-4}$ \\
Normalizing Flow Ensemble Size & $5$ models \\
Normalizing Flow Architecture & $6$ layers, $128$ hidden units \\
\end{tabular}
\end{table}

\subsection{Full sensitivity analysis}

\begin{figure}[htbp]
    \centering
    \includegraphics[width=0.68\linewidth]{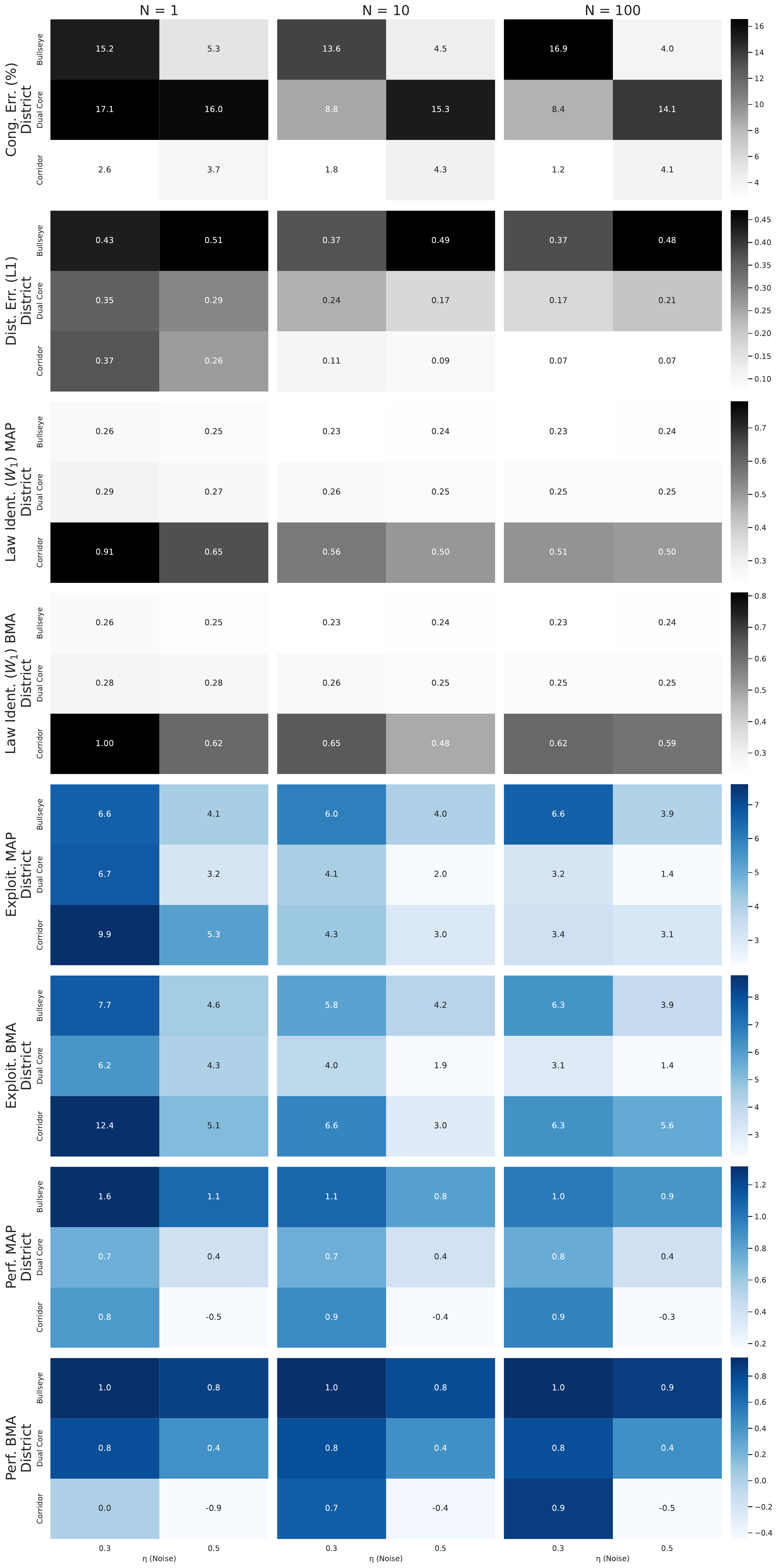}
    \caption{Sensitivity analysis, MAF, Averaged Metrics per $\theta_0$, $\xi = \rho$, $\cI = [H-1]$.}
\end{figure}

\begin{figure}
    \centering
    \includegraphics[width=0.68\linewidth]{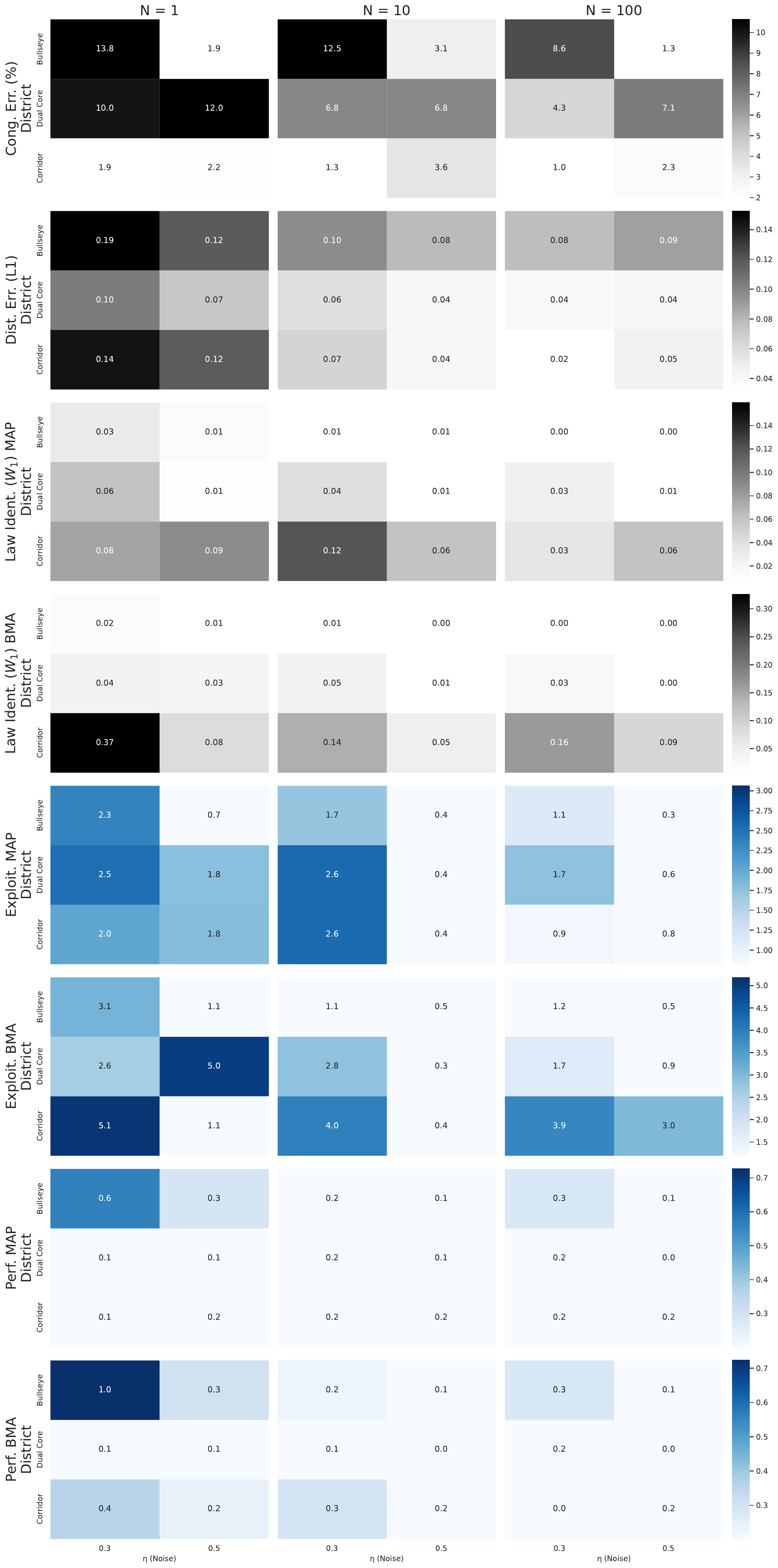}
    \caption{Sensitivity analysis, MAF, Standard Deviation Metrics per $\theta_0$, $\xi = \rho$, $\cI = [H-1]$.}
\end{figure}

\begin{figure}[htbp]
    \centering
    \includegraphics[width=0.68\linewidth]{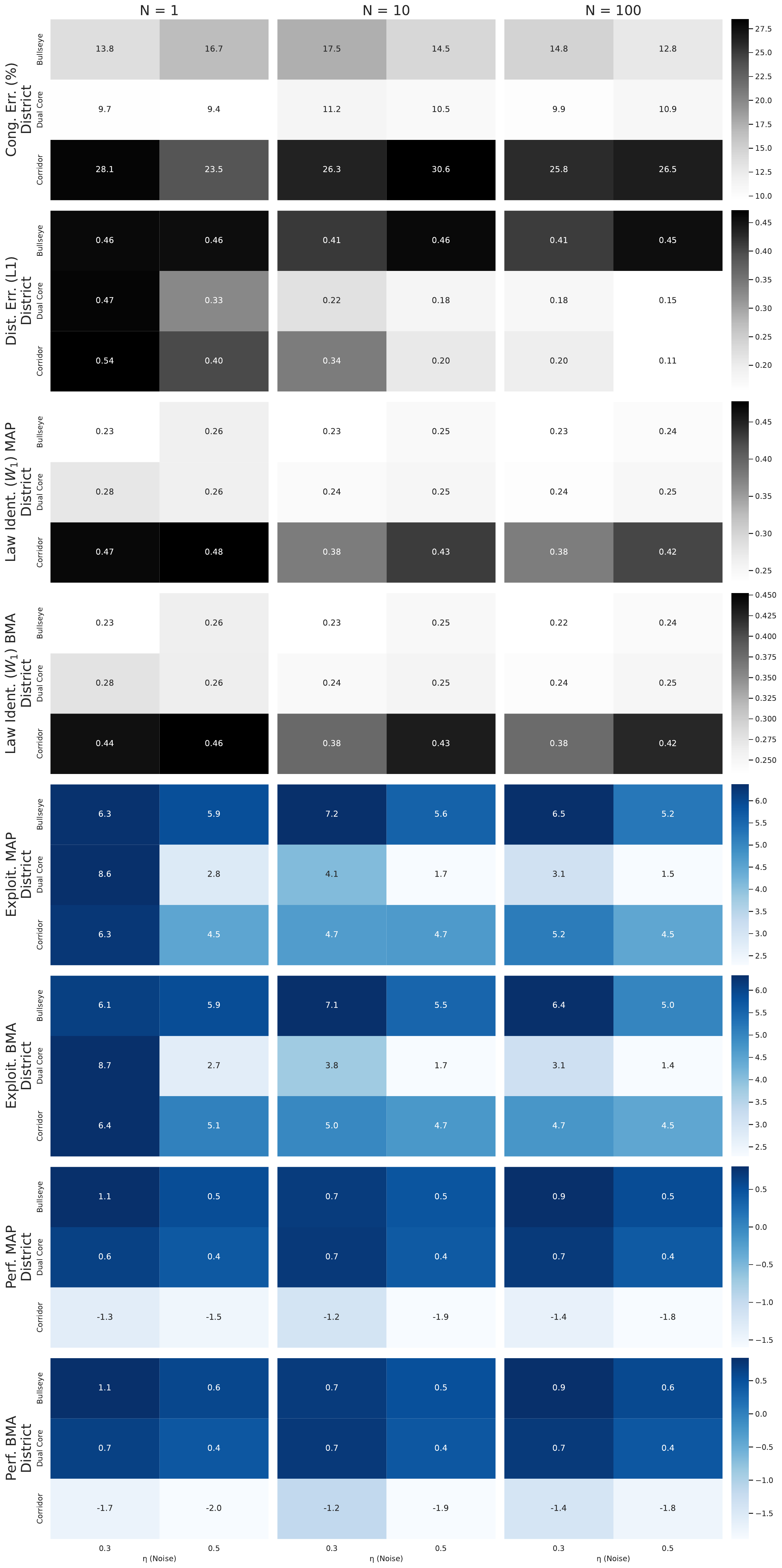}
    \caption{Sensitivity analysis, GP, Averaged Metrics per $\theta_0$, $\xi = \rho$, $\cI = [H-1]$.}
\end{figure}

\begin{figure}
    \centering
    \includegraphics[width=0.68\linewidth]{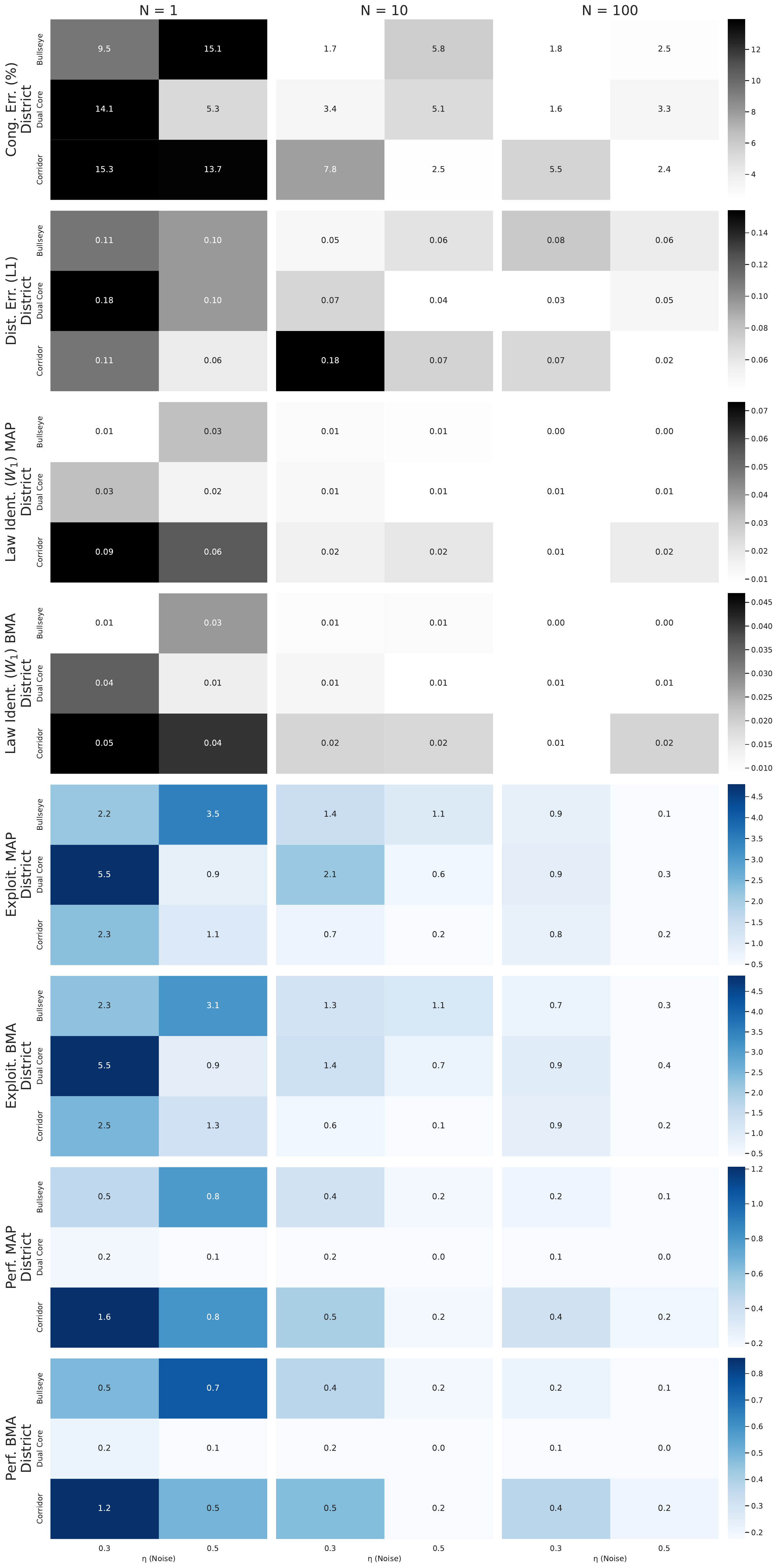}
    \caption{Sensitivity analysis, GP, Standard Deviation Metrics per $\theta_0$, $\xi = \rho$, $\cI = [H-1]$.}
\end{figure}

\begin{figure}[htbp]
    \centering
    \includegraphics[width=1\linewidth]{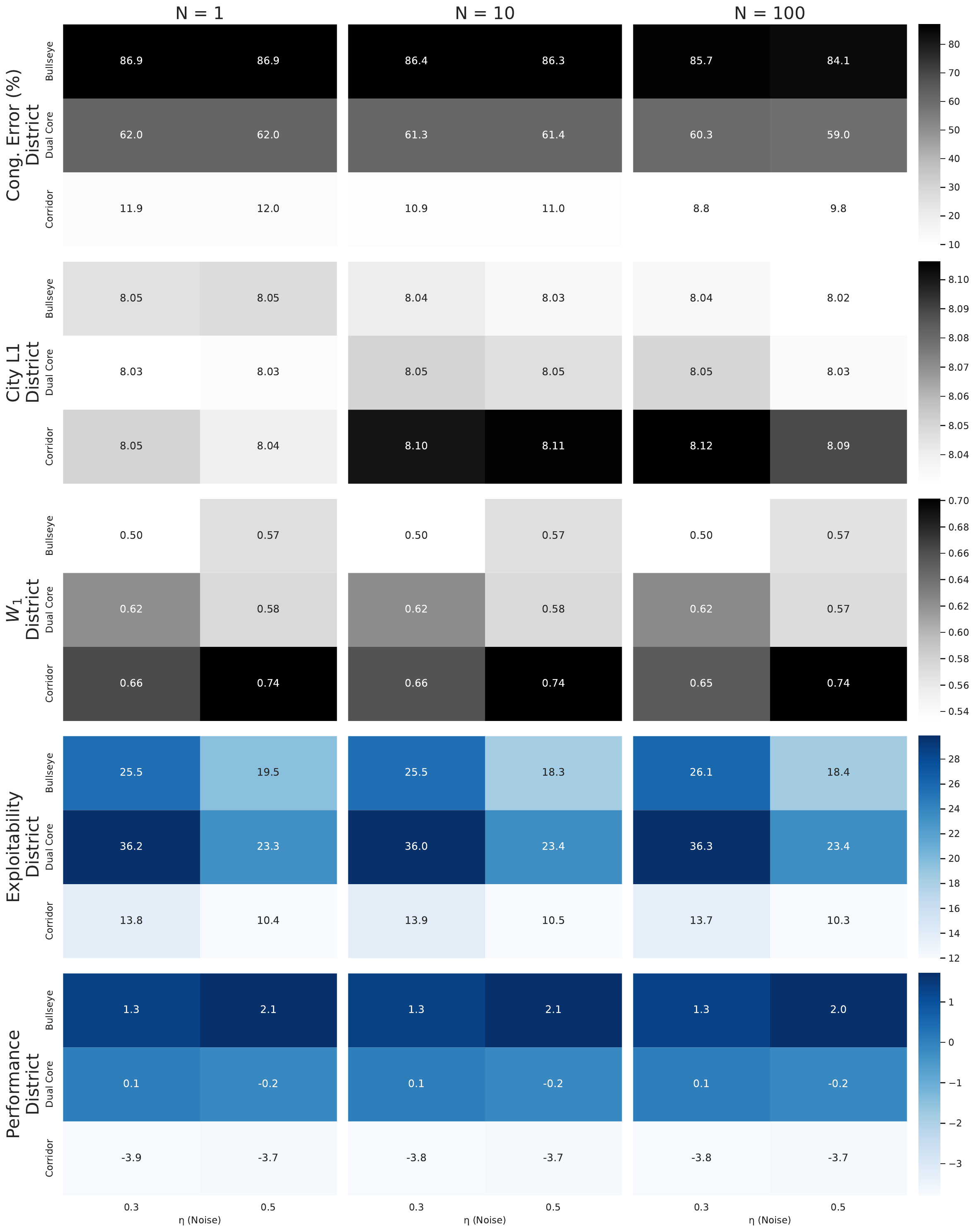}
    \caption{Sensitivity analysis, MF-IRL, Averaged Metrics per $\theta_0$, $\xi = \mu$, $\cI = [H-1]$.}
\end{figure}

\begin{figure}
    \centering
    \includegraphics[width=1\linewidth]{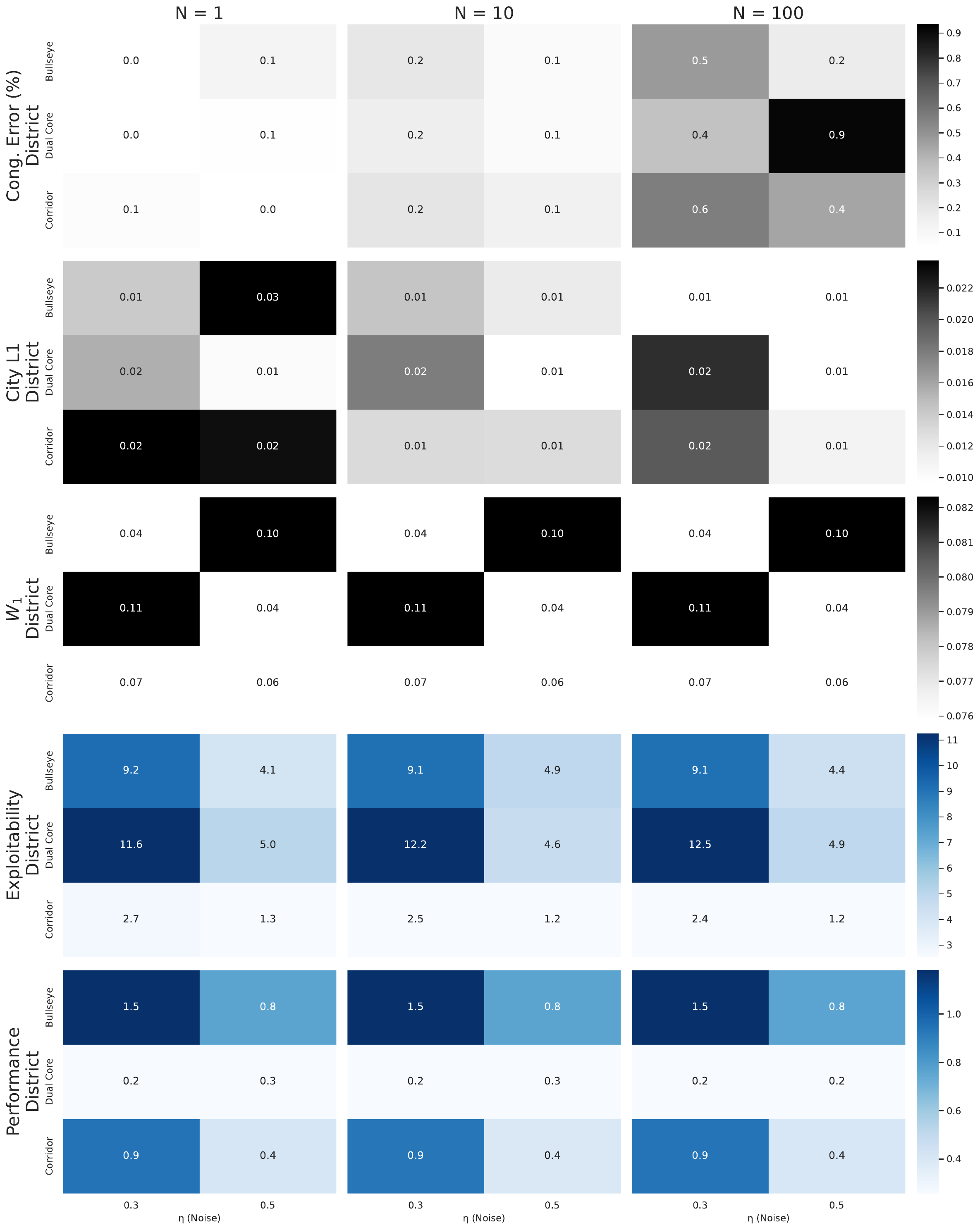}
    \caption{Sensitivity analysis,  MF-IRL, Standard Deviation Metrics per $\theta_0$, $\xi = \mu$, $\cI = [H-1]$.}
\end{figure}

\clearpage

\section{Experiment 3 (SIS): details}

\subsection{SIS epidemic model}
\label{sec:exp-sis}

To evaluate our method on parameter inference related to the dynamics, we consider a Susceptible--Infected--Susceptible (SIS) epidemic model inspired by~\cite{cui2021approximately}, where the infection probability of a susceptible agent is directly modulated by the fraction of infected agents in the population. Since the goal of this experiment is solely to demonstrate that our framework also applies to dynamics and to the state-action observation setting $\xi = \mu$, we do not report the full suite of metrics but focus on parameter recovery. 

\noindent\textbf{Model.}
The state space is $\cX = \{S, I\}$ (Susceptible and Infected), and the action space is $\cA = \{U, D\}$ (Go Out / Unprotected, and Social Distancing). 
Let $\rho = (\rho(S), \rho(I)) \in \Delta_\cX$ denote the current mean-field distribution.  The transition kernel is given by:
\begin{equation}
\label{eq:sis-transitions}
\begin{aligned}
    P_\theta(I \mid S, U, \rho, \varepsilon^0) &= \mathrm{clip}\!\big(\theta\, \rho(I) + \varepsilon^0 \cdot \delta, \; 0, \; 1\big), \\
    P_\theta(S \mid S, D, \rho, \varepsilon^0) &= 1, \\
    P_\theta(S \mid I, a, \rho, \varepsilon^0) &= \gamma, \qquad \forall\, a \in \cA,
\end{aligned}
\end{equation}
where $\theta \in (0, 1]$ is the \emph{infection transmission rate}, $\gamma \in (0, 1]$ is the \emph{recovery probability}, and $\delta > 0$ is a small perturbation strength for common noise events. The common noise $\varepsilon^0 \in \{0, 1\}$ models aggregate shocks (e.g., super-spreader events): when $\varepsilon^0$ is active, the infection probability receives an additive perturbation $\delta$, controlled by a Bernoulli intensity parameter $\eta$.  The key mean-field interaction lies in the term $\theta\, \rho(I)$: a susceptible agent who chooses to go out faces a risk that is directly proportional to the current fraction of infected individuals and to the intrinsic transmissibility $\theta$. Social distancing ($a = D$) provides full protection against infection.

The reward function captures the trade-off between the cost of being infected and the cost of preventive action:
\begin{equation}
\label{eq:sis-reward}
    r(x, a, \rho) = -\mathds{1}_{x = I} - c_D \cdot \mathds{1}_{a = D},
\end{equation}
where $c_D > 0$ is the cost of social distancing. The agent thus weighs the disutility of infection against the opportunity cost of distancing. At equilibrium, healthy agents in low-infection regimes ($\rho(I) \approx 0$) will predominantly go out, while high-infection regimes ($\rho(I) \gg 0$) incentivize distancing, creating a non-trivial fixed-point problem.

\noindent\textbf{Bayesian setup:} We consider $\bar \Theta = [0, 1]$, $\xi = \mu$ and $\cI = [H-1]$. We consider only the MAF inference. The observer considers the prior $q$ being $\text{Unif}([0.1, 0.8])$. We vary the ground truth $\theta_0$ across $\{0.2, 0.3, 0.4, 0.5, 0.6, 0.7\}$. We let vary the number of observations $N \in \{10, 100, 1000\}$ to study posterior concentration as a function of the sample size.

\noindent\textbf{Results:} 
The SIS model presents a qualitatively different identifiability challenge compared to the previous experiments. Because $\theta$ enters multiplicatively in the transition kernel, the sensitivity of the equilibrium trajectory $\rho^*_\theta$ to changes in $\theta$ depends on the equilibrium infection level $\rho^*_\theta(I)$ itself. In the low-infection regime ($\theta \lesssim 0.4$), small changes in $\theta$ produce clearly distinguishable trajectories and the parameter is well-identified. In the high-infection regime ($\theta \gtrsim 0.6$), the equilibrium infection level saturates due to the clipping and the self-limiting nature of the SIS dynamics, leading to diminished sensitivity and an identifiability ceiling. But overall, the averaged error on $\theta_0 \in \{0.2, 0.3, 0.4, 0.5, 0.6, 0.7\}$ remains inferior to $5\%$ (Figure~\ref{fig:sis_posterior}).

\begin{figure}[htbp]
    \centering
    \includegraphics[width=1\linewidth]{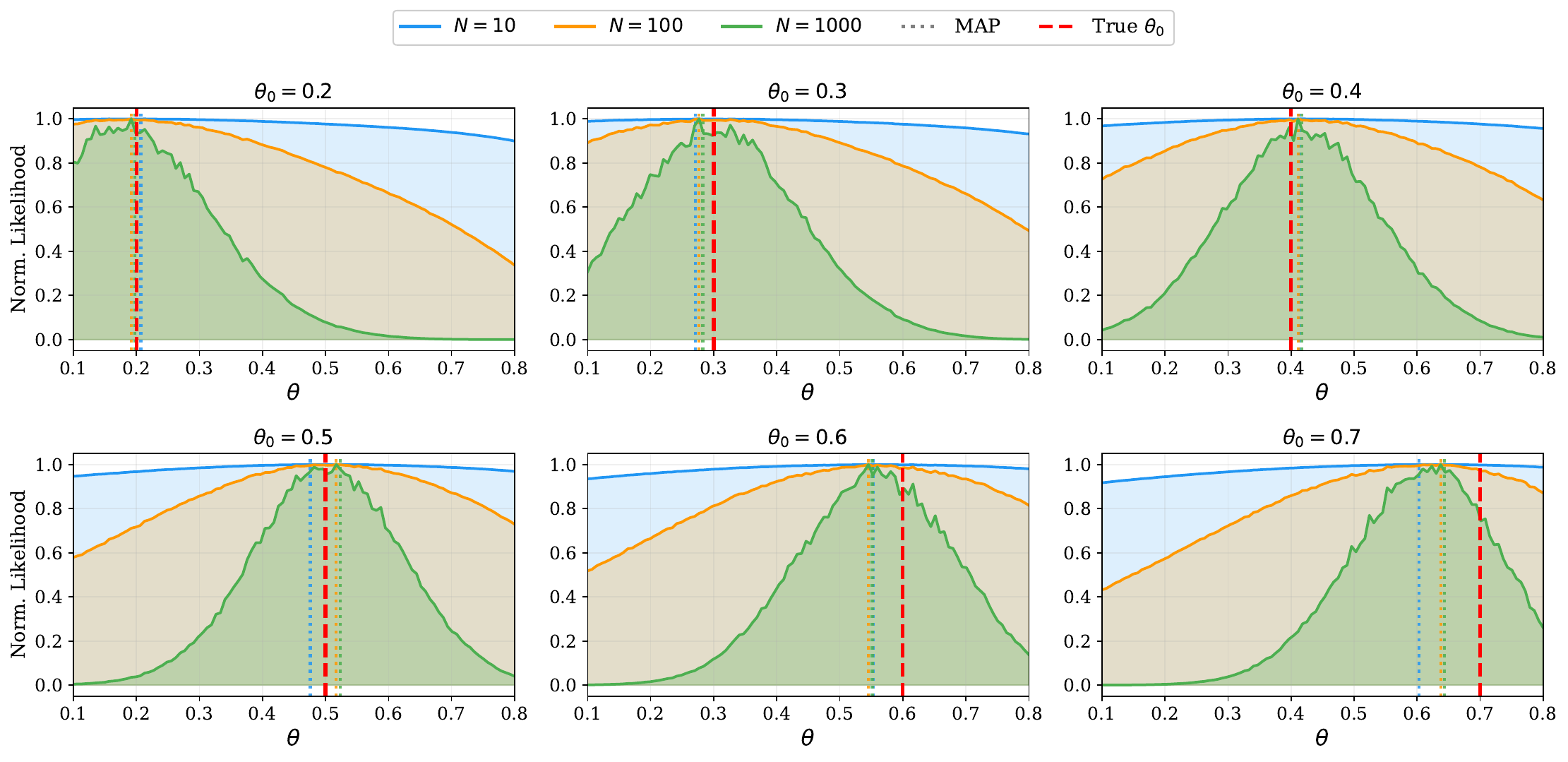}
    \caption{Posterior across $\theta_0$: $\xi = \mu$, $\eta = 0.1$, $\cI = [H-1]$.}
    \label{fig:sis_posterior}
\end{figure}

\subsection{Hyperparameters}

Hyperparameters are provided in Table~\ref{tab:sis_hyperparameters}.

\begin{table}[htbp]
\centering
\caption{Hyperparameters for the SIS Epidemic Experiment (Medium Configuration)}
\label{tab:sis_hyperparameters}
\begin{tabular}{lc}
\hline
\textbf{Hyperparameter} & \textbf{Value} \\ \hline
State/Action Space & $|\cX|=2$, $|\cA|=2$ \\
Horizon ($H$) & 10 \\
Fixed Parameters & $\gamma = 0.3$, $c_D = 0.5$ \\
Common Noise Intensity ($\eta$) & 0.1 \\
Common Noise Perturbation ($\delta$) & 0.05 \\
Prior $q(\theta)$ & MAF: $\text{Unif}([0.1, 0.8])$\\
Test Parameters & $\{0.2, 0.3, 0.4, 0.5, 0.6, 0.7\}$ \\
Observation Horizon ($|\cI|$) & 9 \\ \hline
Bayesian Fictitious Play Rounds ($K_{\text{bays}}$) & 20 \\
Fictitious Play Inner Epochs & 2000 \\
Distillation Epochs & 1000 \\
Monte Carlo Sample Size ($M$) & 1000 \\
Learning Rate (Fictitious Play / Distillation) & $5 \times 10^{-4}$ \\
Batch Size & 64 \\ \hline
PCA Variance Threshold & 0.999 \\
PCA Minimum Dimension ($\xi=\mu$) & 8 \\
Normalizing Flow (MAF) Epochs & 10\,000 \\
Normalizing Flow Learning Rate & $5 \times 10^{-4}$ \\
Normalizing Flow Ensemble Size & 10 models \\
Normalizing Flow Architecture & 6 layers, 128 hidden units \\
Posterior Grid Resolution ($n_{\text{grid}}$) & 100 \\ \hline
\end{tabular}
\end{table}

\section{Experiment 4 (real data): details}
\label{appendix:energy}

In this section, we show that the BE-MFG can find applicability in real-world settings. We propose a simple MFG model for electricity prices, driven by an exogenous demand and by the aggregate production of energy producers. We do not intend to study this model in depth, nor to claim genuine predictive value for electricity prices: such a study would require a more complex model and is beyond the scope of this paper.

\subsection{Model: a mean field game of electricity producers}

For simplicity, the model deals with normalized quantities, which are dimensionless.
An agent is a producer, whose state $x \in \cX$ is its normalized production level, on a grid of $32$ equally spaced values in $[x_{min}, x_{max}] = [1 - 6\sigma_D, 1 + 6\sigma_D]$, where $\sigma_D$ is the standard deviation of the normalized demand, estimated on the data below. We denote by $\Delta$ the grid step.
The time horizon is one day, $H = 24$ hours. At each hour, the producer chooses a ramp $a \in \{-2, \dots, 2\}$, and its production evolves as $x_{t+1} = \min(x_{max}, \max(x_{min}, x_t + a_t \Delta  + \epsilon_t \Delta ))$, where $\epsilon_t \in \{-1, 0, 1\}$ is an idiosyncratic noise with uniform law. The common noise is the normalized demand curve of the day, $(\tilde D_t)_{t}$, which is exogenous. It is carried as a \defi{shared state component} in the sense of Remark~\ref{remark:shared-state}: the demand $\tilde D_t$ is common to all the producers and observed by them, so that the mean field of the model is of the form $\rho_t = \delta_{\tilde D_t}\otimes \hat \rho_t$, where $\hat \rho_t \in \Delta_\cX$ is the distribution of the production levels. In particular, $\tilde D_t$ is recovered from $\rho_t$ by projection on the first marginal, and, since the parameter $\theta = (c, \beta)$ introduced below only enters the reward and not the transition kernel, this model satisfies the assumptions of Theorem~\ref{thm:consistency}.
The price is endogenous: given the aggregate production $\tilde Q_t = \sum_{x} x \rho_t(x)$ of the population, it is
\[
\tilde P_t = \max\big(0, \tilde D_t + \beta (\tilde D_t - \tilde Q_t)\big),
\]
so that the price increases when the demand exceeds the aggregate production. A producer sells its production at this price and is penalized when it produces more than the demand:
\[
r_\theta(x, \rho_t, \tilde D_t) = x \tilde P_t - c \, (x - \tilde D_t)_+ .
\]
The unknown parameter is $\theta = (c, \beta)$, where $c$ is the over-production penalty and $\beta$ is the sensitivity of the price to the gap between demand and production. The initial distribution $\rho_0$ is a discretized Gaussian centered at the production level $\tilde Q_0 = \tilde D_0 + (\tilde D_0 - \tilde P_0)/\beta$ that is consistent with the observed initial demand and price. The optimal production of a producer depends on the price, hence on the production of the whole population, and on the random demand.

\subsection{Data and protocol}

We use the public data of the ENTSO-E Transparency Platform for the French bidding zone: the hourly actual total load, used as the demand, and the hourly day-ahead price, in 2018 and 2019. Days with missing values or outliers are removed. The standard deviation of the normalized demand, used to set the production grid $[1 - 6\sigma_D, 1 + 6\sigma_D]$ of the model, is estimated on the 2018 data, which gives $\sigma_D = 0.089$. The three periods have distinct roles:
\begin{itemize}
    \item \textbf{2018:} the daily demand curves define the common noise of the model, used to compute the BE-NE and to train the MAFs;
    \item \textbf{January to November 2019} ($311$ days): the observations used to infer $\theta$, each day being normalized by its own daily mean;
    \item \textbf{December 2019} ($22$ days): held-out days, on which we compare the predicted and observed prices. To avoid any leakage, these days are normalized by the mean demand and price of November 2019.
\end{itemize}
The observation of a day is the pair of normalized curves $\xi = (\tilde D_t, \tilde P_t)_{t \in [H-1]}$: the observer does not see the distribution of the producers, but only an aggregate quantity (the price) together with the common noise.

\noindent\textbf{BE-MFG.} We use the Gaussian prior $c \sim \cN(2.75, 1.15^2)$, $\beta \sim \cN(0.85, 0.25^2)$, and compute the posterior on $[0.5, 5] \times [0.2, 1.5]$. The BE-NE is computed as in the other experiments, and the conditional density $\Gamma_{\theta, \cI}$ is learned by an ensemble of MAFs directly on the $48$-dimensional observation. The posterior is computed on a $60 \times 40$ grid of parameters, and we use the BMA policy with $64$ posterior samples.

\noindent\textbf{Prediction.} For each day of December and each hour $t$, the mean field is initialized from the observed price $\tilde P_t$, as for $\rho_0$, and propagated for one step with the equilibrium policy, which gives the predicted price $\tilde P_{t+1}$. We report the mean squared error (MSE) between the predicted and observed prices over all hours of December, and the Wasserstein-1 distance $\mathcal{W}_1$ between the empirical distributions of the predicted and observed prices.

\noindent\textbf{Baselines.} (i) MF-IRL (App.~\ref{appendix:mf-irl}), inferred on the same data; since it returns a point estimate, we predict with a dedicated PMFG solved at the estimated parameter. (ii) An autoregressive model (AR): for each hour $t$, $\tilde P_{t+1} = a_t \tilde P_t + b_t \tilde D_t + d_t \tilde D_{t+1} + e_t$, fitted by least squares on January to November 2019.

\subsection{Results}

Table~\ref{tab:energy_results} and Figure~\ref{fig:energy_predictions} report the results. The BE-MFG reproduces the daily shape of the prices of December, including the morning and evening peaks, from a parameter inferred on the previous months. MF-IRL, which is not designed for common noise, returns a price that is almost insensitive to the demand, and fails. The AR regression, a statistical one-step predictor with $4 \times 23$ free coefficients, is slightly more accurate than the BE-MFG, which only has two parameters and explains the prices through an equilibrium mechanism. Overall, this experiment shows that the BE-MFG pipeline can be applied to real data.

\begin{table}[htbp]
\centering
\small
\vspace{8pt}
\caption{Real data: one-step prediction of the prices of December 2019, with $\theta$ inferred on January to November 2019. BE-MFG: mean $\pm$ std over five seeds.}
\label{tab:energy_results}
\begin{tabular}{@{}lcc@{}}
\toprule
Method & MSE & $\mathcal{W}_1$ \\
\midrule
BE-MFG (MAF, BMA) & $0.0116 \pm 0.0006$ & $0.071 \pm 0.004$ \\
MF-IRL & $0.0616$ & $0.191$ \\
AR (fitted on January to November 2019) & $0.0098$ & $0.049$ \\
\bottomrule
\end{tabular}
\end{table}

\begin{figure}[htbp]
    \centering
    \vspace{12pt}
    \includegraphics[width=\linewidth]{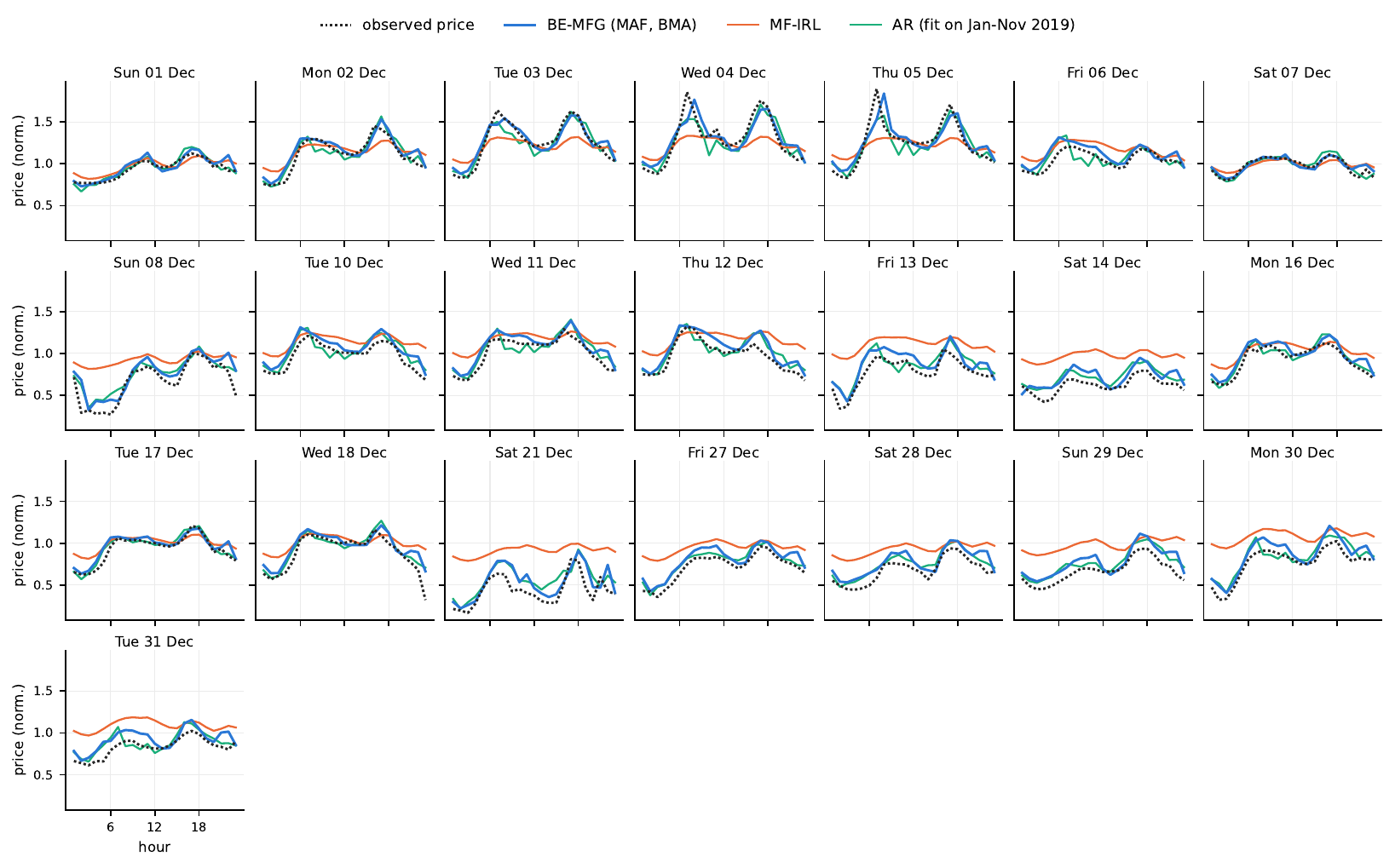}
    \caption{Real data: observed prices of December 2019 (dotted) and one-step predictions of the BE-MFG (BMA), MF-IRL and the AR model, one panel per day. Prices are normalized by the mean price of November 2019.}
    \label{fig:energy_predictions}
\end{figure}

\subsection{Hyperparameters}

\begin{table}[htbp]
\centering
\small
\caption{Hyperparameters for the real-data experiment.}
\label{tab:energy_hyperparameters}
\begin{tabular}{@{}lc@{}}
\toprule
\textbf{Hyperparameter} & \textbf{Value} \\
\midrule
State/Action Space & $|\cX| = 32$, ramps $\{-2, \dots, 2\}$ \\
Horizon ($H$) & $24$ \\
Prior $q(\theta)$ & $\cN(2.75, 1.15^2) \otimes \cN(0.85, 0.25^2)$ \\
Bayesian Fictitious Play Rounds ($K_{\text{bays}}$) & $20$ \\
Fictitious Play Inner Epochs & $8000$ \\
Distillation Epochs & $8000$ \\
Learning Rate (Fictitious Play / Distillation) & $5 \times 10^{-4}$ \\
Batch Size & $64$ \\
\midrule
Normalizing Flow (MAF) Epochs & $8000$ \\
Normalizing Flow Learning Rate & $10^{-3}$ \\
Normalizing Flow Ensemble Size & $5$ models \\
Normalizing Flow Architecture & $5$ layers, $128$ hidden units \\
Posterior Grid Resolution & $60 \times 40$ \\
BMA Posterior Samples & $64$ \\
\midrule
Dedicated PMFG (MF-IRL prediction) & $8$ FP rounds $\times$ $2000$ epochs \\
\bottomrule
\end{tabular}
\end{table}

\section{Grid search}
\label{appendix:grid-search}

A natural approach to the inverse problem is a grid search: solve a dedicated PMFG at each node of a grid over $\bar\Theta$, and select the node whose equilibrium best explains the observations. We do not use it as a baseline in the main paper, because it is not competitive: (i) it requires one NE computation per node, and computing a single NE of an MFG is already expensive; (ii) its accuracy is limited by the resolution of the grid; (iii) its number of nodes grows exponentially with the dimension $d$ of the parameter. In this appendix, we illustrate these three points numerically, on Beach Bar ($d=1$), which is the most favorable case for a grid search, and on District ($d=10$).

\paragraph{Protocol.} For a given $\theta$, the quality of a policy $\pi$ is measured by its relative exploitability $\cE_\theta(\pi)/\cE_\theta(\pi_{\mathrm{unif}})$, where $\pi_{\mathrm{unif}}$ is the uniform policy. For the BE-MFG, we consider the policy induced at $\theta$ (Theorem~\ref{thm:NE_BE-MFG_induces_NE_PMFG}), and its exploitability is computed with the same (non-Bayesian) best-response network as for the dedicated PMFG, so that both methods are exploited over the same class of policies. The dedicated PMFG and the BE-MFG policies are neural networks of comparable sizes (Beach Bar: $\sim 10^4$ parameters; District: $\sim 10^5$ parameters), trained with the same number of gradient steps per FP round ($8{,}000$) and the same batch size ($512$). We use $\eta = 0.3$ and five seeds. The reported wall-clock is the cumulative training time of the best responses, for both methods; it excludes the distillation and the exploitability evaluations, and the exploitability is computed for the FP average policy. 

\paragraph{Computational cost.} Table~\ref{tab:grid-beachbar} shows that, at each $\theta$, a dedicated PMFG reaches $5\%$ relative exploitability in a time comparable to the single BE-MFG solve. A grid search pays this cost at every node: a grid of $8$ nodes costs $3{,}380$ seconds, whereas a single BE-MFG solve reaches $5\%$ at all the nodes in $1{,}082 \pm 590$ seconds. On District (Table~\ref{tab:grid-district}), a grid with only three points per axis would require $3^{10} = 59{,}049$ dedicated solves, each costing between several minutes and two hours, whereas a single BE-MFG solve reaches $10\%$ on the three scenarios in a time comparable to one dedicated solve.

\begin{table}[htbp]
\centering
\small
\caption{Beach Bar: FP rounds and best-response wall-clock (seconds) to reach $5\%$ relative exploitability (mean $\pm$ std over five seeds). The last row is the cost of the whole grid: the sum over the $8$ nodes for the dedicated PMFGs, and the time for the single BE-MFG solve to reach $5\%$ at all the nodes.}
\label{tab:grid-beachbar}
\begin{tabular}{@{}lcccc@{}}
\toprule
 & \multicolumn{2}{c}{FP rounds} & \multicolumn{2}{c}{Wall-clock (s)} \\
\cmidrule(lr){2-3} \cmidrule(lr){4-5}
$\theta$ & PMFG & BE-MFG & PMFG & BE-MFG \\
\midrule
$0.5$ & $8 \pm 1$ & $11 \pm 5$ & $645 \pm 64$ & $1082 \pm 590$ \\
$0.75$ & $7 \pm 0$ & $8 \pm 3$ & $522 \pm 36$ & $747 \pm 306$ \\
$1$ & $5 \pm 1$ & $5 \pm 0$ & $380 \pm 50$ & $439 \pm 37$ \\
$1.25$ & $4 \pm 1$ & $4 \pm 0$ & $289 \pm 58$ & $330 \pm 44$ \\
$1.5$ & $4 \pm 1$ & $4 \pm 1$ & $302 \pm 49$ & $330 \pm 71$ \\
$1.75$ & $5 \pm 1$ & $4 \pm 0$ & $414 \pm 107$ & $403 \pm 45$ \\
$2$ & $6 \pm 0$ & $5 \pm 0$ & $476 \pm 33$ & $421 \pm 45$ \\
$2.25$ & $5 \pm 1$ & $5 \pm 0$ & $352 \pm 94$ & $458 \pm 1$ \\
\midrule
Grid of $8$ nodes & -- & -- & $3{,}380$ & $1{,}082 \pm 590$ \\
\bottomrule
\end{tabular}
\end{table}

\begin{table}[htbp]
\centering
\small
\caption{District ($d=10$): FP rounds and best-response wall-clock (seconds) to reach $10\%$ relative exploitability, and relative exploitability at the end of the FP budget (mean $\pm$ std over five seeds).}
\label{tab:grid-district}
\begin{tabular}{@{}llccc@{}}
\toprule
 & Method & Bullseye & Dual-core & Corridor \\
\midrule
FP rounds & PMFG & $46 \pm 3$ & $34 \pm 8$ & $6 \pm 0$ \\
 & BE-MFG & $33 \pm 4$ & $29 \pm 2$ & $29 \pm 6$ \\
\midrule
Wall-clock (s) & PMFG & $7351 \pm 778$ & $4577 \pm 1816$ & $439 \pm 36$ \\
 & BE-MFG & $6019 \pm 1683$ & $4506 \pm 450$ & $4765 \pm 1723$ \\
\midrule
Rel. exploitability & PMFG & $8.8 \pm 0.7\%$ & $7.6 \pm 0.6\%$ & $1.7 \pm 0.6\%$ \\
 & BE-MFG & $6.0 \pm 0.9\%$ & $6.0 \pm 0.6\%$ & $10.7 \pm 2.6\%$ \\
\bottomrule
\end{tabular}
\end{table}

\paragraph{Accuracy.} We compare parameter recovery on Beach Bar with a grid search over the $5$ dedicated PMFGs of Experiment~1, at $\theta \in \{0.5, 0.875, 1.25, 1.625, 2.0\}$. We sample $100$ true parameters $\theta_0$ uniformly in $[0.5, 2.0]$. For each $\theta_0$, a dedicated NE, computed with the same settings as the nodes of the grid, generates $N=100$ state observations over the full trajectory. The grid search selects the node minimizing the Wasserstein distance $\mathcal{W}_1$ between the observations and trajectories simulated at the node, while the BE-MFG with MAF returns the MAP over a grid of $300$ values of $\theta$, which requires no additional NE computation. We also report the $\mathcal{W}_1$ distance between the law of the observations and the law reconstructed with the recovered parameter, computed on independent draws. Table~\ref{tab:grid-recovery} shows that the BE-MFG recovers $\theta_0$ about $1.7$ times more accurately, while both methods reproduce the law of the observations equally well. 

\begin{table}[htbp]
\centering
\small
\caption{Beach Bar: parameter recovery from $N=100$ observations, for $100$ parameters $\theta_0$ drawn uniformly in $[0.5, 2.0]$ (mean $\pm$ std over five seeds and the $100$ parameters).}
\label{tab:grid-recovery}
\begin{tabular}{@{}lccc@{}}
\toprule
Method & $\vert\hat\theta - \theta_0\vert$ & $\mathcal{W}_1$ in law & Number of NE solves \\
\midrule
Grid search ($5$ nodes) & $0.124 \pm 0.078$ & $0.406 \pm 0.026$ & $5$ \\
BE-MFG (MAF, MAP) & $0.075 \pm 0.053$ & $0.409 \pm 0.023$ & $1$ \\
\bottomrule
\end{tabular}
\end{table}

\end{document}